\documentclass[10pt]{iopart}
\usepackage[latin9]{inputenc}
\usepackage{bm}
\usepackage{amstext}
\usepackage{amssymb}
\usepackage{graphicx}
\usepackage{esint}
\makeatletter

\DeclareTextSymbolDefault{\textquotedbl}{T1}
\providecommand{\tabularnewline}{\\}

\usepackage{iopams}
\usepackage{setstack}

\usepackage{lineno}
\usepackage{iopams}
\modulolinenumbers[5]
\eqnobysec
\usepackage{xcolor}
\usepackage{multirow}
\providecommand{\tabularnewline}{\\}
\usepackage{amsfonts}
\usepackage{amsthm}
\usepackage{mathrsfs}
\usepackage{subfigure}
\usepackage{soul}
\usepackage{etoolbox}

\DeclareMathAlphabet\EuScript{U}{eus}{m}{n}
\SetMathAlphabet\EuScript{bold}{U}{eus}{b}{n}

\renewcommand{\S}{\mathcal{S}}
\newcommand{\E}{\mathcal{E}}
\newcommand{\F}{\mathcal{F}}

\newcommand{\Fp}{\mathcal{F}_p}
\newcommand{\Fn}{\mathcal{F}_{\mbox{\tiny N}}}
\newcommand{\Fc}{\mathcal{F}_{\mbox{\tiny C}}}
\newcommand{\Fa}{\mathcal{F}_{\mbox{\tiny A}}}
\newcommand{\Fq}{\mathcal{F}_{\mbox{\tiny Q}}}
\newcommand{\Fmin}{\mathcal{F}_{\rm min}}
\newcommand{\Fmax}{\mathcal{F}_{2}}

\newcommand{\Fave}{\mathcal{F}_{\rm ave}}
\newcommand{\Fuj}{\mathcal{F}_{1}}
\newcommand{\Fam}{\mathcal{F}_{\mbox{\tiny AM}}}
\newcommand{\Fgm}{\mathcal{F}_{\mbox{\tiny GM}}}
\newcommand{\Fhm}{\mathcal{F}_{\mbox{\tiny HM}}}
\newcommand{\Ns}{N_{\rm samples}}
\newcommand{\Ftele}{\F_{\rm tele}}

\newcommand{\Id}{\mathbf{1}}
\newcommand{\ii}{{\rm i}}

\def\bra#1{\mathinner{\langle{#1}|}}
\def\ket#1{\mathinner{|{#1}\rangle}}
\def\defeq{\mathrel{\mathop:}=}

\newcommand{\proj}[1]{\left| #1 \right\rangle\!\!\left\langle #1 \right|}

\newtheorem{theorem}{Theorem}\newtheorem{corollary}{Corollary}[theorem]
\newtheorem{lemma}[theorem]{Lemma}

\def\@mkboth#1#2{}
\newlength\appendixwidth
\preto\appendix{\addtocontents{toc}{\protect\patchl@section}}
\newcommand{\patchl@section}{%
  \settowidth{\appendixwidth}{\textbf{Appendix }}%
  \addtolength{\appendixwidth}{1.5em}%
  \patchcmd{\l@section}{1.5em}{\appendixwidth}{}{\ddt}%
}

\makeatother

\begin{document}

\title{Quantum fidelity measures for mixed states}

\author{Yeong-Cherng Liang$^{1}$, Yu-Hao Yeh$^{1}$, Paulo E. M. F. Mendonça$^{2,3}$,
Run Yan Teh$^{4}$, Margaret D. Reid$^{4,5}$ and Peter D. Drummond$^{4,5}$}

\address{$^{1}$ Department of Physics, National Cheng Kung University, Tainan
701, Taiwan}

\ead{ycliang@mail.ncku.edu.tw}

\address{$^{2}$Academia da Força Aérea, Caixa Postal 970, 13643-970, Pirassununga,
SP, Brazil}

\ead{pmendonca@gmail.com}

\address{$^{3}$Melbourne Graduate School of Education, University of Melbourne,
Melbourne, VIC 3010, Australia }

\address{$^{4}$Centre for Quantum and Optical Science, Swinburne University
of Technology, Melbourne, VIC 3122, Australia}

\ead{rteh@swin.edu.au, mdreid@swin.edu.au, pdrummond@swin.edu.au}

\address{$^{5}$Institute of Theoretical Atomic, Molecular and Optical Physics
(ITAMP), Harvard University, Cambridge, Massachusetts, USA. }

\submitto{\RPP}
\begin{abstract}
Applications of quantum technology often require fidelities to quantify
performance. These provide a fundamental yardstick for the comparison
of two quantum states. While this is straightforward in the case of
pure states, it is much more subtle for the more general case of mixed
quantum states often found in practice. A large number of different
proposals exist. In this review, we summarize the required properties
of a quantum fidelity measure, and compare them, to determine which
properties each of the different measures has. We show that there
are large classes of measures that satisfy all the required properties
of a fidelity measure, just as there are many norms of Hilbert space
operators, and many measures of entropy. We compare these fidelities,
with detailed proofs of their properties. We also summarize briefly
the applications of these measures in teleportation, quantum memories
and quantum computers, quantum communications, and quantum phase-space
simulations. 
\end{abstract}
\maketitle
\tableofcontents{}

\ioptwocol

\section{Introduction}

Fidelity is a central concept to quantum information. It provides
a mathematical prescription for the quantification of the \emph{degree
of similarity} of a pair of quantum states. In practice, there are
many situations where such a comparison is useful. For example, since
any experimental preparation of a quantum state is limited by imperfections
and noise, one is generally interested to find how close the state
actually produced is to the state whose production was intended. This
is a common issue in quantum communications and quantum computing,
where one is interested in either generating or sending precisely
defined quantum states in the face of noise and other sources of error.

Another common application arises in the context of entanglement quantification~\cite{Guehne2009}:
the closer a given quantum state is to the set of separable states,
the less entangled it is, and vice-versa (see, however,~\cite{Rosset2012}).
Measuring and computing fidelities between quantum states is at the
heart of various quantum information tasks. In recent years, fidelity
measure has also been applied extensively to study quantum phase transitions.
For a review on this subject, see~\cite{gu2010fidelity} and references
therein.

For pure states, fidelity is well-defined. Yet pure states, by their
nature, are exactly what one does \emph{not} expect in a noisy, real-world
environment. Moreover, in large quantum systems, one needs to measure
exponentially many parameters in order fully determine the quantum
states, thus making the task infeasible in practice. Instead, one
could hope to obtain some partial information about the produced quantum
states, for example, by performing tomography on some (random) subsets
of the multipartite quantum states. Importantly, a generic multipartite
pure state is highly entangled across any bipartition~\cite{Hayden2006},
hence its reduced states on subsets are typically mixed. Thus, more
realistically, one must expect to deal with impure or mixed target
states that are obtained by examining a subsystem obtained by tracing
over a larger environment.

In all real-world experiments, there is also noise or coupling to
an environment. This is an essential part of the quantum-classical
transition, since it is the coupling to an environment that allows
measurements. From the point of view of a quantum technologist or
engineer, the environment causes decoherence, and this is the main
challenge in many quantum technology applications. Just as relevant
is the fact that one often wishes to analyze the performance of one
component\textemdash for example, a quantum gate\textemdash embedded
in a larger device, so that the environment is an essential part of
the system of interest, as explained above. Thus, one may argue that
the fidelity for mixed states is generic, and therefore the most practical
type of fidelity. Pure state fidelity represents an idealized case
only, which is typically non-scalable. For mixed states, however,
fidelity has no clearly unique definition, and a number of different
approaches exist.

Here the question is really one of distance: given two\emph{ density
matrices}, how close are they to each other in the appropriate Hilbert
space? This is an important issue, for example, if one wishes to understand
how accurately a given approximate calculation replicates that of
some target density matrix. The concept of distance in any vector
space is never a uniquely defined concept without further considerations
being applied.

There are other situations where mixed state fidelity might seem to
be important, but frequently they are more subtle than would appear
at first sight. For example, while it is known \cite{WoottersClone}
that one cannot clone quantum states, there is great interest in the
idea of optimal cloning \cite{BuzekHilleryPhysRevA.54.1844,ScaraniRMP.77.1225},
in which a state is copied as well as is allowed by quantum theory.
This requires one to quantify the fidelity of the clone or copy, in
order to decide if a given cloning strategy is really optimal. Similarly,
while quantum memories and quantum teleportation are permitted by
quantum mechanics, the real world of the laboratory leads to inevitable
noise and errors. Again, it is a mixed state fidelity that is important,
since the original state that is teleported or stored was typically
not a pure state originally. These applications may involve multiple
states as well, which requires not just fidelities, but averages over
them.

In this review, we analyze these issues by considering measures of
fidelity that apply to mixed as well as pure states. In Section II,
we review the properties of a fidelity measure as defined by Josza
\cite{jozsa1994fidelity}, which apply to the mixed state case, and
illustrate that these are satisfied by a large number of proposed
definitions of fidelity measure. This leads us to compare the different
properties of these measures, which is the main purpose of this review.
In Sections 3 and 4 we summarize the respective mathematical properties,
and make comparisons of the values the different fidelity measures
can have. We choose for specific comparison finite-dimensional ``qudit''
states with varying degrees of purity, and also make comparisons based
on randomly generated density operators. Finally in Section 5, we
review the application of fidelity measures to quantum information
protocols such as teleportation~\cite{bennett1993}, quantum memories~\cite{Lvovsky:2009aa}
and quantum gates~\cite{nielsen2000quantum}. Here, the fidelity
measure can indicate a level of security for a quantum state transfer,
or indicate the effectiveness of a logic operation. We present a useful
tool\textemdash phase space fidelity\textemdash for evaluating the
theoretical prediction of fidelity where systems and theories are
complex.

As applications are based on experimental measurements, we conclude
the review with a short discussion of the different types of experimental
fidelity measurements that have been reported in the literature. These
range from measurements on atomic states such as in ion traps, to
photonic qubit states that might imply entanglement based on post-selection.
The issue of defining the appropriate Hilbert space for the fidelity
measurement is also discussed. In the Summary, we present our main
conclusion, which is that for some situations different fidelity measures
to the commonly used Uhlmann-Josza measure may be advantageous. The
Appendix gives detailed proofs of nobel results where the proofs are
too lengthy for the main part of the review.

\section{Fidelity measures for mixed states}

\label{Sec:MixedFidelity}

\subsection{Measured and relevant Hilbert spaces\label{subsec:Relevant-and-irrelevant}}

Given the wide variety of quantum technologies involved, it is not
surprising that there are many ways to measure fidelity. These may
depend on the applications envisaged, or simply on what is measurable
in an experiment. Yet these are quite different issues. What is feasible
in an experiment may not be the fidelity measurement that is needed.
To understand this point, we must introduce the concept of measured
and relevant Hilbert space, which is fundamental to understanding
mixed state fidelity.

To motivate this analysis, we note that in certain types of relatively
well-structured quantum states, such as the Greenberger-Horne-Zeilinger~\cite{GHZ}
states, large-dimensional fidelity measurements have been reported,
with 10-14 qubit Hilbert spaces \cite{monz201114,song201710,chen2017observation}
being treated in ion-trap, photonic and superconducting quantum circuit
environments. Such measurements are especially important in allowing
the detection of decoherence and super-decoherence, where the decoherence
rates increase with system size \cite{reid2014quantum,galve2017microscopic}.

These experiments are impressive demonstrations, but scaling them
to even larger sizes is likely to become increasingly difficult owing
to their exponential complexity, leading to exponentially many measurements
being required in a general tomographic measurement. For this reason,
we expect measurements of mixed state fidelity to become more common
in large Hilbert spaces, as explained in the Introduction.

The first difficulty in carrying out a physical fidelity measurement
is to identify the relevant Hilbert space. All quantum systems are
coupled to other modes of the universe, but the state of Betelgeuse
is usually only relevant if one wishes to communicate over a 640 light
year distance. Thus, a useful fidelity measurement only measures the
relevant fidelity. Suppose the global quantum state is divided into
the relevant and irrelevant parts, with an orthogonal basis $\left|\psi_{i}\right\rangle _{{\rm rel}}\left|\phi_{i}\right\rangle _{{\rm irr}}$:
\begin{equation}
\left|\Psi\right\rangle =\sum_{ij}C_{ij}\left|\psi_{i}\right\rangle _{{\rm rel}}\left|\phi_{j}\right\rangle _{{\rm irr}}
\end{equation}

Next, the relevant density matrix is obtained by the usual partial
trace procedure, so that 
\begin{equation}
\rho=\sum_{ijk}C_{ij}C_{kj}^{*}\left|\psi_{i}\right\rangle _{{\rm rel}}{\left\langle \psi_{k}\right|_{{\rm rel}}}
\end{equation}
The irrelevant portion, $\left|\phi\right\rangle _{{\rm irr}}$, of
the quantum state may change in time. Yet given a large enough separation,
and for localized interactions, this will generally not alter the
relevant density matrix, owing to the trace over irrelevant parts.
However, if the product does not factorize initially, so the relevant
and irrelevant parts are entangled, then one does always have a mixed
state. This is the generic situation, and one should not assume that
a state is pure, without tomographic measurements that verify purity.

While subdivision is necessary, \emph{where should the dividing line
be}? For example, the internal states of all the different ions in
an ion trap, and their motional states, may all be relevant to an
ion-trap quantum computer. However, the state of motion of the vacuum
chamber will not be. Thus, the mixed state fidelity is a \emph{relative}
measure. It depends on the relevant Hilbert spaces. One will get different
results depending on how large the relevant Hilbert space is, which
is application-dependent.

Yet the relevant space in many cases may still be much larger than
the one in which the fidelity is \emph{measured}. For example, one
might only be able to measure the internal states of one or two ions
\cite{leibfried2003experimental}. The total relevant Hilbert space
can be larger than this. This means that there is some lost of information
which may be important to the application. One is still entitled to
claim to have measured a fidelity, but it is clearly not the only
relevant fidelity.

This means that the definition and understanding of a mixed state
fidelity is very central to quantum technologies. The original idea
of fidelity \cite{schumacher1995quantum}, as first introduced to
the quantum information community, is not immediately applicable to
an arbitrary pair of density matrices. However, it is implicit in
this definition that for a pair of pure states $\rho=\proj{\phi}$
and $\sigma=\proj{\psi}$, their fidelity should be defined by the
transition probability between the two states, i.e., 
\begin{equation}
\F(\proj{\phi},\proj{\psi})=\left|\langle\psi|\phi\rangle\right|^{2}.\label{Eq:Fidelity:AllPure}
\end{equation}
As pointed out subsequently \cite{jozsa1994fidelity}, this is indeed
a natural candidate for a fidelity measure since it corresponds to
the closeness of states in the usual geometry of Hilbert space.

When one of the quantum states, say, $\rho$ is mixed, there also
exists a generalization of Eq.~\eref{Eq:Fidelity:AllPure} in terms
of the transition probability between the two states, namely, 
\begin{equation}
\F(\rho,\proj{\psi})=\bra{\psi}\rho\ket{\psi}.\label{Eq:Fidelity:1mixed}
\end{equation}
Note that this was also implicitly defined in the original fidelity
measure ~\cite{schumacher1995quantum}. In \cite{mendonca2008},
this expression has also been referred as Schumacher's fidelity.

\subsection{Desirable properties of mixed state fidelities}

\label{sec:jozsa_n_additions}

As one might expect, not every function of two density matrices provides
a physically reasonable generalization of Eq.~\eref{Eq:Fidelity:1mixed}
to a pair of mixed states.

For example, at first glance it may seem that 
\begin{equation}
\F(\rho,\sigma)=\tr(\rho\,\sigma),\label{Eq:Fidelity:HSInnerProduct}
\end{equation}
serves as a useful generalization of Eq.~\eref{Eq:Fidelity:1mixed}
to an arbitrary pair of mixed states. However, the Hilbert-Schmidt
inner product leads to an unsatisfactory generalization of fidelity
\cite{jozsa1994fidelity}. For example, let us denote by $\Id_{d}$
the identity operator acting on a $d$-dimensional Hilbert space.
Then, adopting Eq.~\eref{Eq:Fidelity:HSInnerProduct} as the mixed
state fidelity would imply that all pairs of density matrices for
a two-state or qubit system of the form {$({\Id_{2}}/{2},\proj{\phi})$,
are just as similar as the identical pair $({\Id_{2}}/{2},{\Id_{2}}/{2})$.

This problem is soluble through a suitable normalization, as we show
below, but it illustrates the need for a suitable definition of a
fidelity measure. In order to avoid such difficulties, Jozsa proposed
the following list of fidelity axioms \cite{jozsa1994fidelity}, which
should be satisfied by any sensible generalization of Eq.~\eref{Eq:Fidelity:1mixed}
to a pair of mixed states: 
\begin{itemize}
\item[J1a)] \label{J1a} $\F(\rho,\sigma)\in[0,1]$ 
\item[J1b)] \label{J1b} $\F(\rho,\sigma)=1$ \emph{if and only if} $\rho=\sigma$ 
\item[J2)] \label{J2} $\F(\rho,\sigma)=\F(\sigma,\rho)$ 
\item[J3)] \label{J3} $\F(\rho,\sigma)=\tr(\rho\,\sigma)$ if either $\rho$
or $\sigma$ is a pure state 
\item[J4)] \label{J4} $\F(U\rho\,U^{\dagger},U\sigma U^{\dagger})=\F(\rho,\sigma)$
for all unitary $U$ 
\end{itemize}
Henceforth, we shall refer to this set of conditions as Jozsa's axioms.

Apart from these, it is convenient to append to this list the requirement
that any fidelity measure should vanish when applied to quantum states
of orthogonal support, i.e., 
\begin{itemize}
\item[J1c)] \label{J1c} $\F(\rho,\sigma)=0$ \emph{if and only if} $\rho\,\sigma=0$ 
\end{itemize}
Throughout, the requirements (J1)-(J4) will be taken as the most basic
requirements to be satisfied by any generalization of fidelity measure
for a pair of mixed states.

One may be interested in further subtleties in determining fidelity,
beyond the Jozsa axioms, depending on the application of the measure.
For example, in the case of a mixed system, the idea of a state-by-state
fidelity could be important. One may wish to investigate a cloning,
communication or quantum memory experiment. Suppose, in the experiment,
the input state is an unknown qudit state of dimension $d$, and has
maximum entropy. In other words, one has $\rho={\Id_{d}}/{d}$. According
to Jozsa's criterion, the highest fidelity output state $\sigma$
would be another, identical, maximal entropy state. Yet this would
not be a useful criterion on its own \textendash{} it is a necessary,
but not sufficient measure.

Here, one might wish to have the maximum fidelity for every expected
input state, pure or mixed, with an appropriate weight given by their
relative probability. This requires an understanding of the correlations
between the input and output states, given a communication alphabet
$\rho_{A},\rho_{B},\ldots$, which is related to the conditional information
measures found in communication theory \cite{caves1994quantum}. Such
more general issues cannot be investigated by just using simple fidelity
measures of the Jozsa type, and more sophisticated process fidelity
measures \cite{gilchrist2005distance} are needed, which we discuss
in detail later. However, even in this case, {a fidelity measure}
can be useful provided it is applied relative to every input density
matrix in the relevant communication alphabet, and then averaged with
its probability of appearance. This is called the average fidelity:
\begin{equation}
\langle\F(\rho,\sigma)\rangle=\sum P_{j}\F(\rho_{j},\sigma_{j}),\label{eq:fidav}
\end{equation}
where the pair $\rho_{j},\sigma_{j}$ occur with probability $P_{j}$.
However, the average fidelity may involve averages over mixed state
fidelities. Hence this concept is a relative one, and depends on the
precise definition of fidelity used.

\subsection{Uhlmann-Josza fidelity}

The most {widely-employed} generalization of Schumacher's fidelity
that has been proposed in the literature is the {\em Uhlmann-Jozsa
(U-J) fidelity} $\Fuj$ ~\cite{jozsa1994fidelity,uhlmann1976transition},
as the maximal {\em transition probability} between the purification
of a pair of density matrices $\rho$ and $\sigma$. Our choice of
notation for this fidelity will become evident in Section~\ref{Sec:NormBased}:
\begin{equation}
\Fuj(\rho,\sigma)\defeq\max_{\ket{\psi},\ket{\varphi}}{|\langle\psi|\varphi\rangle|^{2}}=\left(\tr\sqrt{\sqrt{\rho}\sigma\sqrt{\rho}}\right)^{2}\,.\label{eq:fid}
\end{equation}

In order to understand this definition, we note that here, $\ket{\psi}$
is what is called a purification of $\rho$. A purification is a state
in a notional extension of the Hilbert space $\mathcal{H}$ of $\rho$
to an enlarged space $\mathcal{H}'=\mathcal{H}\otimes\mathcal{H}_{2}$,
such that $\ket{\psi}$ is a member of this larger space. The limitation
on $\ket{\psi}$ is that when the projector, $\proj{\psi}$, is traced
over the {auxiliary} Hilbert space $\mathcal{H}_{2}$, it reduces
to $\rho,$ i.e, 
\[
\rho=\tr_{\mathcal{H}_{2}}\left[\proj{\psi}\right].
\]
This fidelity measure is often referred simply as {\em the fidelity}.
Nevertheless, the reader should beware that some authors (e.g., those
of~\cite{gu2010fidelity,nielsen2000quantum,luo2004informational,audenaert2008asymptotic})
have referred, instead, to the square root of $\Fuj(\rho,\sigma)$
as the fidelity. We show below that this fidelity measure, $\Fuj(\rho,\sigma)$,
is one of a large class of similar norm-based measures called $\Fp(\rho,\sigma)$.

Given the wide use of this definition of $\Fuj(\rho,\sigma)$, one
might naturally wonder if there is any need to search further, especially
as this definition carries with it a number of desirable properties.
There are also some difficulties with this approach, however, and
we list them here: 
\begin{itemize}
\item The U-J fidelity requires one to calculate or measure traces of square
roots of matrices. This is not trivial in cases of large or infinite
density matrices. 
\item The conceptual basis of the U-J fidelity is that both the density
matrices being compared are derived from an identical enlarged space
$\mathcal{H}'$. This is not always true in many applications of fidelity. 
\item Since the U-J fidelity is a maximum over purifications, the measured
U-J fidelity on a subspace is always greater or equal to the true
relevant fidelity for two pure states. This may introduce a bias in
estimating pure-state fidelities given a measurement over a reduced
Hilbert space. 
\end{itemize}
This leads to an obvious mathematical question: 
\begin{itemize}
\item \emph{Does the Josza set of requirements lead uniquely to the U-J
fidelity, or do other alternatives exist?} 
\end{itemize}
The purpose of this review is to answer this question. We show that,
indeed, other alternatives do exist that satisfy the Josza axioms.
A similar type of situation exists for the quantum entropy, where
it is known that there are many entropy-like measures. Some are more
suitable for given applications than others (see, e.g.,~\cite{Hu:JMP:2006,Muller:JMP:2013,Dupuis:2013}
and references therein).

\subsection{Alternative fidelities}

One of the first proposals of a fidelity measure alternative to $\Fuj$
was that provided by Chen~\emph{et al.} in~\cite{chen2002alternative}:
\begin{equation}
\Fc(\rho,\sigma)\defeq\frac{1-r}{2}+\frac{1+r}{2}\Fn(\rho,\sigma),\label{Eq:Fc}
\end{equation}
where $r=\frac{1}{d-1}$, $d$ is the dimension of the state space,
and 
\begin{equation}
\Fn(\rho,\sigma)\defeq\tr(\rho\,\sigma)+\sqrt{1-\tr{(\rho^{2})}}\sqrt{1-\tr{(\sigma^{2})}}.\label{Eq:Fn}
\end{equation}
For two-dimensional quantum states, $\Fuj|_{d=2}=\Fc|_{d=2}=\Fn|_{d=2}$~\cite{mendonca2008}.
Moreover, $\Fc$ admits a hyperbolic geometric interpretation in terms
of the generalized Bloch vectors~\cite{Kimura2003,Byrd2003}.

In 2008, $\Fn$ itself was proposed as an alternative fidelity measure
in~\cite{mendonca2008}; the same quantity was also independently
introduced in \cite{miszczak2009sub} by the name of super-fidelity,
as it provides an {\em upper bound} on $\Fuj$.

At about the same time, the square of the quantum affinity $A(\rho,\sigma)$
was proposed in \cite{ma2008geometric} (see also \cite{Raggio1984})
as a fidelity measure, by the name of {\em $A$-fidelity}: 
\begin{equation}
\Fa(\rho,\sigma)\defeq\left[\tr\left(\sqrt{\rho}\sqrt{\sigma}\right)\right]^{2}.\label{Eq:Fa}
\end{equation}
It is worth noting that, in contrast to $\Fn$, the $A$-fidelity
$\Fa$ provides a {\em lower bound} on $\Fuj$.

The super-fidelity $\Fn$ clearly does not satisfy the axiom (J1c).
In response to this \cite{wang2008alternative}, one may introduce
the quantity 
\begin{equation}
\Fgm(\rho,\sigma)\defeq\frac{\tr(\rho\,\sigma)}{\sqrt{\tr(\rho^{2})\,\tr(\sigma^{2})}}\,,\label{Eq:FGm}
\end{equation}
which is, instead, incompatible with axiom (J3). $\Fgm$ can be seen
as the Hilbert-Schmidt inner product between $\rho$ and $\sigma$
normalized by the geometric mean (GM) of their purities $\tr{(\rho^{2})}$
and $\tr{(\sigma^{2})}$.

Lastly, let us point out another quantity of special interest. The
non-logarithmic variety of the quantum Chernoff bound, $\Fq$, is
defined by Audenaert {\em et al.}~\cite{audenaert2007discriminating}
as: 
\begin{equation}
\Fq(\rho,\sigma):=\min_{0\le s\le1}\tr(\rho^{s}\,\sigma^{1-s}).\label{Eq:Fq}
\end{equation}
This quantity was not originally proposed as a fidelity measure. Instead,
it is related to the (asymptotic) probability of error incurred in
discriminating between quantum states $\rho$ and $\sigma$ when one
has access to arbitrarily many copies of them. Nonetheless, we will
include $\Fq$ in our subsequent discussion as it does have many desirable
properties of a fidelity measure.

In fact, amongst all the generalized fidelities formulas proposed
so far, only $\Fq$ and $\Fuj$ fully comply with Jozsa's axioms (cf.
Table~\ref{tbl:JAxiomsCheck:ExistingMeasure}). These will be our
main focus in the review out of these previously known fidelities,
although we point out that there is also an infinite class of norm-based
fidelities that comply with Jozsa's axioms as well as these two. We
note however that $\Fq$ is also computationally challenging. Not
only does it involve fractional powers, but one must optimize over
a continuous set of candidate measures, each involving different fractional
powers.

\begin{table}[!h]
\caption{\label{tbl:JAxiomsCheck:ExistingMeasure} Compatibility of existing
fidelity measures against {Jozsa's} axioms.}
\begin{tabular}{c|cccccc}
 & J1a  & J1b  & J1c  & J2  & J3  & J4 \tabularnewline
\hline 
\hline 
$\Fuj$  & $\surd$  & $\surd$  & $\surd$  & $\surd$  & $\surd$  & $\surd$ \tabularnewline
$\Fq$  & $\surd$  & $\surd$  & $\surd$  & $\surd$  & $\surd$  & $\surd$ \vspace{0.1cm}
 \tabularnewline
\hline 
$\Fn$  & $\surd$  & $\surd$  & $\times$  & $\surd$  & $\surd$  & $\surd$ \tabularnewline
$\Fc$  & $\surd$  & $\surd$  & $\times$  & $\surd$  & $\times$  & $\surd$ \tabularnewline
$\F_{A}$  & $\surd$  & $\surd$  & $\surd$  & $\surd$  & $\times$  & $\surd$ \tabularnewline
$\Fgm$  & $\surd$  & $\surd$  & $\surd$  & $\surd$  & $\times$  & $\surd$ \tabularnewline
\end{tabular}
\end{table}

\subsection{Norm-based fidelities}

\label{Sec:NormBased}

In addition, there are many norm-based fidelity measures that satisfy
these axioms. Consider the operator $A=\sqrt{\rho}\sqrt{\sigma}$
whose properties are closely related to the overlap of two density
matrices. There is an infinite class of unitarily invariant norms
of linear operators, called the Schatten-von-Neumann norms (or more
commonly Schatten norms), which can be used to measure the size of
such operators. These are defined, for $p\ge1$, as~\cite{bhatia97}:
\begin{equation}
\left\Vert A\right\Vert _{p}\equiv\left(\tr\left[\left(AA^{\dagger}\right)^{p/2}\right]\right)^{1/p}.
\end{equation}
These norms all satisfy Hölder and triangle inequalities, so that
${\left\Vert A_{1}A_{2}\right\Vert _{p}}\le\left\Vert A_{1}\right\Vert _{2p}\left\Vert A_{2}\right\Vert _{2p}$.
This particular inequality can be deduced, for example, from Corollary
IV.2.6 of~\cite{bhatia97}. They are not yet suitable as fidelity
measures, as they must be normalized appropriately to satisfy the
fidelity axioms. Hence, keeping this in mind, we define a $p$-fidelity
as: 
\begin{equation}
\Fp\left(\rho,\sigma\right)\defeq\frac{\left\Vert \sqrt{\rho}\sqrt{\sigma}\right\Vert _{p}^{2}}{\max\left[\left\Vert \sigma\right\Vert _{p}^{2},\left\Vert \rho\right\Vert _{p}^{2}\right]}.
\end{equation}

The proof that the axioms are satisfied is given in ~\ref{App:p-norm}.
We note that, for the special case of $p=1$, $\Fuj(\rho,\sigma)$
is exactly the Uhlmann-Josza fidelity. This follows since for any
Hermitian density matrix $\sigma$, $\left\Vert \sigma\right\Vert _{1}^{2}={\left(\tr\sigma\right)^{2}}=1$,
which is the same for all density matrices. Hence, the normalizing
term in this case is

\begin{equation}
\max\left[\left\Vert \sigma\right\Vert _{1}^{2},\left\Vert \rho\right\Vert _{1}^{2}\right]=\max\left[{\left(\tr\sigma\right)^{2}},{\left(\tr\rho\right)^{2}}\right]=1.
\end{equation}

In this review, we focus on $\F_{1}(\rho,\sigma)$ and $\F_{2}(\rho,\sigma)$,
which have especially desirable properties. In particular, $\F_{2}(\rho,\sigma)$,
which is defined as: 
\begin{equation}
\F_{2}(\rho,\sigma)\defeq\frac{\tr(\rho\,\sigma)}{\max\left[\tr(\rho^{2}),\tr(\sigma^{2})\right]}\,,\label{eq:Fmax}
\end{equation}
uses the Hilbert-Schmidt operator measure $\left\Vert A\right\Vert _{2}$,
which {is} often simpler to calculate than $\left\Vert A\right\Vert _{1}$.
In fact, all of the even order $p$-fidelities can be easily evaluated,
as they reduce to the form:

\begin{equation}
\F_{2p}(\rho,\sigma)\defeq\frac{\left\{ \tr\left[\left(\rho\,\sigma\right)^{p}\right]\right\} ^{1/p}}{\max\left\{ \left[\tr(\rho^{{2p}})\right]^{1/p},\left[\tr(\sigma^{{2p}})\right]^{1/p}\right\} }\,.\label{eq:Fmax-1}
\end{equation}

We finally note that a very similar type of circumstance occurs for
entropy, which is also sometimes used to calculate distances between
two density matrices. The traditional von Neumann entropy measure
involves a logarithm, and is often difficult to compute or measure.
This can be generalized to the Rényi entropy \cite{renyi1961measures},
which is:

\begin{equation}
S_{p}\left(\rho\right)=\frac{p}{1-p}\ln\left\Vert \rho\right\Vert _{p}.
\end{equation}

The Rényi entropy reduces to the usual von Neumann entropy, $S\left(\rho\right)$
in the limit of $p\rightarrow1$, just as the generalized fidelity
defined above reduces to the Uhlmann-Josza fidelity, in the same limit.
Both generalizations have advantages in simplifying computations \cite{hastings2010measuring}.

\subsection{Hilbert-Schmidt fidelities}

Although the measure $\Fgm$ does not comply with all of Jozsa's axioms,
its functional form suggests alternatives that are also worth investigating,
using the Hilbert-Schmidt norm. An example is the norm-based fidelity
$\F_{2}(\rho,\sigma)$ in the previous subsection, which is a special
case of $\Fp(\rho,\sigma)$.

To this end, note that for an arbitrary symmetric, non-vanishing function
$f$ that takes the purity of $\rho$ and $\sigma$ as arguments,
one can introduce the functional of $f[\tr(\rho^{2}),\tr(\sigma^{2})]$
\begin{equation}
\F_{f}(\rho,\sigma)=\frac{\tr(\rho\,\sigma)}{f[\tr(\rho^{2}),\tr(\sigma^{2})]},\label{eq:Ff}
\end{equation}
which are easily seen to satisfy a number of Jozsa's axioms. Specifically,
the symmetric property of $f$ and the cyclic property of trace guarantees
that the axiom (J2) is satisfied by $\F_{f}$, whereas the non-vanishing
nature of $f$ guarantees that (J1c) is fulfilled. In addition, the
fact that $f$ only takes the purity of $\rho$ and $\sigma$ as arguments
ensures that $\F_{f}$ complies with (J4).

The advantage of measures like this is that they only involve the
use of operator expectation values, in the sense that $\tr(\rho\,\sigma)$
is the expectation of $\sigma$ given the state $\rho$, or vice-versa.
Such measures tend to be readily expressed and {accessible} using
standard quantum mechanical techniques {applicable to infinite-dimensional}
Hilbert spaces. By contrast, measures involving nested square roots
of operators as found with $\Fuj(\rho,\sigma)$ {are not} so readily
calculated using standard quantum techniques in large Hilbert spaces,
which is important when one is treating bosonic cases. Similar issues
arise in the large Hilbert spaces that occur in many-body theory \cite{hastings2010measuring}.

Two classes of functions naturally fit into the above requirements,
namely, means (arithmetic, geometric or harmonic) and extrema (minimum
or maximum) of the purities of $\rho$ and $\sigma$. Other than the
geometric mean and the maximum\textemdash which give, respectively
$\Fgm$ and $\F_{2}$\textemdash the other functions give, explicitly,
\begin{eqnarray}
\Fam(\rho,\sigma) & \defeq\frac{2\tr(\rho\,\sigma)}{\tr(\rho^{2})+\tr(\sigma^{2})}\,,\nonumber \\
\Fhm(\rho,\sigma) & \defeq\frac{\tr(\rho\,\sigma)\left[{\tr(\rho^{2})+\tr(\sigma^{2})})\right]}{2\tr(\rho^{2})\tr(\sigma^{2})}\,,\label{Eq:NewMeasures}\\
\Fmin(\rho,\sigma) & \defeq\frac{\tr(\rho\,\sigma)}{\min\left[\tr(\rho^{2}),\tr(\sigma^{2})\right]},\nonumber 
\end{eqnarray}
where AM and HM stand for arithmetic and harmonic mean, respectively.
The compatibility of these new measures against Jozsa's axioms is
summarized in Table~\ref{tbl:JAxiomsCheck:NewMeasure}.

\begin{table}[!h]
\caption{\label{tbl:JAxiomsCheck:NewMeasure} Compatibility of $\F_{2}$ and
the candidate fidelity measures defined in \eref{Eq:NewMeasures}
against {Jozsa's} axioms.}
\begin{tabular}{c|cccccc}
 & J1a  & J1b  & J1c  & J2  & J3  & J4 \tabularnewline
\hline 
\hline 
$\F_{2}$  & $\surd$  & $\surd$  & $\surd$  & $\surd$  & $\surd$  & $\surd$ \vspace{0.1cm}
 \tabularnewline
\hline 
$\Fam$  & $\surd$  & $\surd$  & $\surd$  & $\surd$  & $\times$  & $\surd$ \tabularnewline
$\Fhm$  & $\times$  & $\times$  & $\surd$  & $\surd$  & $\times$  & $\surd$ \tabularnewline
$\Fmin$  & $\times$  & $\times$  & $\surd$  & $\surd$  & $\times$  & $\surd$ \tabularnewline
\end{tabular}
\end{table}

The (in)consistency of Eq.~\eref{Eq:NewMeasures} with axiom (J1c),
(J2), (J3) and (J4) can be verified easily either by inspection or
by the construction of counter-examples. Likewise, the normalization
of $\F_{2}$ follows easily from Cauchy-Schwarz inequality whereas
the \emph{incompatibility} of $\Fmin$ and $\Fhm$ with (J1a) can
be verified easily, for example, by considering the following pair
of $3\times3$ density matrices: 
\begin{equation}
\rho=\Pi_{0},\quad\sigma=\frac{3}{4}\Pi_{0}+\frac{1}{8}(\Pi_{1}+\Pi_{2}),
\end{equation}
where, for convenience, we denote 
\begin{equation}
\Pi_{i}\defeq\proj{i},
\end{equation}
as the rank-1 projector corresponding to the $i$-th computational
basis state, labelled ${|0\rangle,|1\rangle,...}$. Explicitly, one
finds that 
\begin{equation}
\Fhm(\rho,\sigma)=\frac{153}{152},
\end{equation}
and 
\begin{equation}
\Fmin(\rho,\sigma)=\frac{24}{19},
\end{equation}
both greater than 1, thus being incompatible with Jozsa axiom J1a.
As for the normalization of $\Fam$ (and $\Fgm$), its proofs can
be found in ~\ref{Sec:Proofs}.

In what follows, we will investigate the compatibility of the fidelity
measures listed in Table~\ref{tbl:JAxiomsCheck:ExistingMeasure}
and Table~\ref{tbl:JAxiomsCheck:NewMeasure} against other desirable
properties that have been considered. We will, however, dismiss $\Fhm$
and $\Fmin$ from our discussion as they do not even meet the basic
requirement of normalization. Apart from these we will mainly focus
on $\F_{2}$. This meets all the required fidelity axioms, and will
be termed the Hilbert-Schmidt fidelity.

The great advantage of $\F_{2}$ in computational terms is that as
well as complying with the extended version of Josza's axioms, it
is also relatively straightforward to compute and to measure. It only
involves expectation values of {Hermitian} operators, which are
computable and measurable with a variety of standard techniques in
quantum mechanics.

\section{Auxiliary fidelity properties}

\label{Sec:OtherProperties}

Let us now look into other auxiliary properties that have been discussed
in the literature. We shall focus predominantly in the three measures
that satisfy all the Jozsa axioms, namely, $\F_{1}$, $\F_{2}$ and
$\F_{Q}$. However, for completeness, we also provide a summary of
our understanding of the various properties of those candidate fidelity
measures that satisfy J1a, J1b, J2 and J4.

\subsection{Concavity properties}

A fidelity measure $\F(\rho,\sigma)$ is said to be separately concave
if it is a concave function of any of its argument. More precisely,
$\F(\rho,\sigma)$ is concave in its first argument if for arbitrary
density matrices $\sigma$, $\rho_{i}$ and arbitrary $p_{i}\ge0$
such that $\sum_{i}p_{i}=1$, 
\begin{equation}
\F\left(\sum_{i}p_{i}\rho_{i},\sigma\right)\geq\sum_{i}p_{i}\F(\rho_{i},\sigma).\label{Eq:Concave1st}
\end{equation}
By the symmetry of $\F$ {[}Jozsa's axiom (J2){]}, i.e., $\F(\rho,\sigma)=\F(\sigma,\rho)$,
a fidelity measure that is concave in its first argument is also concave
in its second argument.

A stronger concavity property is also commonly discussed in the literature.
Specifically, $\F(\rho,\sigma)$ is said to be jointly concave in
both of its arguments if: 
\begin{equation}
\F\left(\sum_{i}p_{i}\rho_{i},\sum_{j}p_{j}\sigma_{j}\right)\geq\sum_{i}p_{i}\F(\rho_{i},\sigma_{i}).\label{Eq:ConcaveJoint}
\end{equation}
This is a stronger concavity property in the sense that if Eq.~\eref{Eq:ConcaveJoint}
holds, so must Eq.~\eref{Eq:Concave1st}. This can be seen by setting
$\sigma_{j}=\sigma$ for all $j$ in \eref{Eq:ConcaveJoint}. Conversely,
if $\F(\rho,\sigma)${} is not separately concave, it also cannot
be jointly concave. Essentially, concavity property of a fidelity
measure tells us how the average state-by-state fidelity compares
with the fidelity between the two resulting ensembles of density matrices,
a point which we will come back to in Sec.~\ref{Sec:PureVsMixed}.
A summary of the concavity properties of the various fidelity measures
considered can be found in Table~\ref{tbl:Concavity}.

\begin{table}[!h]
\caption{\label{tbl:Concavity} Summary of the concavity properties of various
fidelity measures. The first column gives a list of the various measures,
while the second and third column give the compatibility of each measure
against the concavity property. {An asterisk $^{*}$ means that the
square root of the measure satisfies the property. Throughout, we
use the symbol $^{\ddag}$ to indicate that a particular property
was\textemdash to our knowledge\textemdash not discussed in the literature
previously. A question mark $?$ indicates that no counterexample
has been found.} }
\begin{tabular}{c|c|c}
 & Separate  & Joint \tabularnewline
\hline 
\hline 
$\Fuj$  & $\surd$~\cite{uhlmann1976transition}  & $\times${*} \tabularnewline
$\Fmax$  & $\times^{\ddag}$  & $\times$ \tabularnewline
$\Fq$  & $\surd$  & $\surd$~\cite{audenaert2007discriminating} \tabularnewline
\hline 
$\Fn$  & $\surd$  & $\surd$~\cite{mendonca2008} \tabularnewline
$\Fc$  & ?  & ? \tabularnewline
$\Fgm$  & $\times$~\cite{wang2008alternative}  & $\times$ \tabularnewline
$\Fam$  & $\times^{\ddag}$  & $\times$ \tabularnewline
$\F_{A}$  & ${?}^{*}$~\cite{luo2004informational}  & $\times^{\ddag*}$\cite{luo2004informational} \tabularnewline
\end{tabular}
\end{table}

To see that Eq.~\eref{Eq:Concave1st} does not hold for $\Fmax$
(as well as $\Fgm$ and $\Fam$), it suffices to set $p_{1}=p_{2}=\frac{1}{2}$
and consider the qubit or $2\times2$ density matrices $\rho_{1}=\frac{1}{10}(\Pi_{0}+9\Pi_{1})$,
$\rho_{2}=\frac{1}{5}(\Pi_{0}+4\Pi_{1})$ and $\sigma=\frac{1}{5}(3\Pi_{0}+2\Pi_{1})$.

On the other hand, to see that $\Fuj$ and $\Fa$ are not jointly
concave, it suffices to consider $p_{1}=\frac{49}{100}$, $p_{2}=\frac{1}{2}$,
$p_{3}=\frac{1}{100}$ together with the qutrit or $3\times3$ density
matrices $\rho_{1}=\Pi_{2}$, $\sigma_{1}=\Pi_{0}$, $\rho_{2}=\sigma_{2}=\Pi_{1}$,
$\rho_{3}=\frac{1}{5}{(3\Pi_{1}+2\Pi_{2})}$, $\sigma_{3}=\frac{1}{5}(2\Pi_{0}+3\Pi_{1})$
in \eref{Eq:ConcaveJoint}.

\subsection{Multiplicativity under tensor products}

A fidelity measure $\F(\rho,\sigma)$ is said to be multiplicative
if for all density matrices $\rho_{i}$, $\sigma_{i}$ and for all
integer $n\ge2$, 
\begin{equation}
\F\left(\bigotimes_{i=1}^{n}\rho_{i},\bigotimes_{j=1}^{n}\sigma_{j}\right)=\prod_{i=1}^{n}\F(\rho_{i},\sigma_{i});\label{Eq:Multiplicative}
\end{equation}
likewise $\F(\rho,\sigma)$ is said to {\em super-multiplicative}
if 
\begin{equation}
\F\left(\bigotimes_{i=1}^{n}\rho_{i},\bigotimes_{j=1}^{n}\sigma_{j}\right)\ge\prod_{i=1}^{n}\F(\rho_{i},\sigma_{i}).\label{Eq:SuperMultiplicative}
\end{equation}
Clearly, if $\F(\rho,\sigma)$ is (super)multiplicative for $n=2$,
it is also (super)multiplicative in the general scenario.

Two special instances of multiplicativity under tensor products are
worth mentioning. The first of which concerns the comparison of two
quantum states when one has access to $n$ copies of each state. In
this case, a (super)multiplicative measure $\F(\rho,\sigma)$ is also
(super)multiplicative in its tensor powers, i.e., 
\begin{equation}
\F(\rho^{\otimes n},\sigma^{\otimes n})=\left[\F(\rho,\sigma)\right]^{n}.\label{Eq:MultiplicativeTensorPower}
\end{equation}
The other special instance concerns the scenario when $\rho$ and
$\sigma$ are each appended with an uncorrelated state $\tau$. In
this case, multiplicativity demands 
\begin{equation}
\F\left(\rho\otimes\tau,\sigma\otimes\tau\right)=\F(\rho,\sigma)\F(\tau,\tau)=\F(\rho,\sigma).
\end{equation}
Intuitively, if $\F(\rho,\sigma)$ is also a measure of the overlap
between $\rho$ and $\sigma$, one would expect that multiplicativity
is satisfied, at least, in these two special instances. In Table~\ref{tbl:Multiplicativity},
we summarize the multiplicativity properties of the various fidelity
measures. For a proof of the (super)multiplicativity of $\Fmax$,
$\Fq$ (and $\Fc$), and a counterexample showing that $\Fam$ is
in general not multiplicative nor supermultiplicative, we refer the
reader to ~\ref{App:CountExamples}.

\begin{table}[!h]
\caption{\label{tbl:Multiplicativity} Summary of the multiplicativity of the
various fidelity measures. The first column gives the list of candidate
measures $\F$. From the second to the fourth column, we have, respectively,
the multiplicativity of the various measures $\F$ under the addition
of an uncorrelated ancillary state, under tensor powers and under
the general situation of \eref{Eq:Multiplicative}.}
\begin{tabular}{c|c|c|c}
 & Ancilla  & Tensor powers  & General \tabularnewline
\hline 
\hline 
$\Fuj$  & $\surd$~\cite{jozsa1994fidelity}  & $\surd$~\cite{jozsa1994fidelity}  & $\surd$~\cite{jozsa1994fidelity} \tabularnewline
$\F_{2}$  & $\surd^{\ddagger}$  & $\surd^{\ddagger}$  & Super$^{\ddagger}$ \tabularnewline
$\Fq$  & $\surd$~\cite{audenaert2008asymptotic}  & $\surd^{\ddagger}$  & Super$^{\ddagger}$ \vspace{0.1cm}
 \tabularnewline
\hline 
$\Fn$  & Super  & Super  & Super~\cite{mendonca2008}\tabularnewline
$\Fc$  & Super  & Super  & Super$^{\ddagger}$ \tabularnewline
$\Fgm$  & $\surd$~\cite{wang2008alternative}  & $\surd$~\cite{wang2008alternative}  & $\surd$~\cite{wang2008alternative} \tabularnewline
$\Fam$  & $\surd^{\ddagger}$  & $\times^{\ddag}$  & $\times^{\ddag}$ \tabularnewline
$\F_{A}$  & $\surd$~\cite{luo2004informational}  & $\surd$~\cite{luo2004informational}  & $\surd$~\cite{luo2004informational}\tabularnewline
\end{tabular}
\end{table}

\subsection{Monotonicity under quantum operations}

The physical operation of appending a given quantum state $\rho$
by a fixed quantum state $\tau$ discussed above is an example of
what is known as a completely positive trace preserving (CPTP) map.
If we denote a general CPTP map by $\mathcal{E}:\rho\to\mathcal{E}(\rho)$,
it is often of interest to determine, for all density matrices $\rho$
and $\sigma$, if the inequality: 
\begin{equation}
\F\left(\E(\rho),\E(\sigma)\right)\ge\F(\rho,\sigma).\label{Eq:MonotonicUp}
\end{equation}
is satisfied for a given candidate fidelity measure $\F$. In particular,
a fidelity measure that satisfies this inequality is said to be non-contractive
(or equivalently, monotonically non-decreasing) under quantum operations.

In this regard, it is worth noting that a measure $\F$ that (1) complies
with the requirement of unitary invariance (J4), (2) is invariant
under the addition of an uncorrelated ancillary state and is either
(3a) non-contractive under partial trace operation or is (3b) jointly
concave is also non-contractive under general quantum operations.
The sufficiency of conditions (1), (2) and (3a) follow directly from
the Stinespring representation (see, e.g.,~\cite{preskill2015lecture})
of CPTP maps while that of (1), (2) and (3b) also make use of a specific
representation of the partial trace operation as a convex mixture
of unitary transformations, see, e.g., Eq.~(33) of~\cite{Carlen2008}.
For example, since $\Fa$ is monotonic~\cite{luo2004informational}
under partial trace operation (likewise for $\Fuj$ and $\Fq$), the
above sufficiency condition allows us to conclude that $\Fa$ is also
monotonic under general quantum operations.

Apart from the partial trace operation and the extension of a quantum
state by a fixed ancillary state, the measurement of a quantum state
in some fixed basis followed by forgetting the measurement outcome
is another class of CPTP maps that one frequently encounters in the
context of quantum information. In particular, if each measurement
operator (the Kraus operator) is a rank-1 projector, the post-measurement
state would be the corresponding eigenstate. In this case, an evaluation
of the fidelity between the different outputs of the CPTP map corresponds
to an evaluation of the fidelity between the corresponding classical
probability distributions.

\begin{table}[!h]
\caption{\label{tbl:Monotonicity} Summary of the behavior of fidelity measures
under quantum operations. The first column gives the list of candidate
measures $\F$. From {the} second to the fourth column, we have,
respectively, the non-decreasing monotonicity of the measures $\F$
under partial trace operation, under projective measurements (see
{text for details}) and under general quantum operations. {That
is, we mark an entry with a tick $\protect\surd$ if Eq.}~\ref{Eq:MonotonicUp}
{holds for the corresponding CPTP map.}}
\begin{tabular}{c|c|c|c}
 & Partial trace  & Projection  & General \tabularnewline
\hline 
\hline 
$\Fuj$  & $\surd$  & ${\surd}$  & $\surd$~ \cite{nielsen2000quantum}\tabularnewline
$\F_{2}$  & $\times^{\ddag}$  & $\times^{\ddag}$  & $\times$ \tabularnewline
$\Fq$  & $\surd$  & $\surd$  & $\surd$~\cite{audenaert2008asymptotic} \vspace{0.1cm}
 \tabularnewline
\hline 
$\Fn$  & $\times$~\cite{mendonca2008}  & ?  & $\times$ \tabularnewline
$\Fc$  & $\times^{\ddag}$  & ?  & $\times$ \tabularnewline
$\Fgm$  & $\times^{\ddag}$  & $\times^{\ddag}$  & $\times$ \tabularnewline
$\Fam$  & $\times^{\ddag}$  & $\times^{\ddag}$  & $\times$\tabularnewline
$\F_{A}$  & $\surd$~\cite{luo2004informational}  & $\surd$~\cite{luo2004informational}  & $\surd$~\cite{Raggio1984} \tabularnewline
\end{tabular}
\end{table}

The monotonicity of the various candidate fidelity measures for the
few different CPTP maps discussed above is summarized in Table~\ref{tbl:Monotonicity}.
That $\Fmax$ may be contractive under partial trace operation can
be seen by considering the two-qubit density matrices $\rho=\Pi_{1}\otimes\rho_{B}$,
$\sigma=\frac{1}{2}\Id_{2}\otimes\sigma_{B}$ where 
\[
\rho_{B}=\left[\begin{array}{cc}
0.3 & 0.3\\
0.3 & 0.7
\end{array}\right],\quad\sigma_{B}=\left[\begin{array}{cc}
0.06 & 0.2\\
0.2 & 0.94
\end{array}\right]
\]
and the partial trace of $\rho$ and $\sigma$ over subsystem B. As
for the monotonicity of $\Fmax$, $\Fgm$ and $\Fam$ under projective
measurements, one may verify that these measures may be indeed contractive
by considering the qubit density matrices 
\begin{eqnarray}
\rho={\left[\begin{array}{cc}
0.35 & -0.25-0.2\,\ii\\
-0.25+0.2\,\ii & 0.65
\end{array}\right]},\nonumber \\
\sigma={\left[\begin{array}{cc}
0.82 & -0.2-0.24\,\ii\\
-0.2+0.24\,\ii & 0.18
\end{array}\right],}
\end{eqnarray}
and a rank-1 projective measurement in the computational basis. For
the monotonicity of $\Fc$, $\Fgm$ and $\Fam$ under the partial
trace operation, we refer the reader to~\ref{App:CountExamples}
for counterexamples.

The monotonicity under partial trace obtained in $\F_{1}$ and $\Fq$
means that a mixed state fidelity measured according to $\F_{1}$
and $\Fq$ is greater than or equal to the corresponding fidelity
when evaluated on a larger Hilbert space. This implies that there
is a bias in using these measured fidelities on a smaller Hilbert
space as an estimator for the fidelity on an enlarged relevant Hilbert
space, when using $\F_{1}$ and $\Fq$.

This potentially undesirable property is not shared by $\F_{2}$.
However, the general question of which fidelity is the best unbiased
estimator under partial trace operations from a randomized extension
of the measured Hilbert space appears to be an open problem.

\label{Sec:Metric}

\subsection{Metrics}

Intuitively, one expects that if $\F(\rho,\sigma)$ is a measure of
the degree of similarity or overlap between $\rho$ and $\sigma$,
a proper distance measure, i.e., a metric can be constructed via some
functionals of $\F(\rho,\sigma)$ which vanishes for $\rho=\sigma$.

In this section, we review what is known about the metric properties
of three functionals of $\F(\rho,\sigma)$, namely, $\arccos[\sqrt{\F(\rho,\sigma)}]$,
$\sqrt{1-\sqrt{\F(\rho,\sigma)}}$ and $\sqrt{1-\F(\rho,\sigma)}$
for the various fidelity measures discussed in the previous section.
Following the literature, one may want to refer to these functionals,
respectively, as the \emph{modified} Bures angle~\cite{nielsen2000quantum},
the \emph{modified} Bures distance~\cite{Hubner1992} and the \emph{modified}
sine distance~\cite{Rastegin06}. The results are summarized in Table~\ref{tbl:RelatedMetrics}
{while a proof of the respective metric properties can be found in
}~\ref{App:Met}.

\begin{table}[!h]
\caption{\label{tbl:RelatedMetrics} Metric properties for some functionals
of the fidelity measures $\F=\F(\rho,\sigma)$, as discussed in Sec.~\ref{Sec:MixedFidelity}.}
\begin{tabular}{c|ccc}
 & $\arccos[\sqrt{\F}]$  & $\sqrt{1-\sqrt{\F}}$  & $\sqrt{1-\F}$ \tabularnewline
\hline 
\hline 
$\Fuj$  & $\surd$~\cite{gilchrist2005distance}  & $\surd$~\cite{gilchrist2005distance}  & $\surd$~\cite{gilchrist2005distance}\tabularnewline
$\F_{2}$  & $\times^{\ddagger}$  & $\times^{\ddagger}$  & $\surd^{\ddagger}$ \tabularnewline
$\Fq$  & $\times^{\ddag}$  & $\times^{\ddag}$  & $\times^{\ddag}$ \vspace{0.1cm}
 \tabularnewline
\hline 
$\Fn$  & $\times$~\cite{mendonca2008}  & $\times$~\cite{mendonca2008}  & $\surd$~\cite{mendonca2008}~\tabularnewline
$\Fc$  & ?  & ?  & $\surd^{\ddagger}$ \tabularnewline
$\Fgm$  & $\times^{\ddag}$  & $\times^{\ddag}$  & $\surd^{\ddagger}$ \tabularnewline
$\Fam$  & $\times^{\ddag}$  & $\times^{\ddag}$  & ? \tabularnewline
$\F_{A}$  & ?  & $\surd$~\cite{Raggio1984}~  & $\surd^{\ddag}$ \tabularnewline
\end{tabular}
\end{table}

Somewhat surprisingly, despite the fact that $\Fq$ shares many nice
properties with $\Fuj$, none of these functionals derived from $\Fq$
actually behave like a metric for the space of density matrices. This
can be verified, for example, by noticing a violation of the triangle
inequality for all these functionals of $\Fq$ under the choice 
\begin{eqnarray}
\rho=\frac{1}{10}(3\Pi_{0}+7\Pi_{1}),\,\nonumber \\
\sigma=\frac{1}{100}(\Pi_{0}+99\Pi_{1}),\,\nonumber \\
\tau=\frac{1}{5}(\Pi_{0}+4\Pi_{1}).
\end{eqnarray}
As for a violation of the triangle inequality by the other functionals
presented in the table, it is sufficient to consider the following
qutrit density matrices: 
\begin{eqnarray}
\rho=\frac{1}{5}(\Pi_{1}+4\Pi_{2}),\,\nonumber \\
\sigma=\Pi_{0},\,\nonumber \\
\tau=\frac{1}{5}\Pi_{0}+\frac{1}{20}\Pi_{1}+\frac{3}{4}\Pi_{2}.
\end{eqnarray}

Notice that other than the functionals considered above, it was also
shown in~\cite{ma2008geometric} that $\max_{\tau}|\Fn(\rho,\tau)-\Fn(\sigma,\tau)|$,
a functional constructed from the super-fidelity $\Fn$, is also a
metric. In fact, with very similar arguments, the same authors showed
in~\cite{ma2009pla} that the same functional with $\Fuj$ instead
of $\Fn$ is also a metric.

\section{Comparisons, bounds, and relations between measures }

\label{Sec:Comparisons}

To understand the differences between the fidelity measures, one must
compare them quantitatively. In some cases there are rigorous bounds
that relate the fidelities, while in other cases comparisons are made
graphically. As a rough guide, the average behavior of full-rank,
random density matrices indicates that when making comparisons with
$p>1$, our numerics for small $p$ and $d$ suggest that: 
\begin{equation}
\F_{p}(\rho,\sigma)\lesssim\Fuj(\rho,\sigma)\lesssim\Fq(\rho,\sigma)
\end{equation}
often holds. However, this is not a hard and fast rule. In qubit cases
these are rather strict bounds (at least for $p=2$). In larger Hilbert
spaces, however, these inequalities are only approximately true, with
exceptions that strongly depend on the rank of the density matrices
being compared.

These different cases are explained below.

\subsection{Bounds}

First, we provide a summary of inequalities relating some of these
fidelity measures, or some bounds on them. To begin with, it was established
in~\cite{raggio1982comparison} (see also~\cite{Raggio1984} and~\cite{luo2004informational})
that 
\begin{equation}
\Fuj(\rho,\sigma)\le\sqrt{\Fa(\rho,\sigma)}\le\sqrt{\Fuj}(\rho,\sigma).\label{Eq:FujVsFa}
\end{equation}
Later, in~\cite{audenaert2008asymptotic} (see also~\cite{audenaert2007discriminating}),
these inequalities were rediscovered and extended to 
\begin{equation}
\Fuj(\rho,\sigma)\le{\Fq}(\rho,\sigma)\le\sqrt{\Fa(\rho,\sigma)}\le\sqrt{\Fuj(\rho,\sigma)},\label{Eq:FujVsFq}
\end{equation}
where the second of these inequalities follows directly from the definition
of $\Fa$ and $\Fq$ given, respectively, in Eqs.~\eref{Eq:Fa}
and \eref{Eq:Fq}. At about the same time, $\Fn$ was also shown~\cite{miszczak2009sub}
to be an upper bound on $\Fuj$, i.e., 
\begin{equation}
\Fuj(\rho,\sigma)\le{\Fn}(\rho,\sigma).\label{Eq:FujVsFn}
\end{equation}
Here, we add to this list by showing that $\Fmax(\rho,\sigma)$ actually
provides a lower bound to $\Fn(\rho,\sigma)$.

\begin{theorem}\label{teo1} For arbitrary Hermitian matrices $\rho$
and $\sigma$ such that $\tr(\rho^{2})\le1$ and $\tr(\sigma^{2})\le1$
\begin{equation}
\F_{2}(\rho,\sigma)\leq\Fn(\rho,\sigma)\,.\label{eq:Fminlb}
\end{equation}
\end{theorem} \begin{proof}

First, let us rewrite an arbitrary pair of $d$-dimensional density
matrices $\rho$, $\sigma$ using an orthonormal basis of Hermitian
matrices $\vec{\Upsilon}=(\Upsilon_{0},\Upsilon_{1},\ldots,\Upsilon_{d^{2}-1})$,
\begin{equation}
\rho=\vec{u}\cdot\vec{\Upsilon}\quad\mbox{and}\quad\sigma=\vec{v}\cdot\vec{\Upsilon},\label{eq:param}
\end{equation}
where $\vec{u}$, $\vec{v}\in\mathbb{R}^{d^{2}}$ are the expansion
coefficients of $\rho$ and $\sigma$ in the basis $\vec{\Upsilon}$.
We may now rewrite $\F_{2}$ and $\Fn$ as 
\begin{eqnarray}
\F_{2}(\rho,\sigma)=\frac{\vec{u}\cdot\vec{v}}{\max(u^{2},v^{2})}\,,\nonumber \\
\Fn(\rho,\sigma)=\vec{u}\cdot\vec{v}+\sqrt{1-u^{2}}\sqrt{1-v^{2}}\,,\label{eq:f2-geom}
\end{eqnarray}
where $u=\|\vec{u}\|_{2}$ and $v=\|\vec{v}\|_{2}$.

It is straightforward to check that inequality~(\ref{eq:Fminlb})
holds true if and only if the following inequality is true: 
\begin{equation}
\frac{\vec{u}\cdot\vec{v}}{\max(u^{2},v^{2})}\leq\sqrt{\frac{1-\min(u^{2},v^{2})}{1-\max(u^{2},v^{2})}}\,,\label{eq:bound}
\end{equation}
which obviously holds if $u=v$. For definiteness, let $u>v$ and
we have 
\begin{equation}
\frac{\vec{u}\cdot\vec{v}}{\sqrt{1-v^{2}}}\leq\frac{u^{2}}{\sqrt{1-u^{2}}}\,.
\end{equation}
This is easily seen to hold since, for $u>v$, the numerator of the
l.h.s. is dominated by the numerator of the r.h.s., whereas the denominator
of the l.h.s. dominates the denominator of the r.h.s.. The case where
$v>u$ is completely analogous. \end{proof} \begin{corollary}\label{cor1}
For arbitrary \emph{qubit} density matrices $\rho$ and $\sigma$,
\begin{equation}
\F_{2}(\rho,\sigma)\leq\Fuj(\rho,\sigma)\,.\label{Eq:F2vsF1}
\end{equation}
\end{corollary} \begin{proof} This follows trivially from Theorem~\ref{teo1}
and the fact that for qubit density matrices $\Fuj(\rho,\sigma)=\Fn(\rho,\sigma)$,
see~\cite{mendonca2008}. \end{proof}

From the inequality of arithmetic and geometric means, as well as
the definitions given in Eqs.~\eref{Eq:FGm}, \eref{Eq:NewMeasures}
and \eref{eq:Fmax}, it is easy to see that the following bounds
hold: 
\begin{equation}
\Fmax(\rho,\sigma)\le\Fam(\rho,\sigma)\leq\Fgm(\rho,\sigma)\,.\label{Eq:F2-gm-am}
\end{equation}
Besides, some straightforward calculation starting from the definitions
given in Eqs.~\eref{Eq:Fc} and \eref{Eq:Fn} leads to 
\begin{equation}
\Fn(\rho,\sigma)\leq\Fc(\rho,\sigma)\,.\label{Eq:FnvsFc}
\end{equation}

Some remarks are now in order. Given that $\Fn(\rho,\sigma)=\Fuj(\rho,\sigma)$
for qubit states and the fact that $\Fq(\rho,\sigma)$ and $\Fn(\rho,\sigma)$
both provide an upper bound on $\Fuj(\rho,\sigma)$, one may wonder: 
\begin{itemize}
\item[1)] Could $\Fn(\rho,\sigma)$ provide a lower bound on $\Fq(\rho,\sigma)$,
or the other way around? 
\item[2)] Since $\Fa(\rho,\sigma)$ and the sub-fidelity introduced in~\cite{miszczak2009sub}
both provide a lower bound on $\Fuj(\rho,\sigma)$, could it be that
one of these quantities also lower bounds the other? 
\item[3)] Could it be that Eq.~\eref{Eq:F2vsF1} also holds for higher-dimensional
density matrices? 
\end{itemize}
Here, let us note that counterexamples to \emph{all} of the above
conjectures can be easily found by considering pairs of qutrit density
matrices. However, we leave open the possibility of bounding these
quantities using nonlinear (including polynomial) functionals of the
other fidelity measures.

\subsection{Comparisons: interpolated qubit states}
\begin{center}
\begin{figure}[h!]
\includegraphics[width=0.95\columnwidth]{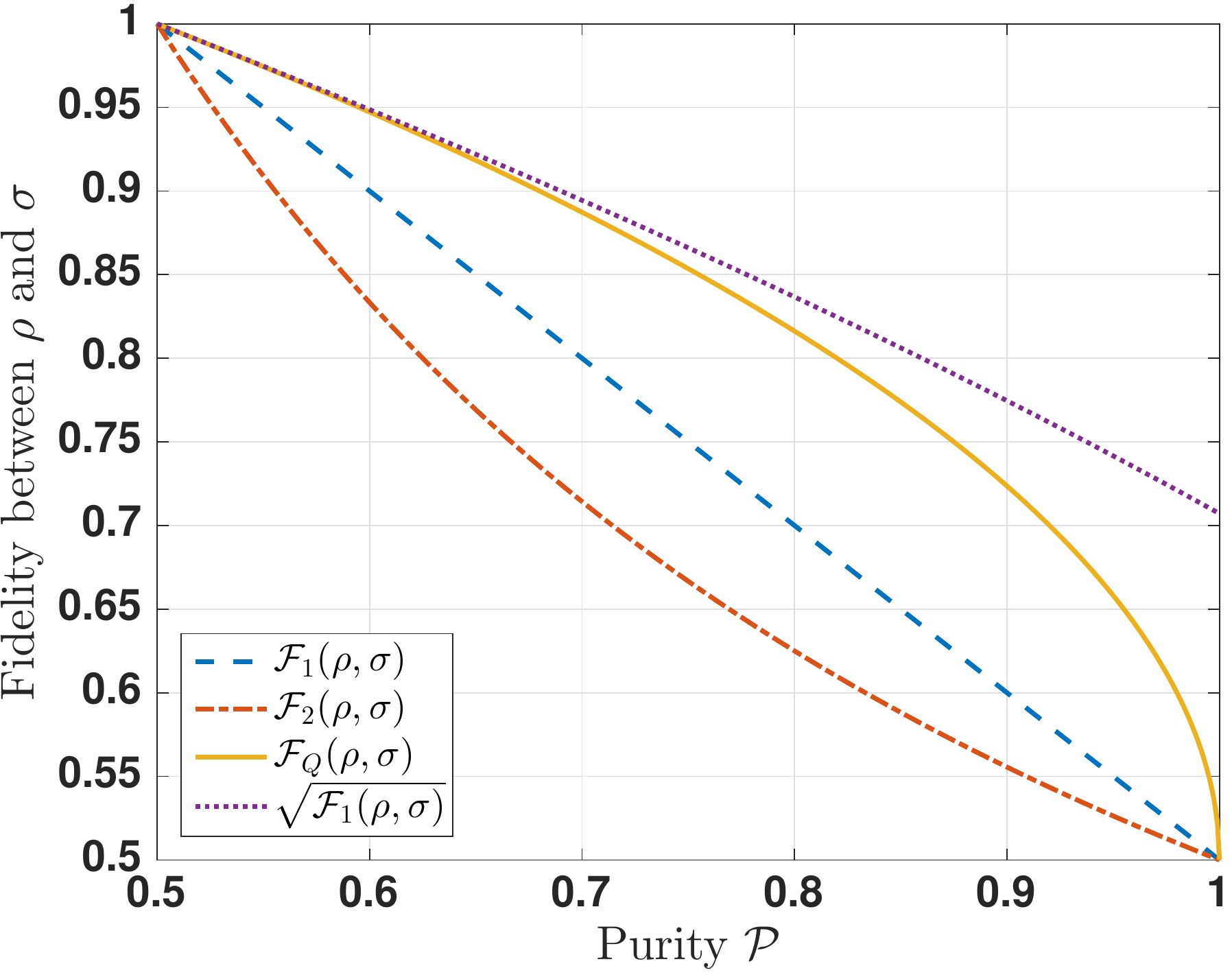} \caption{\label{Fig:FidComp}The fidelity between the single-qubit (mixed)
states given in Eq.~\eref{Eq:ExampleStates} as a function of purity
for the Uhlmann-Jozsa fidelity $\Fuj$, Eq.~\eref{eq:fid}, the
Hilbert-Schmidt fidelity $\Fmax$, Eq.~\eref{eq:Fmax} and the
non-logarithmic variety of the quantum Chernoff bound $\Fq$, Eq.~\eref{Eq:Fq}.
For comparison, cf. Eq.~\eref{Eq:FujVsFq}, we have also included
in the plot the square root of $\Fuj$ (sometimes referred to as \emph{the
fidelity}, see~\cite{nielsen2000quantum}) as a function of the purity
{$\mathcal{P}$}.}
\end{figure}
\par\end{center}

How large are the differences between these measures? In order to
understand this, we first turn to a simple but concrete example. To
this end, consider two families of qubit density matrices that interpolate
between maximally mixed states\textemdash that are necessarily identical\textemdash and
pure states that are distinct, with an interpolation parameter of
$r\in[0,1]$: 
\begin{eqnarray}
\rho(r)=\frac{1}{2}\left(\Id_{2}+r\sigma_{x}\right),\quad\sigma(r)=\frac{1}{2}\left(\Id_{2}+r\sigma_{z}\right).\label{Eq:ExampleStates}
\end{eqnarray}

The purity of a density operator $\rho$ is given by $\mathcal{P}={\tr\ensuremath{\left(\rho^{2}\right)}}$.
For the two quantum states in this comparison, their purities are
set to be equal and are parametrized by $r$, i.e., 
\begin{eqnarray}
{\mathcal{P}} & =\frac{1}{2}\left(1+r^{2}\right)\,,\label{eq:purity}
\end{eqnarray}
where $r=0$ corresponds to a maximally mixed state while $r=1$ corresponds
to a pure state. Note that although these density matrices might look
at first as though they are a sum of two mixed states, they are not.
In the limit of $r=1$, each becomes a distinct pure state, with a
non-vanishing inner product.

In Figure~\ref{Fig:FidComp}, we show the result of the fidelities
$\Fuj(\rho,\sigma)$ {[}cf. Eq.~\eref{eq:fid}{]}, $\Fmax(\rho,\sigma)$
{[}cf. Eq.~\eref{eq:Fmax}{]} and $\Fq(\rho,\sigma)$ {[}cf. Eq.~\eref{Eq:Fq}{]}
as a function of the purity $P$, Eq.~\eref{eq:purity}, of the
states given in Eq.~\eref{Eq:ExampleStates}. It is worth noting
that for these states, the inequalities of Eqs.~\eref{Eq:FujVsFn},
\eref{Eq:F2-gm-am}, and \eref{Eq:FnvsFc} are all saturated,
thereby giving $\Fuj(\rho,\sigma)=\Fn(\rho,\sigma)=\Fc(\rho,\sigma)$
and $\Fmax(\rho,\sigma)=\Fam(\rho,\sigma)=\Fgm(\rho,\sigma)$. Moreover,
for these states, it can also be verified that $\sqrt{\Fa}(\rho,\sigma)=\Fq(\rho,\sigma)$,
thereby saturating the second inequality in Eq.~\eref{Eq:FujVsFq}.

\subsection{Comparisons: random density matrices}

As a second type of comparison, we generate two random density matrices
and compare them, in a Hilbert space of arbitrary dimension. In general,
these are generated by taking a random Gaussian matrix $\bm{g}$,
whose elements $g_{ij}$ are complex random numbers of unit variance.
Starting from a complex Gaussian matrix, called a Ginibre ensemble~\cite{ginibre1965statistical},
a random, positive-semidefinite density matrix with unit trace is
generated \cite{zyczkowski2011generating} by letting: 
\begin{equation}
\rho=\frac{gg^{\dagger}}{\tr\left(gg^{\dagger}\right)}.\label{eq:randomGaussian}
\end{equation}
Some investigations of random matrix fidelity have been carried out
previously using the $\Fuj$ fidelity~\cite{zyczkowski2005average},
while here we focus on the comparative behavior of the different fidelity
measures. 
\begin{center}
\begin{figure}[h!]
\includegraphics[width=1\columnwidth]{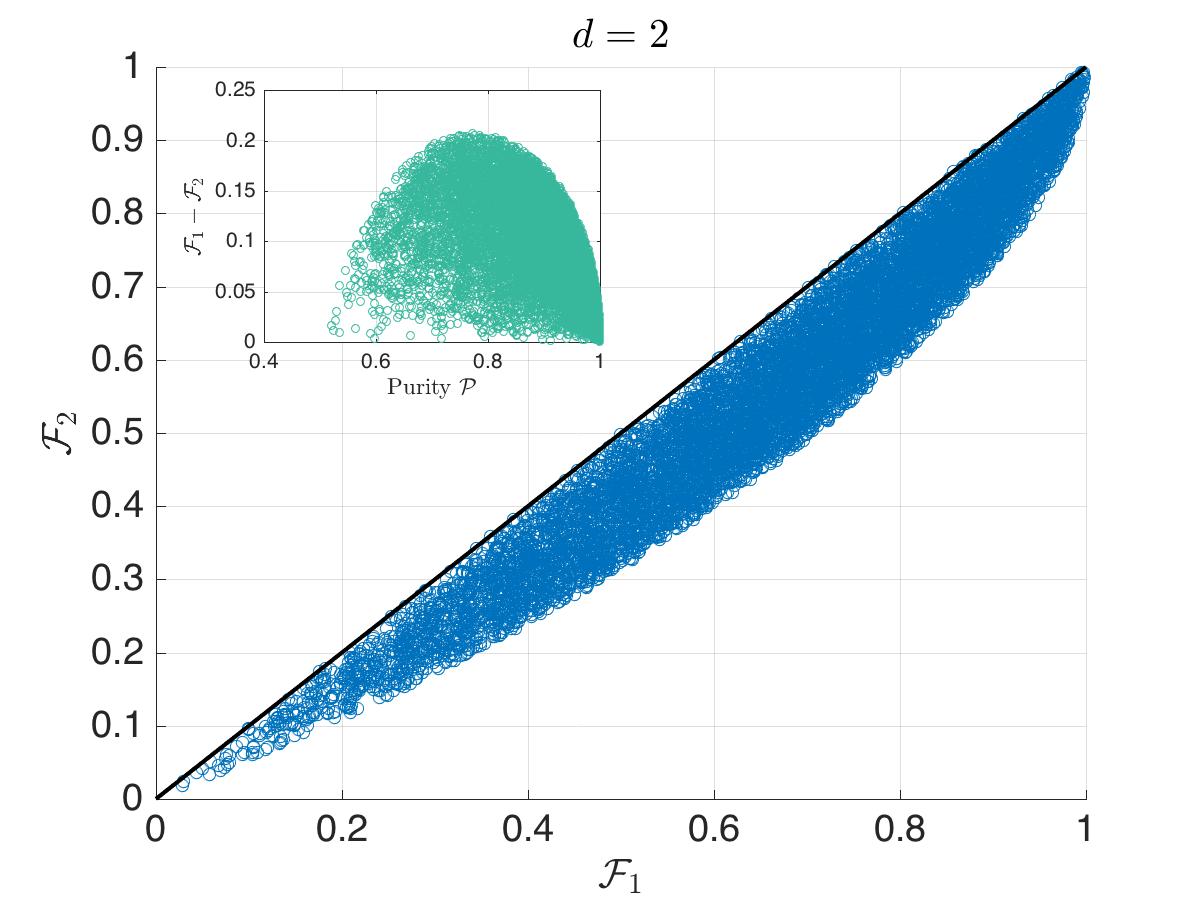} \caption{\label{Fig:F2-F1:d2} Scatter plot showing the Uhlmann-Jozsa fidelity
$\Fuj$, Eq.~\eref{eq:fid}, and the Hilbert-Schmidt fidelity $\Fmax$,
Eq.~\eref{eq:Fmax}, for $10^{4}$ random pairs of single-qubit
(mixed) density matrices. The black line satisfying $\Fuj=\Fmax$
is a guide for the eye. The inset shows a scatter plot of the difference
$\Fuj-\Fmax$ vs the {\em maximum} purity {$\mathcal{P}$} of
these pairs of density matrices. }
\end{figure}
\par\end{center}

In Figure~\ref{Fig:F2-F1:d2}, we compare plots with two qubit random
matrices, and give scatter plots of both $\Fuj-\Fmax$ against the
maximum purity of the pair, as well as $\Fmax$ against $\Fuj$. We
note that for this case, as expected from Eq.~\eref{Eq:F2vsF1},
$\Fuj\ge\Fmax$ for $d=2$. In addition, there is a lower bound on
the state purity, since $\mathcal{P}\ge0.5$ for qubit density matrices.
\begin{center}
\begin{figure}[h!]
\includegraphics[width=1\columnwidth]{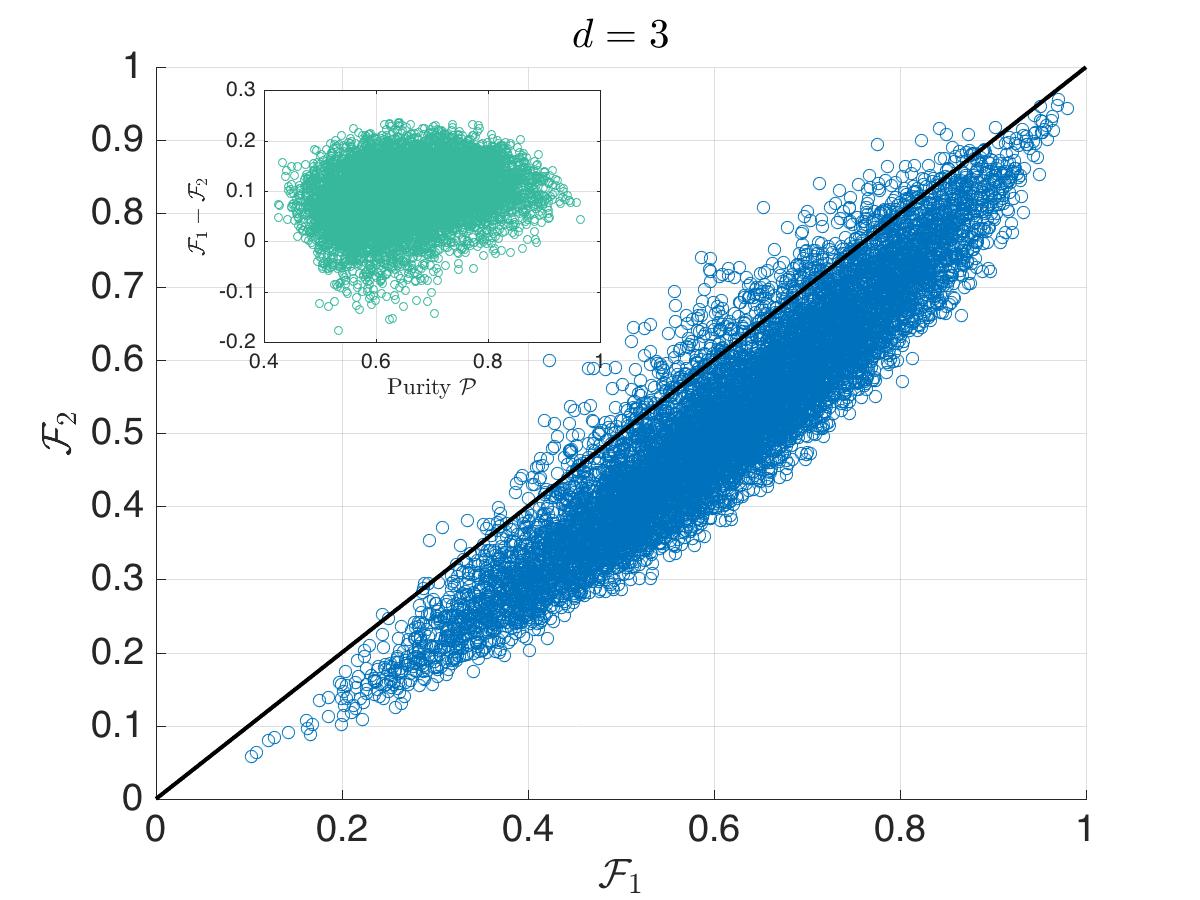} \caption{\label{Fig:F2-F1:d3} Scatter plot showing the Uhlmann-Jozsa fidelity
$\Fuj$, Eq.~\eref{eq:fid}, and the {Hilbert-Schmidt} fidelity
$\Fmax$, Eq.~\eref{eq:Fmax}, for $10^{4}$ random pairs of single-qutrit
(mixed) density matrices. The black line satisfying $\Fuj=\Fmax$
is a guide for the eye. The inset shows a scatter plot of the difference
$\Fuj-\Fmax$ vs the {\em maximum} purity {$\mathcal{P}$ of these}
pairs of density matrices.}
\end{figure}
\par\end{center}

\begin{center}
\begin{figure}[h!]
\includegraphics[width=1\columnwidth]{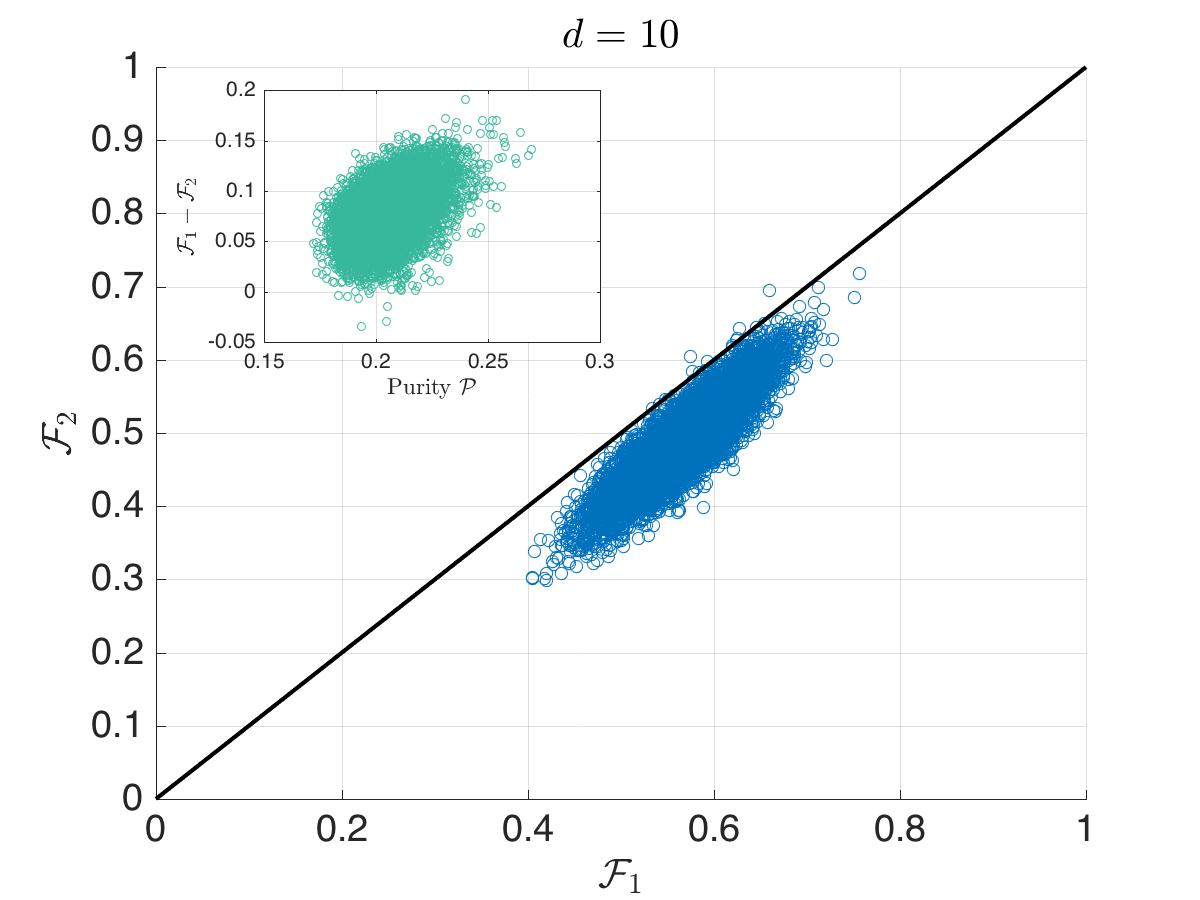} \caption{\label{Fig:F2-F1:d10} Scatter plot showing the Uhlmann-Jozsa fidelity
$\Fuj$, Eq.~\eref{eq:fid}, and the Hilbert-Schmidt fidelity $\Fmax$,
Eq.~\eref{eq:Fmax}, for $10^{4}$ random pairs of $10\times10$
density matrices. The black line satisfying $\Fuj=\Fmax$ is a guide
for the eye. The inset shows a scatter plot of the difference $\Fuj-\Fmax$
vs the {\em maximum} purity $\mathcal{P}$ of these pairs of density
matrices.}
\end{figure}
\par\end{center}

In Figure~\ref{Fig:F2-F1:d3}, we compare plots with two qutrit random
matrices, again with scatter plots of both $\Fuj-\Fmax$ against the
maximum purity of the pair, together with the two fidelities plotted
against each other. For this case, with $d=3$, both $\Fuj>\Fmax$,
and $\Fuj<\Fmax$ are possible. As a concrete example of the latter
case, one only has to consider a comparison of two diagonal qutrit
density matrices, where $p$ is a real coefficient such that $0<p<1$:
\begin{equation}
\rho=\left(1-p\right)\Pi_{1}+p\Pi_{2},\quad\sigma=\left(1-p\right)\Pi_{0}+p\Pi_{2}.
\end{equation}
This has the property that $\Fuj<\Fmax$ for any value of $p$ such
that $0\neq p\neq1$, since $\Fmax$ is divided by the maximum purity,
and this is less than unity. However, although extremely simple, this
is also atypical. In almost all cases, the two density matrices being
compared are not simultaneously diagonal. For the average case, we
can see instead that $\langle\Fuj\rangle>\langle\Fmax\rangle$. This
corresponds to the intuitive expectation that, since $\Fuj$ is a
maximum over purifications, it will generally have a bias towards
high values.

Finally in the case of two random matrices of larger dimension, we
compare two $10\times10$ qudit matrices, with the other details as
previously. For this case, just as with $d=3$, both $\Fuj>\Fmax$,
and $\Fuj<\Fmax$ are possible, although the fraction of cases with
$\Fuj>\Fmax$ has increased substantially. Again, for the average
case, $\langle\Fuj\rangle>\langle\Fmax\rangle$. As expected, there
is a lower bound on the state purity, since since $\mathcal{P}\ge1/d$
in a $d-$dimensional density matrix. Here one has $d=10$, so $\mathcal{P}\ge0.1$
in a $10$ dimensional density matrix. The average fidelity is greatly
reduced in this case, since the larger Hilbert space dimension reduces
the probability that two random matrices will be similar. 
\begin{center}
\begin{figure}[h!]
\includegraphics[width=1\columnwidth]{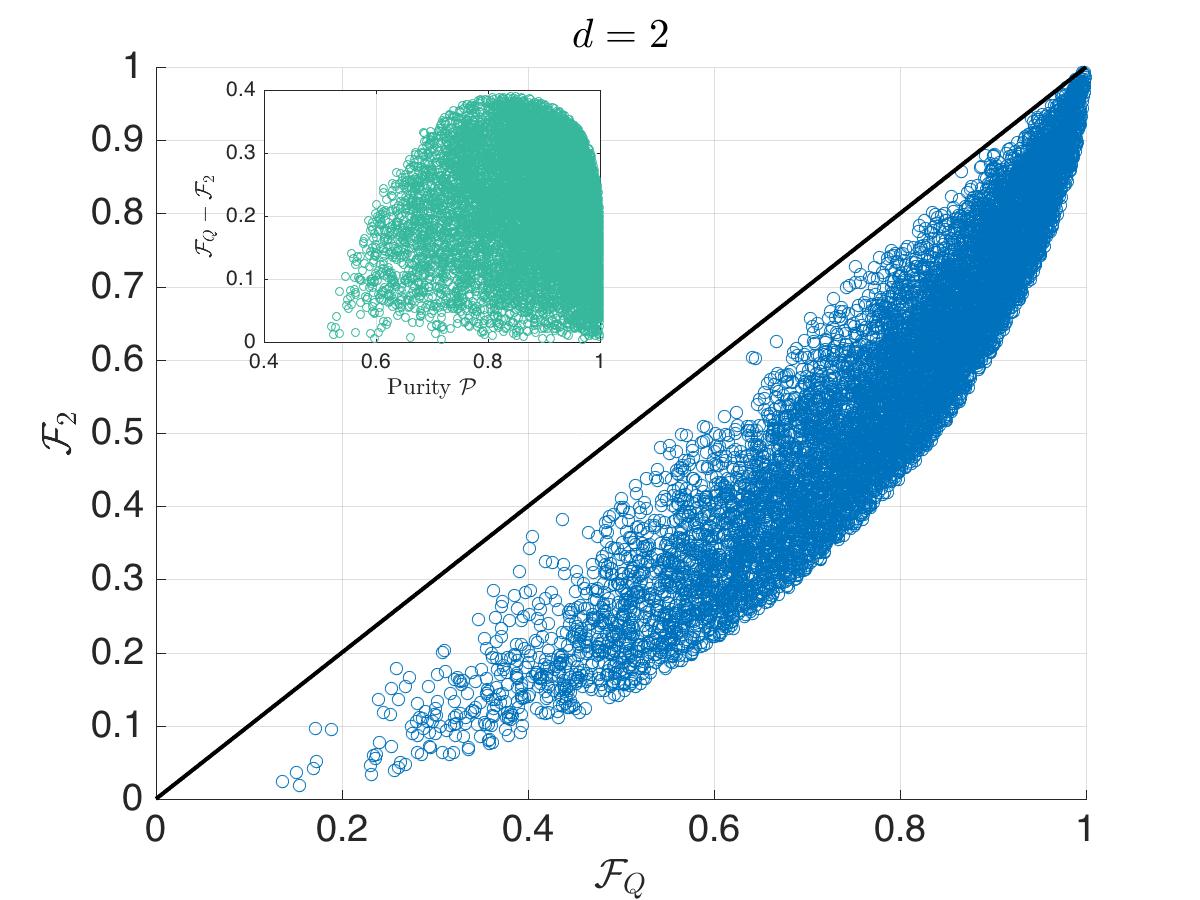} \caption{\label{Fig:F2-Fq:d2} Scatter plot showing the non-logarithmic variety
of the quantum Chernoff bound $\Fq$, Eq.~\eref{Eq:Fq}, and the
Hilbert-Schmidt fidelity $\Fmax$, Eq.~\eref{eq:Fmax}, for $10^{4}$
random pairs of single-qubit (mixed) density matrices. The black line
satisfying $\Fq=\Fmax$ is a guide for the eye. The inset shows a
scatter plot of the difference $\Fq-\Fmax$ vs the {\em maximum}
purity {$\mathcal{P}$} of these pairs of density matrices.}
\end{figure}
\par\end{center}

\begin{center}
\begin{figure}[h!]
\includegraphics[width=1\columnwidth]{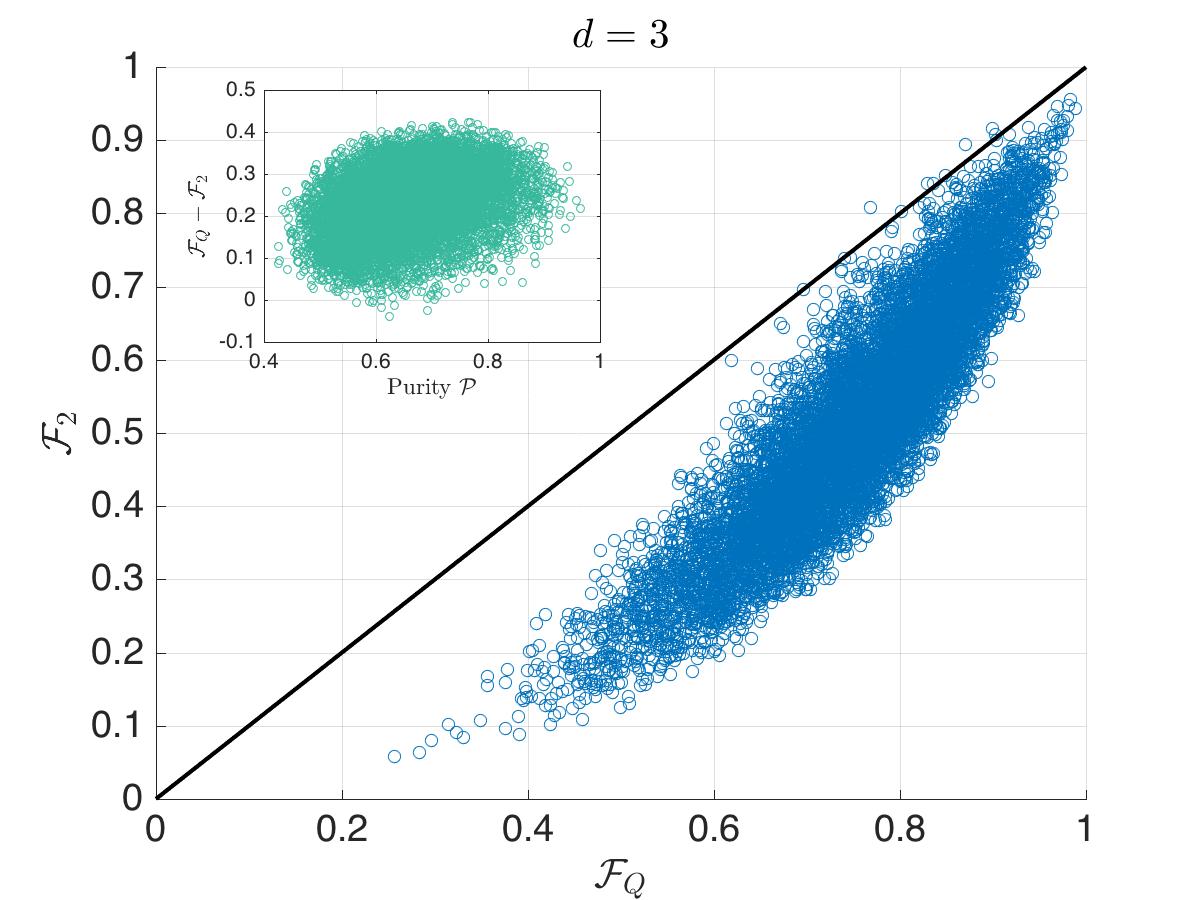} \caption{\label{Fig:F2-Fq:d3} Scatter plot showing the non-logarithmic variety
of the quantum Chernoff bound $\Fq$, Eq.~\eref{Eq:Fq}, and the
{Hilbert-Schmidt} fidelity $\Fmax$, Eq.~\eref{eq:Fmax}, {for
$10^{4}$ random pairs of single-qutrit (mixed) density matrices.}
The black line satisfying $\Fq=\Fmax$ is a guide for the eye. The
inset shows a scatter plot of the difference $\Fq-\Fmax$ vs the {\em
maximum} purity {$\mathcal{P}$} of these pairs of density matrices.}
\end{figure}
\par\end{center}

\begin{center}
\begin{figure}[h!]
\includegraphics[width=1\columnwidth]{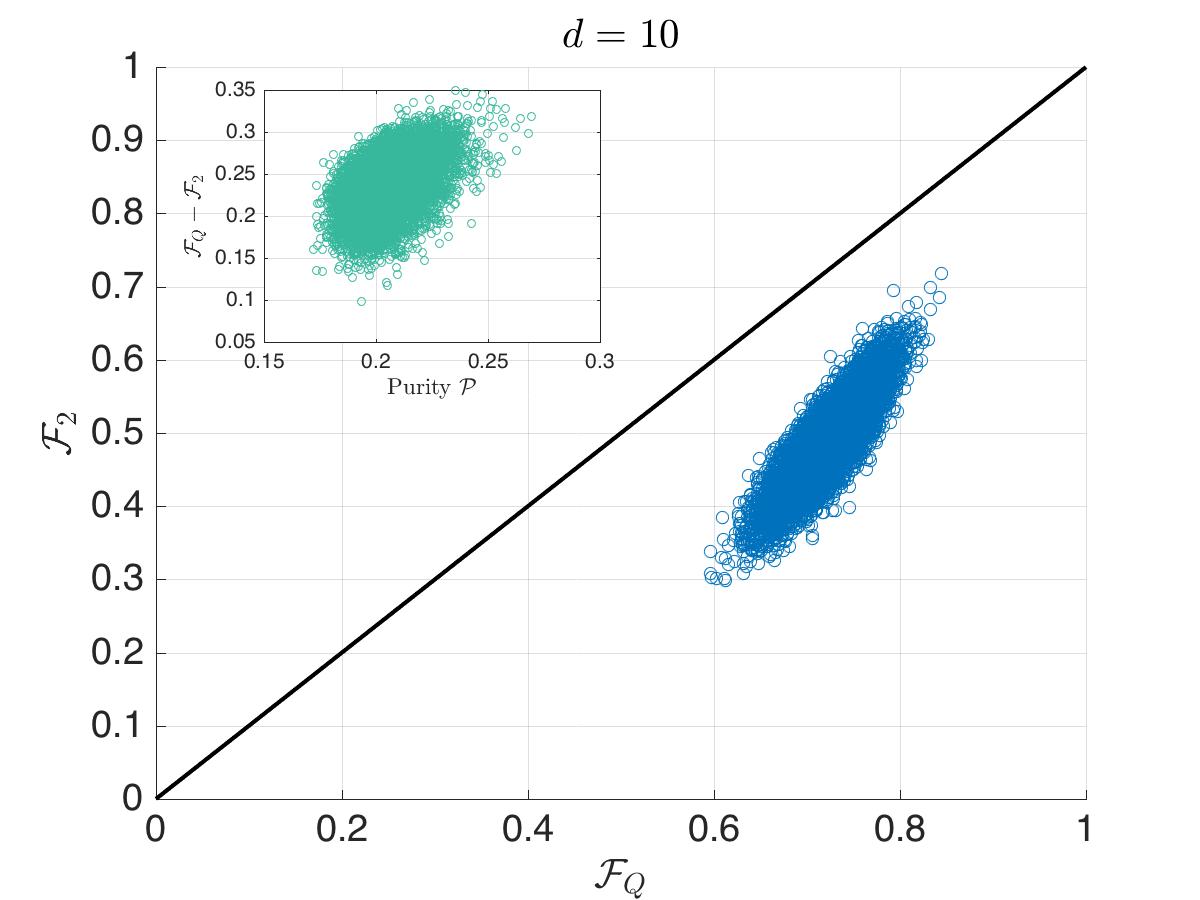} \caption{\label{Fig:F2-Fq:d10} Scatter plot showing the non-logarithmic variety
of the quantum Chernoff bound $\Fq$, Eq.~\eref{Eq:Fq}, and the
Hilbert-Schmidt fidelity $\Fmax$, Eq.~\eref{eq:Fmax}, for $10^{4}$
random pairs of single-qutrit (mixed) density matrices. The black
line satisfying $\Fq=\Fmax$ is a guide for the eye. The inset shows
a scatter plot of the difference $\Fq-\Fmax$ vs the {\em maximum}
purity {$\mathcal{P}$} of these pairs of density matrices.}
\end{figure}
\par\end{center}

The corresponding plots comparing $\Fmax$ and $\Fq$ can be found,
respectively, in Figure~\ref{Fig:F2-Fq:d2}, Figure~\ref{Fig:F2-Fq:d3}
and Figure~\ref{Fig:F2-Fq:d10}.

\section{Applications}

\subsection{Fidelity in quantum physics}

Most current applications of quantum fidelity take place in the context
of quantum technology. Since fidelity is a relative measure, the most
appropriate fidelity depends on the proposed application. These technologies
usually have a well-defined purpose. For example, one may use a quantum
technology like squeezing \cite{caves1981quantum,walls1983squeezed}
or Einstein-Poldosky-Rosen (EPR) correlation~\cite{EPR}/steering~\cite{wiseman2007,reid2009colloquium,ma2017proposal}
to measure gravitational waves \cite{abbott2016observation,einstein1918gravitationswellen}
in a more sensitive way. As well as precision metrology, other applications
include enhanced communications \cite{caves1994quantum}, cryptography
\cite{gisin2002quantum}, quantum computing \cite{steane1998quantum},
many-body quantum simulators \cite{fialko2015fate,bernien2017probing},
quantum data processing and storage \cite{schumacher1996quantum},
and the growing area of quantum thermodynamics \cite{millen2016perspective}.

Each of these fields has their own criteria for success. However,
the components of the technology ultimately depend on the realization
of certain quantum states and their processing. Hence, knowledge of
fidelity can help to measure how close one is to achieving the required
quantum states that are utilized at a given stage in a quantum process~\cite{nielsen2000quantum},
see, e.g.,~\cite{Yang2014,Kaniewski2016,Sekatski1802} and references
therein. This concept is also applicable more widely, in fundamental
physics problems like quantum phase transitions.

A fidelity measure, as with some other (operational) figures of merit
(see, e.g.,~\cite{Chen:PRL:2016,Cavalcanti:PRA:2016,Rosset:PRX:2018}),
is sometimes~\cite{Sekatski1802} used to prove that a device is
quantum in nature, rather than operating at a simply classical level.
More generally, one can analyze quantum operations in terms of process
fidelity and average fidelity, as explained below.

\subsection{Average fidelity in quantum processes with pure input states}

\label{Sec:PureVsMixed}

Jozsa's axioms do not lead to a unique function of two density matrices
that generalizes Eq.~\eref{Eq:Fidelity:1mixed} to a pair of mixed
states. In this section, we would like to examine a desirable feature
of a fidelity measure that arises naturally in a quantum communication
scenario \cite{caves1994quantum}, although this concept also arises
in other types of quantum processes. This also allows us to relate
a fidelity measure as a measure of the degree of similarity between
two quantum states to the notion of fidelity originally introduced
by Schumacher~\cite{schumacher1995quantum}. We focus especially
on the three candidate fidelities that satisfy the axioms, and have
a physical interpretation: $\Fuj$, $\Fq$ and $\F_{2}$.

We note that the uniform average fidelity over all possible pure input
states is often termed ${\cal F}_{{\rm ave}}$. Like the more complex
process fidelity, this can be used as a means of analyzing performance
of quantum logic gates~\cite{knill2008randomized,harty2014high}.
Here we allow ${\cal F}_{{\rm ave}}$ to have non-uniform initial
probabilities, and we note that, for pure states, the fidelities being
averaged are not dependent on the \emph{choice} of fidelity measure,
since these are unique {[}by Jozsa's axiom (J3){]} if one of the density
matrices being compared is a pure state.

Suppose we attempt to copy, store, perform idealized logic operations
or transmit a number of pure states $\rho_{j}$, each occurring with
probability $p_{j}$. Imagine further that the pure states $\rho_{j}$,
are the desired output of the logic operations (communication protocol).
Evidently, this combination of states can also be described by the
mixed state $\rho=\sum_{j}p_{j}\rho_{j}$. Now, let us further imagine
that during the physical processes of interest, there is a probability
of error $\epsilon$, in which case a random error state $\rho_{0}$
is generated. For simplicity, we shall assume (unless otherwise stated)
that $\rho_{0}$ is orthogonal to $\rho_{j}$, for all $j\neq0$.

A physical scenario that matches the above description consists of
transmitting photon number states through some quantum channel, but
with probability $\epsilon$, the Fock state $\rho_{j}=\ket{j}\!\!\bra{j}$
gets transformed to the vacuum state $\ket{0}\!\!\bra{0}$, or some
other undesirable Fock state that corresponds to an error. However,
we do not assume here that the inputs are necessarily orthogonal,
as there are applications in cryptography~\cite{gisin2002quantum}
where non-orthogonality serves as an important part of the communication
protocol.

Let us denote the output state for the transmission of $\rho_{j}$
as $\sigma_{j}=\epsilon\rho_{0}+(1-\epsilon)\rho_{j}$. The overall
combination of the output states can also be described by the following
mixed state: 
\begin{equation}
\sigma=\sum_{j}p_{j}\sigma_{j}=\epsilon\rho_{0}+(1-\epsilon)\rho.
\end{equation}

The mixed state fidelity can now be examined from two points of view:
the {\em average state-by-state fidelity}, and also the {\em
fidelity as a mixed state}. For a sensible generalization of Schumacher's
fidelity to a pair of mixed states, it seems reasonable to require
that the average state-by-state fidelity corresponds exactly to the
mixed state fidelity, at least, under some circumstances.

To this end, let us first define the average state-by-state fidelity
{as} a function of the probabilities and states: 
\begin{equation}
\Fave\left(\bm{p},\bm{\rho},\bm{\sigma}\right)=\sum_{j}p_{j}\F(\rho_{j},\sigma_{j}).
\end{equation}

Since the inputs are individually pure states, the fidelity is just
the state overlap, and one then obtains:

\begin{equation}
\Fave\left(\bm{p},\bm{\rho},\bm{\sigma}\right)=\sum_{j}p_{j}\tr(\rho_{j}\sigma_{j})=1-\epsilon.
\end{equation}
In the case of $\F_{2}$, the average fidelity matches the mixed fidelity,
so that 
\begin{equation}
\Fave\left(\bm{p},\bm{\rho},\bm{\sigma}\right)=\F(\rho,\sigma).\label{Eq:AveVsMixed}
\end{equation}
whenever the average signal state $\rho$ has equal to or larger purity
than that of the average output state $\sigma$, i.e., $\tr(\sigma^{2})\le\tr(\rho^{2})$.
This is a common situation, although not universal.

If all (input) states including the error state are orthogonal to
each other, and hence diagonal in the same basis, it is easy to verify
that the Uhlmann-Jozsa fidelity $\Fuj$, the non-logarithmic variety
of the quantum Chernoff bound, $\Fq$, as well as the $A$-fidelity
$\Fa$ comply with the desired requirement of Eq.~\eref{Eq:AveVsMixed},
without any additional conditions. For a proof, see~\ref{Sec:Proofs}.

It is also not difficult to check that if $\rho$ is a pure state,
$\Fn$ also satisfies the desired requirement; although it is a rare
situation that a pure state would be used for communications in this
way, owing to {the} extremely small information content. As for
$\Fc$, $\Fam$ and $\Fgm$, there does not appear to be a generic
situation where Eq.~\eref{Eq:AveVsMixed} will hold true for all
error probabilities $0<\epsilon<1$.

Let us analyze briefly the case where average output state has a greater
purity than the signal state. At first this seems unlikely, but in
fact it can occur. If the error state $\rho_{0}$ is a pure state,
and this occurs with probability one, then the output is always a
pure state. This could be a vacuum state, owing to an extremely serious
error condition: the channel being completely absorbing. Under these
conditions, one has that $\Fave=0$, which is expected: every input
of the alphabet results in an incorrect output. This is true of any
fidelity measure that satisfies our generalized Josza axiom list,
due to the orthogonality condition, and hence it is true of $\Fuj$,
$\Fq$ and $\F_{2}$.

Finally, we note that the most general error state $\rho_{0}$ depends
on the input state $\rho_{j}$, and is neither pure nor orthogonal
to all other inputs. Under these circumstances, $\Fave\left(\bm{p},\bm{\rho},\bm{\sigma}\right)$
is still well-defined but it may not correspond to any of the mixed
fidelities defined here.

\subsection{Average fidelity in quantum processes with mixed input states}

\label{Sec:PureVsMixed-1}

A quantum process with pure input states is an ideal scenario that
is unlikely to occur in reality. As well as the outputs, even the
inputs are most likely to be mixed states, as they are typically multi-mode
and will be subject to timing jitter, losses and spontaneous emission
noise, to name just a few possible error conditions.

In this case, we may have a scenario like the one given above, except
that all the states are mixed. We again assume for simplicity that
errors in the output channel are {\em orthogonal} to all inputs,
as before. The average state-by-state fidelity is now a function of
the probabilities, states and the fidelity measure: 
\begin{equation}
\Fave\left(\bm{p},\bm{\rho},\bm{\sigma}\right)=\sum_{j}p_{j}\F(\rho_{j},\sigma_{j}).
\end{equation}
In general, for $\Fuj$, $\Fq$, and $\Fa$, the situation becomes
complex, as there is no guarantee that all the states involved can
be diagonalized in the same basis. If they do, we will again end up
with Eq.~\eref{Eq:AveVsMixed}. This makes an evaluation of the
corresponding fidelities generally difficult, since, for example,
the products of matrix square roots that occur in $\Fuj$ are extremely
nontrivial in general.

As for $\F_{2}$, a reduction holds if the individual input states
$\rho_{j}$ and the average state $\rho$ has a degraded purity owing
to errors {{[}i.e., $\tr(\rho_{j}^{2})\ge\tr(\sigma_{j}^{2})$ and
$\tr(\rho^{2})\ge\tr(\sigma^{2})${]}}, although, as noted above,
this may not be true for an extremely lossy channel. With this simplification
the fidelity for each possible input is the same:

\begin{equation}
\F_{2}(\rho_{j},\sigma_{j})=\frac{\tr(\rho_{j}\sigma_{j})}{\tr(\rho_{j}^{2})}=1-\epsilon.\label{Eq:F2:State-by-state}
\end{equation}
In such cases, we have the same result for the average fidelity as
for pure state inputs, namely: 
\begin{equation}
\F_{{\rm ave}}\left(\bm{p},\bm{\rho},\bm{\sigma}\right)=\sum_{j}p_{j}\F_{2}(\rho_{j}\sigma_{j})=1-\epsilon=\F_{2}(\rho,\sigma).\label{Eq:F2:compare}
\end{equation}

As with the pure state case, we emphasize that the average fidelity
is therefore only the same as the mixed state fidelity in rather special
circumstances, even when using the simpler $\F_{2}$ measure of fidelity.
In general, the assumption that the error state is orthogonal to all
inputs is not always valid, and as a result they measure different
properties of the quantum process. Note also that the requirement
of Eq.~\eref{Eq:AveVsMixed} amounts to demanding that the joint
concavity inequality of Eq.~\eref{Eq:ConcaveJoint} is saturated
in this communication scenario.

\subsection{Teleportation and cloning fidelity}

As a measure of the degree of similarity between two quantum states,
fidelity occurs naturally in assessing the quality of a quantum communication~\cite{schumacher1995quantum}
channel. A well-known example of such a channel is the so-called teleportation
channel~\cite{bennett1993,horodecki1999} where the \emph{unknown}
quantum state of a physical system is transferred from a sender to
a receiver via the help of a shared quantum resource between the two
ends. If one could transport the system itself directly to the receiver
intact, this task is trivial. However, since channels are usually
far from ideal\textemdash they may be lossy, for example, especially
over long distances\textemdash the quantum states are readily degraded
on transmission. Moreover, an \emph{unknown} quantum information cannot
be cloned~\cite{nocloningtheorem}, or else measured and regenerated,
without degradation. Thus, whenever a high-fidelity output state is
required at the end-location (for example, a qubit that is to be part
of a quantum secure network), a direct transmission of the quantum
signal would generally not be the preferred option. Instead, teleportation
is likely to serve as an essential component of future quantum communication
networks \cite{quantum-internet}.

A protocol for the \emph{ideal} teleportation of qubit states was
first proposed by Bennett {\em et al.} \cite{bennett1993} and
extensively discussed in a very general setting by Horodecki~\emph{et
al.}~\cite{horodecki1999}. Suppose Alice has in her possession an
arbitrary qudit state $|\psi_{i}\rangle$ that she wants to teleport
to Bob. The protocol involves Alice and Bob sharing an EPR resource,
or more precisely a maximally entangled two-qudit state.

Alice makes a joint measurement (in the basis of maximally entangled
states) on her half of the entangled qudit and the state to be teleported.
She transmits via a classical communication channel those measurement
results to Bob, who reconstructs the original state as $|\psi_{f}\rangle$
by applying a local unitary correction to his half of the maximally
entangled two-qudit state.

Fidelity between the initial state and the final state is the standard
figure of merit used to measure the effectiveness of the teleportation
transfer~\cite{popescu1994}. Specifically, the quality of the teleportation
channel $\Lambda_{\rho}$ is often quantified via the \emph{average}
fidelity~\cite{horodecki1999}: 
\begin{equation}
\Ftele(\Lambda_{\rho})=\int d\psi_{i}\bra{\psi_{i}}\Lambda_{\rho}(\proj{\psi_{i}}){\ket{\psi_{i}}}
\end{equation}
where $\rho$ is the density matrix of the shared resource and the
integral is performed over the Haar measure (i.e., a uniform distribution
over all possible pure input states of a \emph{fixed} Hilbert space
dimension). Not surprisingly, $\Ftele(\Lambda_{\rho})$ depends on
the degree of entanglement of $\rho$. For instance, if $\rho$ is
maximally-entangled, one has an ideal teleportation channel and thus
$\Ftele=1$. In general, for a given resource state $\rho$ of local
dimension $d$, the maximal average fidelity achievable was found
to be~\cite{horodecki1999} 
\begin{equation}
\Ftele^{{\rm max}}(\Lambda_{\rho})=\frac{{\F_{{\rm max}}}d+1}{d+1}\label{Eq:FteleMax}
\end{equation}
where $\F_{{\rm max}}$ is the singlet fraction (or more appropriately,
the fully-entangled fraction) of $\rho$, which is the largest possible
overlap (i.e., Schumacher's fidelity) between $\rho$ and a maximally
entangled two-qudit state: 
\begin{equation}
\F_{{\rm max}}=\max_{U_{A},U_{B}}\F(\rho,U_{A}\otimes U_{B}\proj{\Phi_{d}^{+}}U_{A}^{\dag}\otimes U_{B}^{\dag}).
\end{equation}
Here $U_{A}$ and $U_{B}$ are, respectively, local unitary operator
acting on the Hilbert space of Alice's and Bob's qudit.

\subsubsection{Teleportation fidelity and bounds}

The problem of establishing the quality of a practical teleportation
experimental process then becomes the problem of a fidelity measurement.
To this end, it is insightful to compare the fidelity bound of Eq.~\eref{Eq:FteleMax}
against that arising from employing an optimal cloning machine~\cite{hillery-buzek-clone,classical-fidelity,Bruss1998}.
For simplicity, we shall restrict our discussion to the $d=2$, i.e.,
the qubit case. For general fidelity bounds of cloning, we refer the
reader to~\cite{ScaraniRMP.77.1225}. In the case of an arbitrary
qubit, it is known that a classical strategy of teleportation, whereby
the state is measured and then regenerated, will incur extra noise
(see, e.g.,~\cite{popescu1994,horodecki1999,ham-1,fidelity-bounds}).
This limits the fidelity for any classical measure-and-regenerate
protocol, to ${\Ftele^{\tiny{\rm class}}}\leq\frac{2}{3}$~\cite{popescu1994,horodecki1999,classical-fidelity}.
It is also known that the fidelity $\F=\frac{5}{6}$ is the maximum
for any symmetric $1\to2$ cloning process, i.e., $\F>\frac{5}{6}$
ensures that there can be no ``copies" $|\psi_{c}\rangle^{\otimes2}$
taken of the state $|\psi_{i}\rangle$ that has a fidelity $\F=|\langle\psi_{i}|\psi_{c}\rangle|^{2}$
greater than $\frac{5}{6}$ \cite{hillery-buzek-clone,hillerybuzek-2,dag}.
The \emph{ideal} teleportation protocol clearly exceeds this bound
because the state $|\psi_{i}\rangle$ at Alice's location is destroyed
by her measurements. The final teleported state at Bob's location
is therefore not a copy of the state $|\psi_{i}\rangle$, but a unique
secure transfer of it. In view of these limits, a fidelity $\Ftele>\frac{2}{3}$
is the benchmark figure of merit used to justify quantum teleportation
in qubit teleportation experiments \cite{teleexp,teleexpdemartini}.
The no-cloning teleportation is achieved when $\Ftele>\frac{5}{6}$.
For qubit-teleportation experiments carried out photonically, the
fidelity estimates were evaluated for a post-selected subensemble,
selected conditionally on all photons being detected at Bob's location
\cite{teleexp} (see also~\cite{Wang:2015aa}). The post-selection
was required due to poor detection efficiencies.

\subsubsection{Continuous variable teleportation}

A protocol for continuous variable teleportation was developed by
Vaidmann \cite{vaidcvtele} and Braunstein and Kimble \cite{bkcvtele}.
Here, the EPR resource was a continuous-variable EPR entangled state
\cite{reid1989demonstration,reid2009colloquium}. Defining ``quantum
teleportation'' as taking place where there can be no classical measure-and-regenerate
strategy that can replicate the teleportation fidelity $\Ftele$,
it was shown that $\Ftele>1/2$ is the benchmark fidelity bound to
demonstrate the quantum teleportation of an unknown coherent state
\cite{ham-1}. This fidelity has been achieved experimentally for
optical states \cite{cvteleexperiments,cvteleexperiments-1,cvteleexperiments-2,cvtelenocloning}.
The fidelity at which there can be no replica of the state at a location
different to Bob's corresponds to the no-cloning fidelity \cite{clonein}.
For coherent states, the fidelity $\F=2/3$ is the maximum for any
cloning process \cite{cerf}, and the fidelity $\Ftele>2/3$ is hence
the no-cloning benchmark for the teleportation of a coherent state.
The no-cloning teleportation limit has been achieved for coherent
states \cite{cvtelenocloning}. Although reporting lower teleportation
fidelities, the continuous variable experiments did not rely on {\em
any} post-selection of data.

\subsubsection{Connection with other desired properties}

Fidelity is the figure of merit used to quantify the success and usefulness
of the teleportation protocol. For example, no-cloning teleportation
allows bounds to be placed on the quality of any unwanted copies of
the teleported state. As a consequence, fidelity is also used to determine
the constraints on the nature of the entangled resource, so that certain
conditions are met. For example, by using the fidelity bounds, it
was proved~\cite{mixed-state-tele-1} that any entangled two-qubit
(mixed) state is useful for quantum teleportation if arbitrary local
operations assisted by classical communications (including arbitrary
local filtering operations) are allowed prior to the teleportation
experiment. Similarly, a connection has been shown, between no-cloning
teleportation and the requirement of a steerable resource \cite{epr-steer-tele},
between the fully-entangled fraction and a steerable~\cite{wiseman2007}
resource, as well as a Bell-nonlocal~\cite{brunner-rmp} resource,
see, e.g.,~\cite{hsieh2016} and references therein.

In the context of the papers examined for this review, the teleportation
fidelity is most frequently defined with respect to a pure state that
Alice wants to transport. Bob's teleported state will generally be
mixed. However, one can see that more generally the fidelity for the
teleportation of a mixed state needs also to be considered, given
that the state prepared at Alice's location would not usually be pure.

\subsection{Fidelity in phase space}

\label{Sec:PhaseSpace}

In the following, we show how the fidelity can be computed with phase
space methods \cite{Hillery_Review_1984_DistributionFunctions}. The
symmetrically ordered Wigner function representation \cite{Wigner_1932}
was first applied to dynamical problems by Moyal \cite{Moyal_1949}.
Although it is is generally non-positive, it is common to use tomography
to measure the Wigner function \cite{lvovsky2009continuous}, to represent
a density matrix. Other schemes using a classical-like phase-space
for bosons include the anti-normally ordered, positive Q-function
distribution \cite{Husimi1940}, and the normally-ordered P-function
distribution \cite{Glauber_1963_P-Rep}, which is non-positive and
singular in some cases.

These have been generalized to positive distributions on non-classical
phase spaces \cite{Drummond1980posp,Gilchrist1997posp,Deuar:2002},
which are normally-ordered, non-singular, positive representations.
These have been employed in many different fields such as quantum
optics \cite{drummond1981_II_nonequilibriumparamp,DrummondGardinerWalls1981,Drummond_EPL_1993},
Bose-Einstein condensates \cite{Kheruntsyan2005BEC,Opanchuk2012BEC,Opanchuk2013WignerBEC}
and quantum opto-mechanical systems \cite{Kiesewetter2014opto,Kiesewetter2017opto,Teh2017opto},
as well as spin \cite{Arecchi_SUN,Agarwal:1981,Barry_PD_qubit_SU}
and Fermi \cite{corney2006gaussian} systems.

In these methods, a density operator $\hat{\rho}$ is generically
represented as 
\begin{equation}
\hat{\rho}=\intop P\left(\vec{\alpha}\right)\hat{\Lambda}\left(\vec{\alpha}\right)\,d\vec{\alpha}\,,\label{eq:density_op}
\end{equation}
where $\hat{\Lambda}\left(\vec{\alpha}\right)$ is a projection operator
that forms the basis in the description of a density operator and
$P\left(\vec{\alpha}\right)$ is the quasi-probability density that
corresponds to that operator basis. In the simplest cases, $\vec{\alpha}=\bm{\alpha}=\left(\alpha_{1},...,\,\alpha_{M}\right)$
is a real or complex vector in the relevant $M$-mode phase-space
for the first phase-space representations defined using a classical
phase-space. This can have increased dimensionality in more recent
mappings.

Using Eq. (\ref{eq:density_op}), we can immediately see how the $\F_{2}$
fidelity can be computed. In particular, we show explicitly how the
quantity $\tr\left(\rho\,\sigma\right)$ is obtained. 
\begin{eqnarray}
\tr(\rho\,\sigma) & =\tr\left[\intop\intop P_{\rho}\left(\vec{\alpha}\right)P_{\sigma}\left(\vec{\beta}\right)\hat{\Lambda}\left(\vec{\alpha}\right)\hat{\Lambda}\left(\vec{\beta}\right)\,d\vec{\alpha}\,d\vec{\beta}\right]\nonumber \\
 & =\intop\!\!\intop P_{\rho}\left(\vec{\alpha}\right)P_{\sigma}\left(\vec{\beta}\right)D(\vec{\alpha},\vec{\beta})\,d\vec{\alpha}\,d\vec{\beta}\,.\label{eq:fidelity}
\end{eqnarray}

Depending on the particular phase space representation, $D(\vec{\alpha},\vec{\beta})\equiv\tr\left[\hat{\Lambda}\left(\vec{\alpha}\right)\hat{\Lambda}\left(\vec{\beta}\right)\right]$
in Eq. (\ref{eq:fidelity}) takes different forms. An expression for
this quantity was calculated by Cahill and Glauber \cite{Cahill1969},
for classical phase-space methods, and it is generally only well-behaved
for the Wigner and P-function methods. The Glauber P-function, although
often singular, was proposed as an approximate method in fidelity
tomography using this approach \cite{lobino2008complete}.

Here we focus on two of the most useful representations in a typical
phase space numerical simulation: the Wigner and positive P-representations.
The first of these is a rather classical-like mapping, although generally
not positive-definite, while the second is always probabilistic, although
defined on a phase-space that doubles the classical dimensionality.
In the Wigner representation, $\vec{\alpha}\equiv\bm{\alpha}$, and
\begin{eqnarray}
D(\bm{\alpha},\bm{\beta}) & =\pi^{M}\delta^{M}\left(\bm{\alpha}-\bm{\beta}\right)\,,\label{eq:Tr_basis}
\end{eqnarray}
where $M$ is the dimension of the vector $\bm{\alpha}$ and the $M$-th
dimensional Dirac delta function is $\delta^{M}\left(\bm{\alpha}-\bm{\beta}\right)=\delta\left(\alpha_{1}-\beta_{1}\right)...\delta\left(\alpha_{M}-\beta_{M}\right)$.
This leads to~\cite{Cahill1969} 
\begin{equation}
\tr(\rho\,\sigma)=\pi^{M}\intop P_{\rho}\left(\bm{\alpha}\right)P_{\sigma}\left(\bm{\alpha}\right)\,d\bm{\alpha}\,.\label{eq:fidelity_wigner}
\end{equation}

In the positive P representation, a single mode is characterized by
two complex numbers, so $\vec{\alpha}=\left(\alpha,\alpha^{+}\right)$.
The notation $\alpha^{+}$ indicates that this variable represents
a conjugate operator, and is stochastically complex conjugate to $\alpha$
in the mean. This doubles the dimension of the relevant classical
phase space. There is an intuitive interpretation that it allows one
to map superpositions directly into a phase-space representation.
The density operator in positive P representation is then given by
Eq.~\eref{eq:densityop}, but now the operator bases have the form:
\begin{equation}
\hat{\Lambda}\left(\vec{\alpha}\right)=\frac{|\bm{\alpha}\rangle\langle\bm{\alpha^{+}}|}{\langle\bm{\alpha^{+\textnormal{*}}}|\bm{\alpha}\rangle}\,,
\end{equation}
and the product trace is: 
\begin{eqnarray}
D(\vec{\alpha},\vec{\beta}) & =\frac{\langle\bm{\beta^{+}}|\bm{\alpha}\rangle\langle\bm{\alpha^{+}}|\bm{\beta}\rangle}{\langle\bm{\alpha^{+\textnormal{*}}}|\bm{\alpha}\rangle\langle\bm{\beta^{+\textnormal{*}}}|\bm{\beta}\rangle}\,.\label{eq:Tr_basis_posp}
\end{eqnarray}
The quantity $\tr(\rho\,\sigma)$ is more complicated in this case
but still follows the structure of Eq.~\eref{eq:fidelity}. Fidelity
measures of the form given in Eq.~\eref{eq:Fmax} also involve
the purity, $\tr\left(\rho^{2}\right)$. This can be calculated similarly.

One great advantage of phase space methods is that quasi-probability
densities allow numerical simulation to be carried out. Typically,
samples are drawn from these probability densities, which are then
evolved dynamically. Finally, observables of interest are computed
by the Monte Carlo method, which is usually the only practical technique
for very large Hilbert spaces. Likewise, fidelity measures can be
computed numerically in a typical Monte Carlo scheme. In particular,
we consider $\mathcal{F}_{2}$, which, as we will discuss, is the
most tractable form of fidelity.

{Let $\rho$ be} the initial density operator of the system and
we want to compute the fidelity of a quantum state at a later time,
which is characterized by the density operator $\sigma$, with respect
to $\rho$. {Suppose that the initial state $\rho$} and its corresponding
phase space distribution are known in a numerical simulation.

{For simplicity, suppose that $\rho$ is a pure state, which is a
common but not essential assumption,} and implies that $\tr\left(\rho^{2}\right)=1$.
Even for cases where $\tr\left(\rho^{2}\right)<1$, the final state
after a time evolution will usually be no purer than the initial state.
There are exceptions to this rule, since a dissipative time-evolution
can evolve a mixed state of many particles to a pure vacuum state,
but we first consider the case of non-increasing purity here for definiteness.

In other words, ${\rm max}\left[\tr\left(\rho^{2}\right),\tr\left(\sigma^{2}\right)\right]=\tr\left(\rho^{2}\right)$.
This is convenient as the exact (quasi)-probability distribution for
$\sigma$ is not known and only a set of samples of this distribution
is available, which leads to the sampled fidelity we discuss next.

Next, consider the sampled fidelity in the Wigner representation.
The quantity $\tr\left(\rho\,\sigma\right)$ in Eq. (\ref{eq:fidelity_wigner})
in the Monte Carlo scheme is given by: 
\begin{eqnarray}
\tr\left(\rho\,\sigma\right) & =\pi^{M}\intop P_{\rho}\left(\bm{\alpha}\right)P_{\sigma}\left(\bm{\alpha}\right)\,d\bm{\alpha}\nonumber \\
 & \approx\pi^{M}\frac{1}{N_{{\rm samples}}}\sum_{i=1}^{N_{{\rm samples}}}P_{\rho}\left(\bm{\alpha}_{i}\right)\,,\label{eq:tr_rho_sig_wigner}
\end{eqnarray}
where $\Ns$ is the sample size of the probability distribution $P_{\sigma}$.
We note that there can be issues with the fact that the same random
variable occurs in both the distributions, leading to practical problems
if both the distributions are sampled. This can be avoided if one
of the Wigner functions is known analytically.

The same quantity can be computed in the positive P representation.
It is then possible to use two independent sets of random variables,
so that both the distributions can be obtained from random sampling:
\begin{eqnarray}
\tr\left(\rho\,\sigma\right) & \approx & \frac{1}{N_{{\rm samples}}^{2}}\sum_{i,j}^{N_{{\rm samples}}}D(\vec{\alpha}_{i},\vec{\beta}_{j})\,.\label{eq:tr_rho_sig_posp}
\end{eqnarray}
Here, the factor $N_{{\rm samples}}^{2}$ comes from the product of
$P_{\rho}\left(\vec{\alpha}\right)P_{\sigma}\left(\vec{\beta}\right)$
in Eq.~\eref{eq:fidelity} under the usual assumption of equally
weighted samples. This shows that it is possible to compute $\mathcal{F}_{2}$
fidelities from a phase-space simulation. This is useful when trying
to predict performance of a quantum technology or memory in an application
involving storage of an exotic quantum state. We emphasize that if
one of the calculated states is a pure state, then all fidelity measures
give the same result.

Admittedly, this quantity is more complicated than Eq.~(\ref{eq:tr_rho_sig_wigner})
in the Wigner representation. In addition, the sampling error can
be very large in some cases, as discussed by Rosales-Zarate and Drummond
\cite{Rosales-Zarate2011entropy}. When this occurs, representations
such as the generalized Gaussian representations \cite{corney2003gaussian,corney2005gaussian,corney2006gaussian,joseph2018phase}
can be employed, and clearly the purities can be estimated in a similar
way if the initial state is not pure.

Overall, $\mathcal{F}_{2}$ appears to be the most suitable fidelity
measure in a dynamical simulation or measurement using phase-space
techniques, where only the initial state with its probability distribution
is known. It is the only measure using easily computable Hilbert-Schmidt
norms that satisfies all of the Jozsa axioms.

%dummy comment inserted by tex2lyx to ensure that this paragraph is not empty%dummy comment inserted by tex2lyx to ensure that this paragraph is not empty%dummy comment inserted by tex2lyx to ensure that this paragraph is not empty

\subsection{Techniques of fidelity measurement}

As pointed out in Section (\ref{subsec:Relevant-and-irrelevant}),
fidelity is a relative measure. The results of fidelity measurements
on different Hilbert spaces are \textbf{not }the same. The ${\mathcal{F}_{1}}$
fidelity may improve if the measured Hilbert space has a lower dimensionality,
simply because it is defined as a maximum over all possible purifications.
Hence, the ${\mathcal{F}_{2}}$ fidelity has the advantage that it
is generally less biased towards high values, as shown in the examples
of the previous section, and is always true for qubits.

The adage of being cautious in comparing apples to oranges should
be remembered. In general terms we will distinguish six different
approaches that are described below, as applied to typical physical
implementations \cite{nielsen2000quantum,divincenzo2000physical}.
The real utility of a given fidelity measure is how well it matches
the requirements of a given application. Analyzing quantum logic gates
and memories is one of the most widespread and useful applications
of fidelity, and hence we will give examples of these applications.

These general considerations about physical implementation apply to
all of the various applications listed in this section. The examples
referenced here are necessarily incomplete, as this is not a full
review of experimental implementations. Nevertheless, some typical
experimental measurements are referenced in each of the following
application examples.

\subsubsection{Atomic tomography fidelity }

Atomic or ionic fidelity measurements involve a finite, stationary,
closed quantum system, where each state can be accessed and projected
\cite{leibfried2003experimental,longdell2004experimental}. These
measurements are usually relatively simple. To obtain the entire density
matrix for a calculation of mixed state fidelity involves a tomographic
measurement. Thus, for example, in qubit tomography one must measure
both the diagonal elements, which are level occupations, as well as
off-diagonal elements that are obtained through Rabi rotations that
transform off-diagonal elements into level occupations for measurement.
This approach is easiest to implement when the Hilbert space has only
two or three levels. Typical examples of this technique involve trapped
ions, whose level occupations are measured using laser pulses and
fluorescence photo-detection.

It is not always clear in such measurements how the translational
state is measured, or if it is even part of the relevant Hilbert space,
which is necessary in order to understand how the fidelity is defined.
The problem is that the full quantum state of an isolated ion or atom
always has both internal and center-of-mass degrees of freedom, so
that a pure state is: 
\begin{equation}
\left|\Psi\right\rangle =\sum_{ij}C_{ij}\left|\psi_{i}\right\rangle _{{\rm int}}\left|\phi_{j}\right\rangle _{{\rm CM}}.
\end{equation}
Here $\left|\psi\right\rangle _{{\rm int}}$ is the internal state
defined by the level structure, while $\left|\phi\right\rangle _{{\rm CM}}$
defines the center-of-mass degree of freedom. The actual density matrix
even for a single ion or atom therefore involves different spatial
modes for the center-of-mass, such that: 
\begin{equation}
\rho=\sum_{ijkl}\rho_{ijkl}\left|\psi_{i}\right\rangle _{{\rm int}}\left|\phi_{j}\right\rangle _{{\rm CM}}\left\langle \phi_{k}\right|_{{\rm CM}}\left\langle \psi_{l}\right|_{{\rm int}}.
\end{equation}
Next, we can consider two possible situations: 
\begin{description}
\item [{Full~tomography:}] Suppose that the target state is $\left|\Psi_{0}\right\rangle =\left|\psi_{0}\right\rangle _{{\rm int}}\left|\phi_{0}\right\rangle _{{\rm CM}}$,
so that the center-of-mass position is part of the relevant Hilbert
space. Under these conditions, only the measured states with $k=j=0$
are in the same \emph{overall} quantum state as the target state.
This may prevent complete visibility in an interference measurement
in which the center-of-mass position is relevant. If this is the case,
one should consider the center-of-mass position as part of the relevant
Hilbert space. Hence one must consider rather carefully if it is necessary
to investigate the translational state fidelity as well in this type
of application. This has been investigated in ion-trap quantum computer
gate fidelity measurements \cite{ospelkaus2011microwave}. 
\item [{Partial~tomography:}] The center-of-mass part of the Hilbert space
may not matter if the internal degrees of freedom are decoupled sufficiently
from the spatial degrees of freedom, so that only the internal degrees
of freedom are relevant over the time-scales that are of interest.
In these cases the density matrix can be written as: 
\begin{equation}
\rho=\rho_{{\rm int}}\otimes\rho_{{\rm CM}}.
\end{equation}
Provided this factorization is maintained throughout the experiment,
it may well be enough to only measure the internal degrees of freedom.
However, any spin-orbit or similar effective force that couples the
internal and translational degrees of freedom will cause entanglement.
This will sometimes mean a reduction in fidelity, since the entangled
state can become mixed after tracing out the spatial degrees. Relatively
high fidelities have been measured with this approach. Depending on
the system, this can be viewed as occurring because coupling to phonons
is weak \cite{fuchs2011quantum} or because experiments occur on faster
time-scales than the atomic motion \cite{bernien2017probing}. 
\end{description}

\subsubsection{Photonic~fidelity }

Photonic measurements are typical of quantum memories \cite{Lvovsky:2009aa,CHANELIERE201877},
communications or cryptography when photons are used as the information
carrier. In the case of a quantum memory, a quantum state is first
encoded in a well-defined spatiotemporal mode(s), then dynamically
coupled into the memory subsystem, stored for a chosen period, and
coupled out into a second well-defined spatiotemporal field mode(s)
where it can be measured \cite{HeReidPhysRevA.79.022310}. We note
that temporal mode structure is an essential part of defining a quantum
state.

The actual quantum state in these cases is defined as an outer product
of photonic states in each possible mode $\left|n_{k}\right\rangle _{k}$,
where $n_{k}$ is the photon number in each mode, so that: 
\begin{equation}
\left|\Psi\right\rangle =\sum_{\bm{n}}C_{\bm{n}}\left|n_{1}\right\rangle _{1}{\otimes}\ldots\left|n_{K}\right\rangle _{K}.
\end{equation}
Here each mode has an associated mode function $\bm{u}_{k}$, which
is typically localized in space-time, since technology applications
are carried out in finite regions of space, over finite time-intervals.
We implicitly assume a finite total number of modes $K$, although
there is no physical upper bound except possibly that from quantum
gravity.

To obtain the output density matrix, the most rigorous approach is
to use pulsed homodyne detection to isolate the mode(s) used, with
a variable phase delay or other methods to measure the off-diagonal
elements. This gives a projected quadrature expectation value of the
relevant single mode operator. By tomographic reconstruction, one
can obtain the Wigner function \cite{lvovsky2009continuous}. We show
in Section~\ref{Sec:PhaseSpace} that this phase-space technique
directly gives the quantum fidelity as an $\Fmax$ measure. Obtaining
any other fidelity measure generally requires the reconstructed density
matrix. See, however,~\cite{miszczak2009sub,Bartkiewicz:PRA:2013}
where the authors discuss a direct measurement of the superfidelity,
$\Fn$, together with a lower bound called the subfidelity for photonic
states encoded in the polarization degree of freedom. Since $\Fn=\F_{1}$
for qubit states \cite{mendonca2008}, their method actually amounts
to a direct measurement of $\F_{1}$ for these qubit states without
resorting to quantum state tomography.

\subsubsection{Conditional~fidelity}

In some types of photonic fidelity measurement the state may not be
found at all in some of the measurements. This is typically the case
in photo-detection experiments with low photon number, where a qubit
can be encoded into two spatial or polarization modes, as $\left|\psi\right\rangle =\frac{1}{\sqrt{|a|^{2}+|b|^{2}}}{\left(a\left|0\right\rangle _{1}\left|1\right\rangle _{2}+b\left|1\right\rangle _{1}\left|0\right\rangle _{2}\right)}$.
Problems arise when no photon is detected at all, either because the
photodetector was inefficient, or because the photon was lost during
the transmission, thereby making the input a vacuum state, which is
in a larger Hilbert space.

As a result, reported measurements are sometimes defined by simply
conditioning all results on the presence of a detected photon(s).
Unless the target state is itself defined to be the conditioned state,
this conditional fidelity is best regarded as an upper bound for the
true mixed state fidelity, which includes these loss effects. The
potential difficulty with conditional fidelity measured in this way
is that it essentially involves an assumption of fair sampling. In
other words, photons may be lost through detector inefficiency, but
they may also be lost in any number of other ways. 

While detector loss can be regarded as simply a measurement issue,
unrelated to the state itself, there is also a possibility that the
state was already degraded before it reached the detector, and hence
the true fidelity is lower than estimated by the conditional measurement.
Yet in many applications, like quantum logic gates and computing,
one may have to repeat the same quantum memory process many times
in succession. In these cases any inefficiency that occurs prior to
detection grows exponentially with the number of gates, and can become
an important issue. In other quantum information processes however,
it might be argued that this effect is not relevant. In quoting fidelity,
it is thus important to match the target state with the final intended
application.

\subsubsection{Inferred~fidelity}

In a similar way to efficiency problems that lead to conditional fidelity
measures, the spatiotemporal mode may change from measurement to measurement
in photonic experiments. This leads to a mixed state. As a result,
the increased number of modes present can enlarge the Hilbert space
in a way that is not detectable through measurements of photodetection
events without using interferometry or local oscillators. This approach
is sometimes combined with a conditional measurement.

An example of this approach is a recently reported quantum memory
for orbital angular momentum qubits \cite{nicolas2014quantum}. This
experiment has many robust and useful features, using spatial mode
projection to ensure that the correct transverse mode is matched from
input to output. However, the report does not describe how longitudinal
or temporal mode structure was determined, leaving this issue as an
open question at this stage.

A photon-counting approach cannot usually detect the full mode structure.
For example, suppose one has a wide range of longitudinal modes that
can be occupied, having distinct frequencies and/or temporal mode
structures, and each occurring with a probability $P_{k}\ll1$, so
that: 
\begin{equation}
\rho=\sum_{k}P_{k}\left|0\right\rangle _{1}\ldots\left|1\right\rangle _{k}\ldots\left\langle 0\right|_{1}\ldots\left\langle 1\right|_{k}\ldots.
\end{equation}
This is a mixed state in which a single photon could be in any longitudinal
mode with a given probability $P_{k}$.

Let us now compare this with a desired pure state $\sigma$, for example:

\begin{equation}
\sigma=\left|1\right\rangle _{1}\ldots\left|0\right\rangle _{k}\ldots\left\langle 1\right|_{1}\ldots\left\langle 0\right|_{k}\ldots..
\end{equation}
It is clear that, for any definition of fidelity, $\mathcal{F}=\tr\left(\sigma\rho\right)=P_{1}\ll1$.
One could attempt to measure this fidelity with a photon-counting
measurement, combined with the {\em assumption} that there is only
one longitudinal mode present. If all measurements give exactly one
count, then this measurement, combined with the single-mode assumption
would lead to an inferred state fidelity of $\mathcal{F}_{{\rm inf}}=1$.
This does not match the true fidelity in this example.

Fidelity measurements like these generally make the assumption that
the mode structure that is measured matches the desired mode structure,
even when it is not measured directly. As a result, the inferred fidelity
may be higher than the true fidelity, and should be considered as
an upper bound. This may cause problems if one must carry out a binary
quantum logic operation with input signals derived from two different
sources such that the modes should be matched in time and/or frequency.
Under these conditions, it is the true fidelity, including the effects
of losses and modal infidelity that is important. These questions
have been investigated in experiments that carry out full tomographic
measurements to reconstruct single-photon Fock states using homodyne
measurement techniques \cite{lvovsky2001quantum}.

\subsubsection{Cloned~fidelity }

A fifth type of fidelity measurement is obtained as a variant of a
quantum game in which a number of copies of a quantum state may be
recorded or stored \cite{MassarPopPhysRevLett.74.1259}. From subsequent
measurements, it is possible to infer, for example, using maximum
likelihood measurements, what the original state was. The inferred
state can be compared with the original using fidelities. Unfortunately
the no-cloning theorem tells us that multiple copies of any single
quantum state cannot be obtained reliably. Hence, while one can infer
a fidelity from multiple copies of a state, the entire process that
includes first generating multiple copies of a quantum state will
always involve an initial reduction in fidelity. This is important
in some types of application.

\subsubsection{Logic and process fidelity}

Finally, we turn to a different type of fidelity used to analyze quantum
processes rather than states or density matrices. Quantum processes
are also known~\cite{preskill2015lecture} as quantum channels, or
mathematically as completely positive maps or super-operators. For
the case of a quantum process, one may wish to analyze the fidelity
of an actual quantum operation to an intended quantum operation. This
may include any quantum technology from logic gates to memories, or
indeed any input-output process. An operation is defined in the general
sense of any quantum map ${\cal E\left(\rho\right)}$, from an input
density matrix $\rho_{{\rm in}}$ to an output density matrix $\rho_{{\rm out}}.$
Their fidelity is discussed by Gilchrist \emph{et. al.}~\cite{gilchrist2005distance}.

Just as any density matrix has a matrix representation in the Hilbert
space of state vectors of dimension $d$, quantum channels have a
matrix representation in terms of a basis set of $d^{2}$ quantum
operators $A_{j}$, where tr$\left(A_{j}^{\dagger}A_{k}\right)=\delta_{jk}$.
Using this basis, any quantum operation can be written as: 
\begin{equation}
\rho_{{\rm out}}={\cal E}\left(\rho_{{\rm in}}\right)=\sum_{mn}P_{mn}A_{m}\rho_{{\rm in}}A_{n}^{\dagger}
\end{equation}
Here $P_{mn}$ are the elements of the so-called process matrix $P$,
which provides a convenient way of representing the operation ${\cal E}.$

At first, this seems rather different to density matrices, as discussed
throughout this review Yet it is easy to show via the Choi-Jamiolkowski
isomorphism \cite{Choi:1975,jamiolkowski1974effective}, that one
can define a new quantum ``density matrix"~\cite{gilchrist2005distance}
on the enlarged Hilbert space of dimension $d^{2}$, such that $\rho^{{\cal E}}=P/d$.
Hence any fidelity or distance measure for quantum states can also
be applied to processes, by the simple technique of dimension squaring.
We will not investigate this in detail here except to remark that
all of the different fidelity measures used for density matrices can
be applied directly to quantum processes. For process fidelity it
is common to impose additional requirements for the fidelity in addition
to the axioms used here.

A typical example is the measurement of quantum process fidelity in
a CNOT gate \cite{o2004quantum}. In this early photonic measurement,
the counting fidelity was measured using conditional techniques. Thus,
as explained above, these results should be regarded as an upper bound
to the actual quantum fidelity, once mode-mismatch errors and losses
are included. Other, more recent, process fidelity measurements with
quantum logic gates have been carried out with ion traps \cite{benhelm2008towards,knill2008randomized},
liquid nuclear magnetic resonance \cite{ryan2009randomized}, solid-state
silicon \cite{veldhorst2015two} and superconducting Josephson qubits
\cite{lucero2008high}.

\section{Summary}

We have reviewed the requirements for a mixed state fidelity measure
\cite{jozsa1994fidelity}, and analyzed a number of candidate measures
of fidelity for their compliance with these requirements, as well
as other considerations. While there are many candidates, most of
them do not fully comply with the Josza axioms, although some of these
alternatives do have useful properties.

Despite the above observation, there do exist an infinite number of
compliant fidelities. Among them, three well-defined measures that
fully satisfy the Josza axioms for fidelity measures are of particular
interest, due to their physical interpretation and measurement properties:
these are the Uhlmann-Josza fidelity $\Fuj$, the non-logarithmic
variety of the quantum Chernoff bound $\Fq$ and the Hilbert-Schmidt
fidelity, $\Fmax$. It is worth noting that both $\Fuj$ and $\Fmax$
are particular cases of an infinite family of Josza-compliant fidelity
measures $\Fp$, each associated with a Schatten-von-Neumann $p$-norm.

In this review, we have focused on two specific cases of these norm
based fidelity measures $\Fp$, as well as the quantum Chernoff bound
$\Fq$. Analyzing the properties of this family of candidate measures
for other integer values of $p>2$ is clearly something that may be
of independent interest. On the other hand, despite much effort, the
validity of a few desired properties of various candidate fidelity
measures remains unknown (see Table~\ref{tbl:Concavity}, Table~\ref{tbl:Monotonicity}
and Table~\ref{tbl:RelatedMetrics} for details). For each of these
conjectured properties, at least 2000 optimizations with different
initial starting points have been carried out for each Hilbert space
dimension $d=2,3,\ldots,10$. Given the fact that intensive numerical
searches have been carried out for counterexamples to these properties
for small Hilbert space dimensions, we are inclined to conjecture
that these properties are indeed valid.

An intriguing result is that of the $\Fp$ fidelities investigated,
the $\Fuj$ fidelity gives the largest average values when random
density matrices are compared. While this relationship is only universal
for the qubit case\textemdash otherwise there are occasional exceptions\textemdash it
is found on average for higher dimensional Hilbert spaces as well.
This is clearly related to the fact that the $\Fuj$ fidelity is defined
as a maximum fidelity over purifications. These purifications represent
an unmeasured portion of Hilbert space. Hence, measuring $\Fuj$ on
a subspace could introduce a bias compared to a more complete measurement
on a larger relevant space, if there are additional errors in the
unmeasured part of the relevant Hilbert space.

To conclude, while the Uhlmann-Josza measure is the most widely known
fidelity measure, there are other alternatives which have properties
that can make them preferable under some circumstances. They may be
either simpler to compute or more relevant to certain applications.
For example, the Hilbert-Schmidt fidelity measure $\Fmax$ is well-defined
even for unnormalized density matrices, and appears less biased towards
high values. Finding out the full implications of this and other mathematical
properties of the candidate measures, however, is too broad a research
topic to be considered within the present review.

\section*{Acknowledgements}

YCL, PEMFM, and PDD contributed equally towards this work. This work
is supported by the Ministry of Science and Technology, Taiwan (Grants
No. 104-2112-M-006-021-MY3 and 107-2112-M-006-005-MY2) and the Center
for Quantum Technology, Hsinchu, Taiwan. PDD and MDR thank the Australian
Research Council and the hospitality of the Institute for Atomic and
Molecular Physics (ITAMP) at Harvard University, supported by the
NSF. YCL acknowledges useful discussions with N. Gisin and N. Sangouard.
PDD thanks B. Sparkes for useful discussions.

\appendix
%dummy comment inserted by tex2lyx to ensure that this paragraph is not empty%dummy comment inserted by tex2lyx to ensure that this paragraph is not empty%dummy comment inserted by tex2lyx to ensure that this paragraph is not empty%dummy comment inserted by tex2lyx to ensure that this paragraph is not empty

\section{Detailed proofs}

\label{Sec:Proofs}

In this Appendix, detailed proofs are obtained for the fidelity results
in the earlier sections, where they are not given already.

\subsection*{Norm based fidelity properties}

\label{App:p-norm}

\begin{theorem} All norm-based fidelities, $\Fp$, obey the Josza
axioms for $p\ge1$.

\end{theorem}

\begin{proof} ~ 
\begin{itemize}
\item[J1a)] $\Fp(\rho,\sigma)\in[0,1]$. The minimum bound is trivial, since
norms are positive semi-definite. That the maximum bound holds is
obtained from the Hölder inequality, since: $\left\Vert \sqrt{\rho}\sqrt{\sigma}\right\Vert _{p}^{2}\leq\left\Vert \sqrt{\rho}\right\Vert _{2p}^{2}\left\Vert \sqrt{\sigma}\right\Vert _{2p}^{2}=\left\Vert \rho\right\Vert _{p}\left\Vert \sigma\right\Vert _{p}\leq\max\left[\left\Vert \sigma\right\Vert _{p}^{2},\left\Vert \rho\right\Vert _{p}^{2}\right]$. 
\item[J1b)] $\Fp(\rho,\sigma)=1$ \emph{if and only if} $\rho=\sigma$. Clearly,
$\Fp(\rho,\rho)=1$ for identical operators, since $\left\Vert \sqrt{\rho}\sqrt{\rho}\right\Vert _{p}=\left\Vert \rho\right\Vert _{p}^{2}=\max\left[\left\Vert \rho\right\Vert _{p}^{2},\left\Vert \rho\right\Vert _{p}^{2}\right]$.
To prove the converse, {note from the above proof of J1a that the
maximum bound is attained if and only if the Hölder inequality becomes
an equality. Taking into account of the normalization of density matrices,
we thus see that the maximum bound of $\Fp(\rho,\sigma)=1$ is attained
if and only if $\rho=\sigma$.} 
\item[J1c)] $\Fp(\rho,\sigma)=0$ \emph{if and only if} $\rho\,\sigma=0$ . This
follows since $\rho\,\sigma=0\Longleftrightarrow\sqrt{\rho}\sqrt{\sigma}=0\Longleftrightarrow\left\Vert \sqrt{\rho}\sqrt{\sigma}\right\Vert _{p}=0$. 
\item[J2)] $\Fp(\rho,\sigma)=\Fp(\sigma,\rho)$ is clearly true from the symmetry
of the matrix norm under transposition. 
\item[J3)] $\Fp(\rho,\sigma)=\tr(\rho\,\sigma)$ if either $\rho$ or $\sigma$
is a pure state. From the definition, we have $\|\sqrt{\rho}\sqrt{\sigma}\|_{p}^{2}=\left[\tr\left(\sqrt{\rho}\,\sigma\,\sqrt{\rho}\right)^{\frac{p}{2}}\right]^{\frac{2}{p}}$.
Let $\rho$ be a pure state, then $\rho=\sqrt{\rho}=\proj{\psi}$
for some $\ket{\psi}$, the expression thus simplifies to $\left[\left(\bra{\psi}\sigma\ket{\psi}\right)^{\frac{p}{2}}\tr\left(\rho^{\frac{p}{2}}\right)\right]^{\frac{2}{p}}$,
which reduces to $\bra{\psi}\sigma\ket{\psi}=\tr(\rho\,\sigma)$ since
$\tr\left(\rho^{\frac{p}{2}}\right)=\tr\rho=1$. Similarly, the normalizing
factor is unity, since $\left\Vert \rho\right\Vert _{p}^{2}=1\ge\left\Vert \sigma\right\Vert _{p}^{2}$,
and the argument holds if $\rho,\sigma$ are interchanged. 
\item[J4)] $\Fp(U\rho\,U^{\dagger},U\sigma U^{\dagger})=\Fp(\rho,\sigma)$.
This follows from the unitary invariance of the matrix norm. 
\end{itemize}
\end{proof}

\subsection*{Normalization}

We now show that both $\Fam$ and $\Fgm$ obey axiom (J1a) and (J1b).

\begin{theorem} $\Fam(\rho,\sigma),\Fgm(\rho,\sigma)\in[0,1]$ with
the upper bound attained if and only if $\rho=\sigma$. \end{theorem}
\begin{proof} The non-negativity of $\Fam$ and $\Fgm$ is obvious.
Next, we prove that these quantities are upper bounded by 1. To this
end, we recall that the geometric mean between two numbers is always
upper bounded by its arithmetic mean, hence, $\Fam(\rho,\sigma)\le\Fgm(\rho,\sigma)\le1$,
where the second inequality follows easily from the Cauchy-Schwarz
inequality.

Moreover, the Cauchy-Schwarz inequality is saturated if and only if
its entries are scalar multiples of each other. Since our entries
have unity trace (density matrices), saturation can only occur if
$\rho=\sigma$. It is easy to see by inspection that both $\Fam(\rho,\rho)$
and $\Fgm(\rho,\rho)$ indeed equal to unity. Thus $\Fam(\rho,\sigma)=\Fgm(\rho,\sigma)=1$
if and only if $\rho=\sigma$. \end{proof}

\subsection*{Multiplicativity}

\label{App:Proof:Mul}

%% Supermultiplicaitivity: F2

\begin{theorem} {The measure} $\Fmax$ is {generally} super-multiplicative{,
but is multiplicative when appended by an uncorrelated ancillary state,
or when considering tensor powers of the same states.} \end{theorem}
\begin{proof} Let us define $\rho=\rho_{1}\otimes\rho_{2}$ and $\sigma=\sigma_{1}\otimes\sigma_{2}$,
then start by noting that 
\begin{equation}
\F_{2}(\rho,\sigma)=\frac{\tr(\rho_{1}\sigma_{1})\tr(\rho_{2}\sigma_{2})}{\max\left[\tr(\rho_{1}^{2})\tr(\rho_{2}^{2}),\tr(\sigma_{1}^{2})\tr(\sigma_{2}^{2})\right]}
\end{equation}
and 
\begin{equation}
\F_{2}(\rho_{1},\sigma_{1})\F_{2}(\rho_{2},\sigma_{2})=\frac{\prod_{i=1}^{2}\tr(\rho_{i}\sigma_{i})}{\prod_{i=1}^{2}\max\left[\tr(\rho_{i}^{2}),\tr(\sigma_{i}^{2})\right]}
\end{equation}
Clearly, saturation of inequality~(\ref{Eq:SuperMultiplicative})
is obtained if the denominators of the two equations above coincide,
i.e., 
\begin{eqnarray}
\max\left[\prod_{i=1}^{2}\tr(\rho_{i}^{2}),\prod_{j=1}^{2}\tr(\sigma_{j}^{2})\right]=\prod_{i=1}^{2}\max\left[\tr(\rho_{i}^{2}),\tr(\sigma_{i}^{2})\right].\nonumber \\
\,\label{eq:maxeq}
\end{eqnarray}
It is easy to check that this equation holds when (at least) one of
the following is observed 
\begin{itemize}
\item $\tr(\rho_{i}^{2})=\tr(\sigma_{i}^{2})$ for some $i=1,2$, 
\item $\tr(\rho_{i}^{2})>\tr(\sigma_{i}^{2})$ for both $i=1,2$, 
\item $\tr(\rho_{i}^{2})<\tr(\sigma_{i}^{2})$ for both $i=1,2$. 
\end{itemize}
Note that {the} first condition is satisfied for the scenario when
each quantum state is appended{, respectively, by a quantum state
with the same purity (e.g., when they are both appended by} same
ancillary state $\tau$. {On the other hand,} the second/ third
condition {is satisfied for the} scenario of tensor powers, i.e.,
$\rho_{1}=\rho_{2}$ etc., cf., Eq.~\eref{Eq:MultiplicativeTensorPower}.
When none of the above conditions is satisfied, it is easy to see
that we must have the r.h.s. of Eq.~\eref{eq:maxeq} {larger}
than its l.h.s., and hence super-multiplicativity. For example, if
$\tr(\rho_{1}^{2})>\tr(\sigma_{1}^{2})$ and $\tr(\rho_{2}^{2})<\tr(\sigma_{2}^{2})$,
then the r.h.s. of \eref{eq:maxeq} becomes $\tr(\rho_{1}^{2})\tr(\sigma_{2}^{2})$
which has to be {larger} than both $\tr(\rho_{1}^{2})\tr(\rho_{2}^{2})$
or $\tr(\sigma_{1}^{2})\tr(\sigma_{2}^{2})$ by assumption. The proof
for the case when $\tr(\rho_{1}^{2})<\tr(\sigma_{1}^{2})$ and $\tr(\rho_{2}^{2})>\tr(\sigma_{2}^{2})$
proceeds analogously. \end{proof}

%% Supermultiplicaitivity: Fc

\begin{theorem} {The measure $\Fc$ is generally supermultiplicative.}
\end{theorem}

\begin{proof} Let $d_{1}\defeq\dim\rho_{1}=\dim\sigma_{1}$, $d_{2}\defeq\dim\rho_{2}=\dim\sigma_{2}$
and, in terms of these, define $r\defeq(d_{1}-1)^{-1}$, $s\defeq(d_{2}-1)^{-1}$
and $t\defeq(d_{1}d_{2}-1)^{-1}$ {[}or, equivalently, $t=rs/(1+r+s)${]},
in such a way that $r,s\in(0,1]$. Then, the statements of the supermultiplicativity
of $\Fc$ (to be proved) and of $\Fn$ (proved in~\cite{mendonca2008})
can be expressed, respectively, as

\begin{eqnarray}
2(1-t)+2(1+t)\Fn(\rho_{1}\otimes\rho_{2},\sigma_{1}\otimes\sigma_{2}) & \geq\nonumber \\
\left[(1-r)+(1+r)x\right]\left[(1-s)+(1+s)y\right]\,,\label{eq:tbp}\\
2(1-t)+2(1+t)\Fn(\rho_{1}\otimes\rho_{2},\sigma_{1}\otimes\sigma_{2}) & \geq\nonumber \\
2(1-t)+2(1+t)xy\,,\label{eq:ap}
\end{eqnarray}
where, for brevity, we have defined $x\defeq\Fn(\rho_{1},\sigma_{1})\in[0,1]$
and $y\defeq\Fn(\rho_{2},\sigma_{2})\in[0,1]$. In what follows, the
validity of inequality~inequality~(\ref{eq:tbp}) is established
by showing that the r.h.s. of~(\ref{eq:ap}) dominates the r.h.s.
of~inequality~(\ref{eq:tbp}); an inequality that can be written
as 
\begin{equation}
2(1-t)-(1-r)(1-s)\geq f_{r,s}(x,y)\label{eq:tbp2}
\end{equation}
where we have defined the functions 
\begin{eqnarray}
f_{r,s}(x,y) & \defeq(1+r)(1-s)x+(1-r)(1+s)y\nonumber \\
 & -\left[2(1+t)-(1+r)(1+s)\right]xy
\end{eqnarray}

To see that inequality~(\ref{eq:tbp2}) holds, it is our interest
to find out the maximum of the functions $f_{r,s}(x,y)$. The extreme
point of $f_{r,s}(x,y)$ is given by respectively setting the partial
derivative of $x$ and $y$ to zero. It can be verified that, unless
$r,s\in\{0,1\}$, the extreme points of $f_{r,s}(x,y)$ lie outside
the domain $x,y\in[0,1]$. In the former cases, the extreme points
lie on the boundaries of the domain. To identify them, note that the
functions $f_{r,s}(x,y)$ are increasing both in $x$ and $y$ for
all values of parameters $r,s\in[0,1]$. This follows, for example,
from the observation that their partial derivatives with respect to
$x$ and $y$ are always linear and assume non-negative values in
the extremes of the domain $x,y\in[0,1]$. Indeed, 
\begin{eqnarray}
\left.\frac{\partial f_{r,s}(x,y)}{\partial x}\right|_{y=0} & =(1+r)(1-s)\geq0\nonumber \\
\left.\frac{\partial f_{r,s}(x,y)}{\partial x}\right|_{y=1} & =\frac{2r(1+r)}{1+r+s}>0\nonumber \\
\left.\frac{\partial f_{r,s}(x,y)}{\partial y}\right|_{x=0} & =(1-r)(1+s)\geq0\nonumber \\
\left.\frac{\partial f_{r,s}(x,y)}{\partial y}\right|_{x=1} & =\frac{2s(1+s)}{1+r+s}>0
\end{eqnarray}
Thanks to that, it suffices to verify the validity of inequality~(\ref{eq:tbp2})
for $x=1$ and $y=1$, where $f_{r,s}(x,y)$ is maximal. In this case,
however, a straightforward simplification process shows that the inequality
is satisfied with saturation. \end{proof}

%% Multiplicaitivity: FQ

\begin{theorem} The measure $\Fq$ is generally super-multiplicative
under tensor products{, but is multiplicative when appended by an
uncorrelated ancillary state, or when considering tensor powers of
the same states.} \end{theorem}

\begin{proof} For given density matrices $\rho_{1}$, $\rho_{2}$,
$\sigma_{1}$ and $\sigma_{2}$, we see that 
\begin{eqnarray}
\Fq(\rho_{1}\otimes\rho_{2},\sigma_{1}\otimes\sigma_{2}) & =\min_{s}\,\,\tr\left[\left(\rho_{1}\otimes\rho_{s}\right)^{s}\left(\sigma_{1}\otimes\sigma_{2}\right)^{1-s}\right]\nonumber \\
 & =\min_{s}\,\,\tr(\rho_{1}^{s}\,\sigma_{1}^{1-s})\tr(\rho_{2}^{s}\,\sigma_{2}^{1-s})\nonumber \\
 & =\min_{s}\,\,f_{1}(s)\,f_{2}(s)\nonumber \\
 & =f_{1}(s^{*})f_{2}(s^{*})
\end{eqnarray}
where $f_{i}(s)\defeq\tr(\rho_{i}^{s}\,\sigma_{i}^{1-s})$, $s^{*}$
is a minimizer of the function $f_{1}(s)\,f_{2}(s)$, and it is worth
reminding that each $f_{i}(s)$ is a convex function of $s$~\cite{audenaert2007discriminating}.

On the other hand, it also follows from the definition of $\Fq$ that
\begin{eqnarray}
\Fq(\rho_{1},\sigma_{1})\Fq(\rho_{2},\sigma_{2}) & =\min_{s_{1}}\,\,f_{1}(s_{1})\,\min_{s_{2}}\,\,f_{2}(s_{2})\nonumber \\
 & =f_{1}(s_{1}^{*})\,f_{2}(s_{2}^{*}),
\end{eqnarray}
where $s_{i}^{*}\in\mathcal{S}_{i}$ is a minimizer of $f_{i}(s_{i})$,
and $\mathcal{S}_{i}\subseteq[0,1]$ is the set of minimizers of $f_{i}(s_{i})$.
Note that the convexity of $f_{i}$ guarantees that $\mathcal{S}_{i}$
is a convex interval in $[0,1]$.

Evidently, for general density matrices $\rho_{1}$, $\rho_{2}$,
$\sigma_{1}$ and $\sigma_{2}$, the set of minimizers for $f_{1}(s)$
and $f_{2}(s)$ differ, i.e., $\mathcal{S}_{1}\neq\mathcal{S}_{2}$.
Without loss of generality, let us assume that $\min\mathcal{S}_{1}\le\min\mathcal{S}_{2}$.
There are now two cases to consider, if $\max\mathcal{S}_{1}\ge\min\mathcal{S}_{2}$,
we must also have $\mathcal{S}_{2}\cap\mathcal{S}_{1}\neq\{\}$, and
hence the minimizer $s^{*}$ for $f_{1}(s)f_{2}(s)$ must also minimize
$f_{1}$ and $f_{2}$. In this case, we have, 
\begin{equation}
\Fq(\rho_{1}\otimes\rho_{2},\sigma_{1}\otimes\sigma_{2})=\Fq(\rho_{1},\sigma_{1})\Fq(\rho_{2},\sigma_{2}),
\end{equation}
meaning that multiplicativity holds true for this set of density matrices.

In the event that $\max\mathcal{S}_{1}<\min\mathcal{S}_{2}$, we have
$\mathcal{S}_{2}\cap\mathcal{S}_{1}=\{\}$, which implies that any
minimizer $s^{*}$ for $f_{1}(s)f_{2}(s)$ must be such that {$s^{*}\not\in\S_{1}$
and/or $s^{*}\not\in\S_{2}$}. As a result, we must have $f_{i}(s^{*})\ge f_{i}(s_{i}^{*})$
for all $i$, and thus 
\begin{eqnarray}
\Fq(\rho_{1}\otimes\rho_{2},\sigma_{1}\otimes\sigma_{2}) & =f_{1}(s^{*})f_{2}(s^{*}),\nonumber \\
 & \ge f_{1}(s_{1}^{*})f_{2}(s_{2}^{*}),\nonumber \\
 & =\Fq(\rho_{1},\sigma_{1})\Fq(\rho_{2},\sigma_{2}),
\end{eqnarray}
which demonstrates the super-multiplicativity of $\Fq$.

For tensor powers of the same state, cf. Eq.~\eref{Eq:MultiplicativeTensorPower},
the fact that $\tr(A\otimes B)=\left(\tr\,A\right)\left(\tr\,B\right)$
and the above arguments make it evident that the minimizer for a single
copy is also the minimizer for an arbitrary number of copies. Thus,
$\Fq$ is multiplicative under tensor powers. Similarly, for an uncorrelated
ancillary state $\tau$, the definition of $\Fq$, and axiom (J1b)
immediately imply that the measure is multiplicative when each quantum
state is appended by $\tau$, i.e., $\Fq(\rho\otimes\tau,\sigma\otimes\tau)=\Fq(\rho,\sigma)$.
\end{proof}

%% Multiplicaitivity: F_AM

\begin{theorem} The measure $\Fam$ is multiplicative under uncorrelated
ancilla state, that is 
\begin{equation}
\Fam(\rho\otimes\tau,\sigma\otimes\tau)=\Fam(\rho,\sigma)\,
\end{equation}
\end{theorem} \begin{proof} The proof follows trivially from the
application of the tensor product identities 
\begin{eqnarray}
(A\otimes B)(C\otimes D) & = & AC\otimes BD\label{eq:otimesmult}\\
\tr(A\otimes B) & = & \tr(A)\tr(B)\label{eq:trotimes}
\end{eqnarray}
Explicitly, 
\begin{eqnarray}
\Fam(\rho\otimes\tau,\sigma\otimes\tau) & = & \frac{2\tr(\rho\,\sigma\otimes\tau^{2})}{\tr(\rho^{2}\otimes\tau^{2})+\tr(\sigma^{2}\otimes\tau^{2})}\nonumber \\
 & = & \frac{2\tr(\rho\,\sigma)\tr(\tau^{2})}{[\tr(\rho^{2})+\tr(\sigma^{2})]\tr(\tau^{2})}\nonumber \\
 & = & \frac{2\tr(\rho\,\sigma)}{[\tr(\rho^{2})+\tr(\sigma^{2})]}\nonumber \\
 & = & \mathcal{F}_{AM}(\rho,\sigma),
\end{eqnarray}
where Eq.~(\ref{eq:otimesmult}) was used to establish the second
equality and Eq.~(\ref{eq:trotimes}) was used to establish the third
equality. \end{proof}

\subsection*{Proofs of average fidelity properties}

Here, we give the proofs showing when the respective fidelity measure
satisfies Eq.~\eref{Eq:AveVsMixed}. Throughout, as mentioned in
Section~\ref{Sec:PureVsMixed} and Section~\ref{Sec:PureVsMixed-1},
we will assume that the error state $\rho_{0}$ is orthogonal to {\em
all} of the signal state, i.e., 
\begin{equation}
\rho_{j}\rho_{0}=\rho_{0}\rho_{j}=0\quad\forall\,j\neq0.\label{Eq:OrthoError}
\end{equation}
If, in addition to Eq.~\eref{Eq:OrthoError}, all signal states
are orthogonal to each other, i.e., 
\begin{equation}
\rho_{j}\rho_{k}=0\quad\forall\,j\neq k,\label{Eq:OrthoAll}
\end{equation}
then all $\rho_{j}$ commute pairwise, and hence diagonalizable in
the same basis. In this case, we will see that Eq.~\eref{Eq:AveVsMixed}
holds for $\Fuj$, $\Fq$ and $\Fa$ independent of the purity of
the signal state $\rho_{j}$. \begin{proof} Recall from the main
text that in the present discussion, the average state-by-state fidelity
is 
\begin{eqnarray}
\F_{{\rm ave}}(\bm{p},\bm{\rho},\bm{\sigma})=\sum_{i\neq0}p_{i}\F(\rho_{i},\sigma_{i}).\label{Eq:AveF}
\end{eqnarray}
Note, however, that Eq.~\eref{Eq:OrthoAll} implies that for all
$i\neq0$, 
\begin{eqnarray}
\Fuj(\rho_{i},\sigma_{i}) & =\left(\tr\sqrt{\sqrt{\rho_{i}}\sigma_{i}\sqrt{\rho_{i}}}\right)^{2},\nonumber \\
 & =\left(\tr\sqrt{\sqrt{\rho_{i}}\left[\epsilon\rho_{0}+(1-\epsilon)\rho_{i}\right]\sqrt{\rho_{i}}}\right)^{2},\nonumber \\
 & =\left(\sqrt{1-\epsilon}\tr\sqrt{\sqrt{\rho_{i}}\rho_{i}\sqrt{\rho_{i}}}\right)^{2},\nonumber \\
 & =(1-\epsilon)\left(\tr\rho_{i}\right)^{2}=1-\epsilon\label{Eq:F1:State-by-state}
\end{eqnarray}
while 
\begin{eqnarray*}
\Fuj(\rho,\sigma) & =\left(\tr\sqrt{\sqrt{\rho}\sigma\sqrt{\rho}}\right)^{2},\\
 & =\left(\tr\sqrt{\sqrt{\rho}\left[\epsilon\rho_{0}+(1-\epsilon)\rho\right]\sqrt{\rho}}\right)^{2},\\
 & =\left(\sqrt{1-\epsilon}\tr\sqrt{\sqrt{\rho}\,\rho\,\sqrt{\rho}}\right)^{2},\\
 & =(1-\epsilon)\left(\tr\sqrt{\rho^{2}}\right)^{2}=1-\epsilon
\end{eqnarray*}
Substituting Eq.~\eref{Eq:F1:State-by-state} to Eq.~\eref{Eq:AveF}
and comparing the resulting expression with the above equation then
verifies Eq.~\eref{Eq:AveVsMixed} for $\Fuj$.

Next, we shall prove the equivalence for $\Fq$. As with $\Fuj$,
Eq.~\eref{Eq:OrthoAll} implies that for all $i\neq0$, 
\begin{eqnarray}
\Fq(\rho_{i},\sigma_{i}) & =\min_{0\le s\le1}\tr(\rho_{i}^{s}\,\sigma_{i}^{1-s}),\nonumber \\
 & =\min_{0\le s\le1}\tr(\rho_{i}^{s}\,\left[\epsilon\rho_{0}+(1-\epsilon)\rho_{i}\right]^{1-s}),\nonumber \\
 & =\min_{0\le s\le1}\tr(\rho_{i}^{s}\,\left[\epsilon^{1-s}\rho_{0}^{1-s}+(1-\epsilon)^{1-s}\rho_{i}^{1-s}\right]),\nonumber \\
 & =\min_{0\le s\le1}(1-\epsilon)^{1-s}\tr(\rho_{i})=1-\epsilon,\label{Eq:Fq:State-by-state}
\end{eqnarray}
where the last equality follows from the fact that $0\le1-\epsilon\le1$,
and thus $(1-\epsilon)^{1-s}\ge1-\epsilon$ for all $0\le s\le1$.

In a similar manner, the simultaneous diagonalizability of $\rho$
and $\rho_{0}$ gives 
\begin{eqnarray}
\Fq(\rho,\sigma) & =\min_{0\le s\le1}\tr(\rho^{s}\,\sigma^{1-s}),\nonumber \\
 & =\min_{0\le s\le1}\tr\left\{ \rho^{s}\,\left[\epsilon\rho_{0}+(1-\epsilon)\rho\right]^{1-s}\right\} ,\nonumber \\
 & =\min_{0\le s\le1}\tr\left\{ \rho^{s}\,\left[\epsilon^{1-s}\rho_{0}^{1-s}+(1-\epsilon)^{1-s}\rho^{1-s}\right]\right\} ,\nonumber \\
 & =\min_{0\le s\le1}(1-\epsilon)^{1-s}\tr(\rho)=1-\epsilon.
\end{eqnarray}
Substituting Eq.~\eref{Eq:Fq:State-by-state} into Eq.~\eref{Eq:AveF}
and comparing with the last equation immediately leads to the verification
of Eq.~\eref{Eq:AveVsMixed} for $\Fq$ whenever Eq.~\eref{Eq:OrthoAll}
holds.

To prove the equivalence for $\Fa$, we note from Eq.~\eref{Eq:OrthoAll}
that for all $i\neq0$, 
\begin{eqnarray}
\Fa(\rho_{i},\sigma_{i}) & =\left[\tr\left(\sqrt{\rho_{i}}\sqrt{\sigma_{i}}\right)\right]^{2},\nonumber \\
 & =\left[\tr\left(\sqrt{\rho_{i}}\sqrt{\epsilon\rho_{0}+(1-\epsilon)\rho_{i}}\right)\right]^{2},\nonumber \\
 & =\left\{ \tr\left[\sqrt{\rho_{i}}\left(\sqrt{\epsilon}\sqrt{\rho_{0}}+\sqrt{1-\epsilon}\sqrt{\rho_{i}}\right)\right]\right\} ^{2},\nonumber \\
 & =\left[\sqrt{1-\epsilon}\tr\left(\rho_{i}\right)\right]^{2}=1-\epsilon.\label{Eq:Fa:State-by-state}
\end{eqnarray}

Similarly, Eq.~\eref{Eq:OrthoAll} implies that 
\begin{eqnarray}
\Fa(\rho,\sigma) & =\left[\tr\left(\sqrt{\rho}\sqrt{\sigma}\right)\right]^{2},\nonumber \\
 & =\left[\tr\left(\sqrt{\rho}\sqrt{\epsilon\rho_{0}+(1-\epsilon)\rho}\right)\right]^{2},\nonumber \\
 & =\left\{ \tr\left[\sqrt{\rho}\left(\sqrt{\epsilon}\sqrt{\rho_{0}}+\sqrt{1-\epsilon}\sqrt{\rho}\right)\right]\right\} ^{2},\nonumber \\
 & =\left[\sqrt{1-\epsilon}\tr\left(\rho\right)\right]^{2}=1-\epsilon.
\end{eqnarray}
Hence, by substituting Eq.~\eref{Eq:Fa:State-by-state} into Eq.~\eref{Eq:AveF}
and comparing with the last equation immediately leads to the verification
of Eq.~\eref{Eq:AveVsMixed} for $\Fa$ whenever Eq.~\eref{Eq:OrthoAll}
holds. \end{proof}

Notice that if instead of Eq.~\eref{Eq:OrthoAll}, we have the
promise that 
\begin{equation}
\tr(\rho_{j}^{2})\ge\tr(\sigma_{j}^{2}),\quad\tr(\rho^{2})\ge\tr(\sigma^{2}),
\end{equation}
then 
\begin{eqnarray}
\Fmax(\rho,\sigma) & =\frac{\tr(\rho\,\sigma)}{\max\left[\tr(\rho^{2}),\tr(\sigma^{2})\right]},\nonumber \\
 & =\frac{\tr\left\{ \rho\,\left[\epsilon\rho_{0}+(1-\epsilon)\rho\right]\right\} }{\tr(\rho^{2})},\nonumber \\
 & =1-\epsilon.
\end{eqnarray}
Eq.~\eref{Eq:F2:compare} then follows by combining Eq.~\eref{Eq:F2:State-by-state},
Eq.~\eref{Eq:AveF} and the last equation above.

Finally, note that if instead we have the premise that $\tr(\rho^{2})=1$,
then $\rho=\rho_{j}$, $\sigma=\sigma_{j}$, and there is only term
in the sum of Eq.~\eref{Eq:AveF}. Consequently, we must also have
$\Fn(\rho,\sigma)=\Fn(\rho_{j},\sigma_{j})=\F_{{\rm ave}}(\bm{p},\bm{\rho},\bm{\sigma})$
in this case.

\subsection*{Counterexamples}

\label{App:CountExamples}

We provide here some counterexamples that have been left out in the
main text for showing certain desired properties of various fidelity
measures.

To verify that $\Fam$ is generally not (super)multiplicative, it
suffices to consider $\rho=\frac{1}{5}(\Pi_{0}+4\Pi_{1})$ and $\sigma=\frac{1}{2}\Id_{2}$.
An explicit calculation gives $\Fam(\rho\otimes\rho,\sigma\otimes\sigma)\approx0.702$
while $[\Fam(\rho,\sigma)]^{2}\approx0.718>\Fam(\rho\otimes\rho,\sigma\otimes\sigma)$,
thus showing a violation of the desired (super)multiplicative property.

To see that $\Fc$, $\Fgm$ and $\Fam$ can be contractive under the
partial trace operation, it suffices to consider the following pair
of two-qubit density matrices (written in the product basis): 
\begin{equation}
\rho=\Pi_{0}\otimes\Pi_{0},\quad\sigma=\frac{1}{8}\left(\begin{array}{cccc}
3 & 0 & 0 & \sqrt{3}\\
0 & 4 & 0 & 0\\
0 & 0 & 0 & 0\\
\sqrt{3} & 0 & 0 & 1
\end{array}\right),
\end{equation}
and consider the partial trace of $\rho$ and $\sigma$ over the first
qubit subsystem.

\section{Metric properties}

\label{Sec:Proofs-1}

In this Appendix, detailed proofs are obtained for the metric properties
of fidelities.

\label{App:Met}

\subsection*{General definitions}

For a given fidelity measure $\F(\rho,\sigma)$, let us define the
following functionals of $\F$: 
\begin{eqnarray}
A[\F(\rho,\sigma)] & \defeq\arccos[{\sqrt{\F(\rho,\sigma)}}],\nonumber \\
B[\F(\rho,\sigma)] & \defeq\sqrt{2-2\sqrt{\F(\rho,\sigma)}},\nonumber \\
C[\F(\rho,\sigma)] & \defeq\sqrt{1-\F(\rho,\sigma)}.
\end{eqnarray}

In what follows, we will provide the proofs of the metric properties
of these functionals of $\F$ for the various fidelity measures discussed
in Sec.~\ref{Sec:Metric}. By a metric, we mean a mapping $\mathfrak{D}$
on a set $S$ such that for every $a,b,c\in S$, the mapping $\mathfrak{D}:S\times S\to\mathbb{R}$
satisfies the following properties: 
\begin{enumerate}
\item[(M1)] $\mathfrak{D}(a,b)\geq0$ (Nonnegativity)\,, 
\item[(M2)] $\mathfrak{D}(a,b)=0$ iff $a=b$ (Identity of Indiscernible)\,, 
\item[(M3)] $\mathfrak{D}(a,b)=\mathfrak{D}(b,a)$ (Symmetry)\,, 
\item[(M4)] $\mathfrak{D}(a,c)\leq\mathfrak{D}(a,b)+\mathfrak{D}(b,c)$ (Triangle
Inequality)\,. 
\end{enumerate}
Our main tool is a simplified version of Schoenberg's theorem~\cite{Schoenberg1938},
reproduced as follows. \begin{theorem}[Schoenberg]\label{thm:schoenberg}
Let $\mathcal{X}$ be a nonempty set and $K:\mathcal{X}\times\mathcal{X}\to\mathbb{R}$
a function such that $K(x,y)=K(y,x)$ and $K(x,y)\geq0$ for all $x,y\in\mathcal{X}$,
with saturation iff $x=y$. If the implication 
\begin{equation}
\sum_{i=1}^{n}{c_{i}}=0\Rightarrow\sum_{i,j=1}^{n}{K(x_{i},x_{j})c_{i}c_{j}}\leq0\label{eq:ndk}
\end{equation}
holds for all $n\geq2$, $\{x_{1},\ldots,x_{n}\}\subseteq\mathcal{X}$
and $\{c_{1},\ldots,c_{n}\}\subseteq\mathbb{R}$, then $\sqrt{K}$
is a metric. \end{theorem}

The theorem has previously been used in~\cite{mendonca2008} to show
that $C[\Fn]$ is a metric for the space of density matrices. However,
the alternative proof for the metric properties of $B[\Fuj(\rho,\sigma)]$
and $C[\Fuj(\rho,\sigma)]$ given in~\cite{mendonca2008} is flawed
due to an erroneous application of the above theorem.

%% Metric property of functionals of F_2

\subsection*{$\Fmax$ metric properties}

Here, we will show that $C[\mathcal{F}_{2}(\rho,\sigma)]$ is a metric
for the space of density matrices. Because $\mathcal{F}_{2}$ complies
with Jozsa's axioms (J1) and (J2), it is immediate that $C[\mathcal{F}_{2}]$
is non-negative (M1), fulfills the indiscernible identity (M2) and
is symmetric (M3). Hence, in order to establish $C[\mathcal{F}_{2}]$
as a genuine metric, we only have to prove that it satisfies the triangle
inequality 
\begin{equation}
C[\mathcal{F}_{2}(\rho,\sigma)]\leq C[\mathcal{F}_{2}(\rho,\tau)]+C[\mathcal{F}_{2}(\tau,\sigma)]\label{eq:triangineqrhosigmatau}
\end{equation}
for \emph{arbitrary} $d$-dimensional density matrices $\rho$, $\sigma$
and $\tau$. To this end, we shall first prove the following lemma:
% Lemma
\begin{lemma}\label{Lemma:trigonometric} For any $\vartheta,\varphi\in[0,2\pi]$
and $p,q\in[0,1]$ 
\begin{eqnarray}
\sqrt{1-pq\cos(\vartheta+\varphi)} & \leq\sqrt{1-p\cos\vartheta}+\sqrt{1-q\cos\varphi},\nonumber \\
\sqrt{1-pq\cos(\vartheta-\varphi)} & \geq\sqrt{1-p\cos\vartheta}-\sqrt{1-q\cos\varphi}\nonumber \\
\,\label{eq:lemmaineq1}
\end{eqnarray}
\end{lemma}

% Proof of Lemma
\begin{proof} We start by showing that both inequalities hold in
the special case of $p=q=1$. To see this, note that 
\begin{eqnarray}
\sqrt{1-\cos(\vartheta+\varphi)} & =\sqrt{2}\left|\sin\frac{\vartheta}{2}\cos\frac{\varphi}{2}+\sin\frac{\varphi}{2}\cos\frac{\vartheta}{2}\right|\nonumber \\
 & \leq\sqrt{2}\left|\sin\frac{\vartheta}{2}\cos\frac{\varphi}{2}\right|+\sqrt{2}\left|\sin\frac{\varphi}{2}\cos\frac{\vartheta}{2}\right|\nonumber \\
 & \leq\sqrt{2}\left|\sin\frac{\vartheta}{2}\right|+\sqrt{2}\left|\sin\frac{\varphi}{2}\right|\nonumber \\
 & =\sqrt{1-\cos\vartheta}+\sqrt{1-\cos\varphi}
\end{eqnarray}
where we have used the trigonometric identity $\sqrt{1-\cos{x}}=\sqrt{2}|\sin{\frac{x}{2}}|$
in the first and last lines, the inequality $|x+y|\leq|x|+|y|$ in
the second line, and the fact that cosines are upper bounded by $1$
in the third line.

To see that inequality (\ref{eq:lemmaineq1}) also holds for $p=q=1$,
we rewrite it in the equivalent form 
\begin{equation}
\left|\sin\left(\frac{\vartheta-\varphi}{2}\right)\right|\geq\sin\left(\frac{\vartheta}{2}\right)-\sin\left(\frac{\varphi}{2}\right)\,,\label{eq:ineqsin}
\end{equation}
which is clearly valid if $\sin\frac{\vartheta}{2}\leq\sin\frac{\varphi}{2}$.
If instead $\sin\frac{\vartheta}{2}>\sin\frac{\varphi}{2}$, then
both sides of inequality (\ref{eq:ineqsin}) are non-negative and
can thus be squared to yield an equivalent inequality that can be
simplified to 
\begin{equation}
4\sin\left(\frac{\vartheta}{2}\right)\sin\left(\frac{\varphi}{2}\right)\sin^{2}\left(\frac{\vartheta-\varphi}{4}\right)\geq0.
\end{equation}
Since this clearly holds for $\vartheta,\varphi\in[0,2\pi]$, inequality~(\ref{eq:lemmaineq1})
for $p=q=1$ must also hold.

In order to show Lemma~\ref{Lemma:trigonometric} for general $p,q\in[0,1]$,
let us {\em define} the angles $\Theta,\Phi\in[0,\pi]$ as follows
\begin{eqnarray}
\cos\Theta & =p\cos\vartheta,\quad\sin\Theta=\sqrt{1-p^{2}\cos^{2}\vartheta}\,,\nonumber \\
\cos\Phi & =q\cos\varphi,\quad\sin\Phi=\sqrt{1-q^{2}\cos^{2}\varphi}\,.
\end{eqnarray}
This gives 
\begin{eqnarray}
 & \pm\cos(\Theta\pm\Phi)\nonumber \\
= & \pm\left(pq\cos\vartheta\cos\varphi\mp\sqrt{1-p^{2}\cos^{2}\vartheta}\sqrt{1-q^{2}\cos^{2}\varphi}\right)\nonumber \\
\leq & \pm\left(pq\cos\vartheta\cos\varphi\mp pq\sqrt{1-\cos^{2}\vartheta}\sqrt{1-\cos^{2}\varphi}\right)\nonumber \\
= & \pm pq(\cos\vartheta\cos\varphi\mp|\sin\vartheta\sin\varphi|)\nonumber \\
\leq & \pm pq\cos(\vartheta\pm\varphi).\label{eq:cossumrel}
\end{eqnarray}

Inequalities (\ref{eq:lemmaineq1}) can now be obtained from inequality
(\ref{eq:cossumrel}) as follows: 
\begin{eqnarray*}
\pm\sqrt{1-pq\cos(\vartheta\pm\varphi)} & \leq\pm\sqrt{1-\cos(\Theta\pm\Phi)}\\
 & \leq\pm\left(\sqrt{1-\cos\Theta}\pm\sqrt{1-\cos\Phi}\right)\\
 & =\pm\left(\sqrt{1-p\cos{\vartheta}}\pm\sqrt{1-q\cos{\varphi}}\right),
\end{eqnarray*}
where the second inequality follows from the already verified inequality
of (\ref{eq:lemmaineq1}) in the case of $p=q=1$. \end{proof}

\begin{theorem} $C[\mathcal{F}_{2}(\rho,\sigma)]$ is a metric for
the space of density matrices \end{theorem} \begin{proof} Using
equations~\eref{eq:param} and \eref{eq:f2-geom}, the triangle
inequality of \eref{eq:triangineqrhosigmatau} can be rewritten
as 
\begin{equation}
\sqrt{1-\frac{\vec{r}\cdot\vec{s}}{\mu\left(\vec{r},\vec{s}\right)}}\leq\sqrt{1-\frac{\vec{r}\cdot\vec{t}}{\mu\left(\vec{r},\vec{t}\right)}}+\sqrt{1-\frac{\vec{t}\cdot\vec{s}}{\mu\left(\vec{t},\vec{s}\right)}},\label{Ineq:triangle}
\end{equation}
where $\mu\left(\vec{r},\vec{s}\right)\equiv\max(r^{2},s^{2})$ and
the entries of the vectors $\vec{r},\vec{s},\vec{t}\in\mathbb{R}^{d^{2}}$
are the expansion coefficients of $\rho,\sigma,$ and $\tau$, respectively,
in some basis of Hermitian matrices.

Without loss of generality, we henceforth assume that $\rho$ and
$\tau$ are, respectively, the density matrices of maximal and minimal
purities (i.e., $r\geq s\geq t$). Hence, \eref{eq:triangineqrhosigmatau}
unfolds into three inequalities to be proven, which correspond to
$C[\mathcal{F}_{2}(\rho,\sigma)]\leq C[\mathcal{F}_{2}(\rho,\tau)]+C[\mathcal{F}_{2}(\tau,\sigma)]$,
$C[\mathcal{F}_{2}(\sigma,\tau)]\leq C[\mathcal{F}_{2}(\rho,\sigma)]+C[\mathcal{F}_{2}(\rho,\tau)]$
and $C[\mathcal{F}_{2}(\rho,\tau)]\leq C[\mathcal{F}_{2}(\rho,\sigma)]+C[\mathcal{F}_{2}(\tau,\sigma)]$
respectively. In terms of the expansion coefficients vectors $\vec{r}$,
$\vec{s}$, and $\vec{t}$, these be written as:

\begin{eqnarray}
\sqrt{1-\frac{s}{r}\cos\theta_{rs}} & \leq\sqrt{1-\frac{t}{r}\cos\theta_{rt}}+\sqrt{1-\frac{t}{s}\cos\theta_{ts}}\,,\nonumber \\
\sqrt{1-\frac{t}{s}\cos\theta_{ts}} & \leq\sqrt{1-\frac{t}{r}\cos\theta_{rt}}+\sqrt{1-\frac{s}{r}\cos\theta_{rs}}\,,\nonumber \\
\sqrt{1-\frac{t}{r}\cos\theta_{rt}} & \leq\sqrt{1-\frac{s}{r}\cos\theta_{rs}}+\sqrt{1-\frac{t}{s}\cos\theta_{ts}}\,,\nonumber \\
\,\label{eq:ineqtoproveangles3}
\end{eqnarray}
where the angles $\theta_{rs}$, $\theta_{rt}$, and $\theta_{ts}\in[0,\pi]$
have been {\em defined} (in accordance with the Cauchy-Schwarz
inequality) as follows: 
\begin{equation}
\cos\theta_{rs}=\frac{\vec{r}\cdot\vec{s}}{rs}\,,\quad\cos\theta_{rt}=\frac{\vec{r}\cdot\vec{t}}{rt}\,,\quad\cos\theta_{ts}=\frac{\vec{t}\cdot\vec{s}}{ts}\,.
\end{equation}

Besides, since 
\begin{equation}
\cos(\theta_{rs}+\theta_{ts})\leq\cos\theta_{rt}\leq\cos(\theta_{rs}-\theta_{ts})\,,\label{eq:costhetartineq}
\end{equation}
(a demonstration of which will be given at the end of this proof),
the replacement of $\cos\theta_{rt}$ with either $\cos(\theta_{rs}-\theta_{ts})$
or $\cos(\theta_{rs}+\theta_{ts})$ in inequality (\ref{eq:ineqtoproveangles3}),
yields the following \emph{stronger} set of inequalities to be proven,
where we define $\Delta^{\pm}=\theta_{rs}\pm\theta_{ts}$:

\begin{eqnarray}
\sqrt{1-\frac{t}{r}\cos(\Delta^{-})} & \geq\sqrt{1-\frac{s}{r}\cos\theta_{rs}}-\sqrt{1-\frac{t}{s}\cos\theta_{ts}}\nonumber \\
\sqrt{1-\frac{t}{r}\cos(\Delta^{-})} & \geq\sqrt{1-\frac{t}{s}\cos\theta_{ts}}-\sqrt{1-\frac{s}{r}\cos\theta_{rs}}\nonumber \\
\sqrt{1-\frac{t}{r}\cos(\Delta^{+})} & \leq\sqrt{1-\frac{s}{r}\cos\theta_{rs}}+\sqrt{1-\frac{t}{s}\cos\theta_{ts}}\,\nonumber \\
\,\label{eq:Ineq:stronger}
\end{eqnarray}

Note that $\frac{t}{r}=\frac{s}{r}\cdot\frac{t}{s}$ and $\frac{t}{r},\frac{s}{r},\frac{t}{s}\in(0,1]$
by our assumption that $r\ge s\ge t>0$. Thus, through appropriate
identifications of these ratios with $p,q$ and applying Lemma~\ref{Lemma:trigonometric},
inequalities~(\ref{eq:Ineq:stronger}), and hence the desired triangle
inequalities~(\ref{Ineq:triangle}) follow.

We complete this proof with a demonstration of inequalities (\ref{eq:costhetartineq}).
Consider, first, the following pair of vectors, each of which being
orthogonal to $\vec{s}$ 
\begin{equation}
\vec{u}=\vec{r}-\frac{r}{s}\cos\theta_{rs}\vec{s}\quad\mbox{and}\quad\vec{v}=\vec{t}-\frac{t}{s}\cos\theta_{ts}\vec{s}.
\end{equation}
Using these and the orthogonality of $\vec{u}$, $\vec{v}$ to $\vec{s}$,
we may expand the scalar product between $\vec{r}$ and $\vec{t}$
as 
\begin{equation}
\vec{r}\cdot\vec{t}=rt\cos\theta_{rt}=\vec{u}\cdot\vec{v}+rt\cos\theta_{rs}\cos\theta_{ts}.
\end{equation}
Then, using $\vec{u}\cdot\vec{v}=uv\cos\theta_{uv}$ with $u=r\sin\theta_{rs}$,
$v=t\sin\theta_{ts}$, and $\theta_{uv}\in[0,\pi]$, we arrive at
\begin{equation}
\cos\theta_{rt}=\sin\theta_{rs}\sin\theta_{ts}\cos\theta_{uv}+\cos\theta_{rs}\cos\theta_{ts}.
\end{equation}
From here, the trivial inequalities $\cos\theta_{uv}\geq-1$ and $\cos\theta_{uv}\leq1$
imply, respectively, the upper and lower bounds on $\cos\theta_{rt}$
in (\ref{eq:costhetartineq}). \end{proof}

\subsection*{$\Fc$ metric properties}

\begin{theorem} $C[\Fc(\rho,\sigma)]$ is a metric for the space
of density matrices of fixed Hilbert space dimension. \end{theorem}

\begin{proof} For density matrices $\rho_{i}$ acting on $d$-dimensional
complex Hilbert space $\mathbb{C}^{d}$, note that $r=\frac{1}{d-1}$
is a constant. Thus, for any $c_{i}'s\in\mathbb{R}$ and which are
such that $\sum_{i}c_{i}=0$, we have 
\begin{eqnarray}
 & \sum_{i,j}c_{i}c_{j}\left[1-\Fc(\rho_{i},\rho_{j})\right]\nonumber \\
= & -\sum_{i,j}c_{i}c_{j}\left[\frac{1-r}{2}+\frac{1+r}{2}\Fn(\rho_{i},\rho_{j})\right]\nonumber \\
= & -\frac{1+r}{2}\sum_{i,j}c_{i}c_{j}\,\Fn(\rho_{i},\rho_{j}).
\end{eqnarray}
{From the proof of} the metric property of $C[\Fn(\rho,\sigma)]$
{given in}~\cite{mendonca2008}, we know that the last expression
above must be non-positive. Hence, $\sqrt{1-\Fc(\rho,\sigma)}$ is
a metric for the space of density matrices of fixed dimension. \end{proof}

\subsection*{$\Fgm$ metric properties}

\begin{theorem} $C[\Fgm(\rho,\sigma)]$ is a metric for the space
of density matrices. \end{theorem}

\begin{proof} For any $c_{i}'s\in\mathbb{R}$ and which are such
that $\sum_{i}c_{i}=0$, note that 
\begin{eqnarray}
 & \sum_{i,j}c_{i}c_{j}\left[1-\Fgm(\rho_{i},\rho_{j})\right]=-\sum_{i,j}c_{i}c_{j}\frac{\tr(\rho_{i}\,\rho_{j})}{\sqrt{\tr(\rho_{i}^{2})\tr(\rho_{j}^{2})}}\nonumber \\
 & =-\tr\left[\left(\sum_{i}\frac{c_{i}\rho_{i}}{\sqrt{\tr(\rho_{i}^{2})}}\right)^{2}\right]\le0.
\end{eqnarray}
Hence, $C[\Fgm(\rho,\sigma)]=\sqrt{1-\Fgm(\rho,\sigma)}$ is a metric
for the space of density matrices.

\end{proof}

\subsection*{$\Fa$ metric properties}

The metric properties of $B[\Fa(\rho,\sigma)]$ were first mentioned
in~\cite{Raggio1984}, while those of $A[\Fa(\rho,\sigma)]$ and
$C[\Fa(\rho,\sigma)]$ were numerically investigated in~\cite{ma2008geometric},
suggesting that they may indeed be metrics for the space of density
matrices. In Section C.3,~\cite{mendonca2008}, it was briefly mentioned
that both $B[\Fa(\rho,\sigma)]$ and $C[\Fa(\rho,\sigma)]$ can be
proved to be a metric using Schoenberg's theorem. Here, we shall provide
an explicit proof of these facts. An alternative proof for $B[\Fa(\rho,\sigma)]$
can also be found in~\cite{ma2011}.

\begin{proof} For any $c_{i}'s\in\mathbb{R}$ and which are such
that $\sum_{i}c_{i}=0$, note that 
\begin{eqnarray}
 & \sum_{i,j}c_{i}c_{j}\left[1-\sqrt{\Fa(\rho_{i},\rho_{j})}\right]=-\sum_{i,j}c_{i}c_{j}\tr\left(\sqrt{\rho_{i}}\sqrt{\rho_{j}}\right)\nonumber \\
 & =-\tr\left[\left|\sum_{i}c_{i}\sqrt{\rho_{i}}\right|^{2}\right]\le0.
\end{eqnarray}
Likewise, we can show that: 
\begin{eqnarray}
 & \sum_{i,j}c_{i}c_{j}\left[1-\Fa(\rho_{i},\rho_{j})\right]=-\sum_{i,j}c_{i}c_{j}\left[\tr\left(\sqrt{\rho_{i}}\sqrt{\rho_{j}}\right)\right]^{2}\nonumber \\
= & -\sum_{i,j}c_{i}c_{j}\tr\left[\left(\sqrt{\rho_{i}}\otimes\sqrt{\rho_{i}}\right)\left(\sqrt{\rho_{j}}\otimes\sqrt{\rho_{j}}\right)\right]\nonumber \\
= & -\tr\left[\left|\sum_{i}c_{i}\sqrt{\rho_{i}}\otimes\sqrt{\rho_{i}}\right|^{2}\right]\le0.
\end{eqnarray}
This concludes our proof for the metric properties of $B[\Fa(\rho,\sigma)]$
and $C[\Fa(\rho,\sigma)]$. \end{proof}

\section*{References}

 \bibliographystyle{unsrt}
\bibliography{survey}

\begin{thebibliography}{100}

\bibitem{Guehne2009}
O.~G{\"u}hne and G.~T{\'o}th.
\newblock Entanglement detection.
\newblock {\em Phys. Rep.}, 474(1):1 -- 75, 2009.

\bibitem{Rosset2012}
D.~Rosset, R.~Ferretti-Sch\"obitz, J.-D. Bancal, N.~Gisin, and Y.-C. Liang.
\newblock Imperfect measurement settings: Implications for quantum state
  tomography and entanglement witnesses.
\newblock {\em Phys. Rev. A}, 86:062325, Dec 2012.

\bibitem{gu2010fidelity}
S.-J. Gu.
\newblock Fidelity approach to quantum phase transitions.
\newblock {\em Int. J. Mod. Phys. B}, 24(23):4371--4458, 2010.

\bibitem{Hayden2006}
P.~Hayden, D.~W. Leung, and A.~Winter.
\newblock Aspects of generic entanglement.
\newblock {\em Commun. Math. Phys.}, 265(1):95--117, Jul 2006.

\bibitem{WoottersClone}
W.~K. Wootters and W.~H. Zurek.
\newblock A single quantum cannot be cloned.
\newblock {\em Nature}, 299(5886):802--803, 1982.

\bibitem{BuzekHilleryPhysRevA.54.1844}
V.~Buzek and M.~Hillery.
\newblock Quantum copying: Beyond the no-cloning theorem.
\newblock {\em Phys. Rev. A}, 54:1844--1852, 1996.

\bibitem{ScaraniRMP.77.1225}
V.~Scarani, S.~Iblisdir, N.~Gisin, and A.~Ac\'{\i}n.
\newblock Quantum cloning.
\newblock {\em Rev. Mod. Phys.}, 77:1225--1256, Nov 2005.

\bibitem{jozsa1994fidelity}
R.~Jozsa.
\newblock Fidelity for mixed quantum states.
\newblock {\em J. Mod. Opt.}, 41(12):2315--2323, 1994.

\bibitem{bennett1993}
C.~H. Bennett, G.~Brassard, C.~Cr{\'e}peau, R.~Jozsa, A.~Peres, and W.~K.
  Wootters.
\newblock Teleporting an unknown quantum state via dual classical and
  {Einstein-Podolsky-Rosen} channels.
\newblock {\em Phys. Rev. Lett.}, 70:1895, 1993.

\bibitem{Lvovsky:2009aa}
A.~I. Lvovsky, B.~C. Sanders, and W.~Tittel.
\newblock Optical quantum memory.
\newblock {\em Nat. Photonics}, 3:706 EP --, 12 2009.

\bibitem{nielsen2000quantum}
M.~A. Nielsen and I.~L. Chuang.
\newblock {\em Quantum Computation and Information}.
\newblock Cambridge University Press Cambridge, 2000.

\bibitem{GHZ}
D.~M. Greenberger, M.~A. Horne, and A.~Zeilinger.
\newblock {\em Bell's {T}heorem, {Q}uantum {Theory}, and {C}onceptions of the
  {U}niverse}, pages 69--72.
\newblock Kluwer, Dordrecht, 1989.

\bibitem{monz201114}
T.~Monz, P.~Schindler, J.~T. Barreiro, M.~Chwalla, D.~Nigg, W.~A. Coish,
  M.~Harlander, W.~H{\"a}nsel, M.~Hennrich, and R.~Blatt.
\newblock 14-qubit entanglement: Creation and coherence.
\newblock {\em Phys. Rev. Lett.}, 106(13):130506, 2011.

\bibitem{song201710}
C.~Song, K.~Xu, W.~Liu, C.-P. Yang, S.-B. Zheng, H.~Deng, Q.~Xie, K.~Huang,
  Q.~Guo, L.~Zhang, et~al.
\newblock 10-qubit entanglement and parallel logic operations with a
  superconducting circuit.
\newblock {\em Phys. Rev. Lett.}, 119(18):180511, 2017.

\bibitem{chen2017observation}
L.-K. Chen, Z.-D. Li, X.-C. Yao, M.~Huang, W.~Li, H.~Lu, X.~Yuan, Y.-B. Zhang,
  X.~Jiang, C.-Z. Peng, et~al.
\newblock Observation of ten-photon entanglement using thin {BiB$_3$O$_6$}
  crystals.
\newblock {\em Optica}, 4(1):77--83, 2017.

\bibitem{reid2014quantum}
M.~D. Reid, B.~Opanchuk, L.~Rosales-Z{\'a}rate, and P.~D. Drummond.
\newblock Quantum probabilistic sampling of multipartite 60-qubit
  {B}ell-inequality violations.
\newblock {\em Phys. Rev. A}, 90(1):012111, 2014.

\bibitem{galve2017microscopic}
F.~Galve, A.~Mandarino, M.~G.~A. Paris, C.~Benedetti, and R.~Zambrini.
\newblock Microscopic description for the emergence of collective dissipation
  in extended quantum systems.
\newblock {\em Sci. Rep.}, 7:42050, 2017.

\bibitem{leibfried2003experimental}
D.~Leibfried, B.~DeMarco, V.~Meyer, D.~Lucas, M.~Barrett, J.~Britton, W.~M.
  Itano, B.~Jelenkovi{\'c}, C.~Langer, T.~Rosenband, et~al.
\newblock Experimental demonstration of a robust, high-fidelity geometric two
  ion-qubit phase gate.
\newblock {\em Nature}, 422(6930):412, 2003.

\bibitem{schumacher1995quantum}
B.~Schumacher.
\newblock Quantum coding.
\newblock {\em Phys. Rev. A}, 51(4):2738, 1995.

\bibitem{mendonca2008}
P.~E. M.~F. Mendonca, R.~d.~J. Napolitano, M.~A. Marchiolli, C.~J. Foster, and
  Y.-C. Liang.
\newblock Alternative fidelity measure between quantum states.
\newblock {\em Phys. Rev. A}, 78:052330, 2008.

\bibitem{caves1994quantum}
C.~M. Caves and P.~D. Drummond.
\newblock Quantum limits on bosonic communication rates.
\newblock {\em Rev. Mod. Phys.}, 66(2):481, 1994.

\bibitem{gilchrist2005distance}
A.~Gilchrist, N.~K. Langford, and M.~A. Nielsen.
\newblock Distance measures to compare real and ideal quantum processes.
\newblock {\em Phys. Rev. A}, 71(6):062310, 2005.

\bibitem{uhlmann1976transition}
A.~Uhlmann.
\newblock The ``transition probability'' in the state space of a*-algebra.
\newblock {\em Rep. Math. Phys.}, 9(2):273--279, 1976.

\bibitem{luo2004informational}
S.~Luo and Q.~Zhang.
\newblock Informational distance on quantum-state space.
\newblock {\em Phys. Rev. A}, 69(3):032106, 2004.

\bibitem{audenaert2008asymptotic}
K.~M.~R. Audenaert, M.~Nussbaum, A.~Szko{\l}a, and F.~Verstraete.
\newblock Asymptotic error rates in quantum hypothesis testing.
\newblock {\em Commun. Math. Phys.}, 279(1):251--283, 2008.

\bibitem{Hu:JMP:2006}
X.~Hu and Z.~Ye.
\newblock Generalized quantum entropy.
\newblock {\em J. Math. Phys.}, 47(2):023502, 2006.

\bibitem{Muller:JMP:2013}
M.~M{\"u}ller-Lennert, F.~Dupuis, O.~Szehr, S.~Fehr, and M.~Tomamichel.
\newblock On quantum {R{\'e}nyi} entropies: A new generalization and some
  properties.
\newblock {\em J. Math. Phys.}, 54(12):122203, 2013.

\bibitem{Dupuis:2013}
F.~Dupuis, L.~Kr{\"a}mer, P.~Faist, J.~M. Renes, and R.~Renner.
\newblock Generalized entropies.
\newblock In {\em XVIIth International Congress on Mathematical Physics}, pages
  134--153. World Scientific, 2013.

\bibitem{chen2002alternative}
J.-L. Chen, L.~Fu, A.~A. Ungar, and X.~Zhao.
\newblock Alternative fidelity measure between two states of an n-state quantum
  system.
\newblock {\em Phys. Rev. A}, 65(5):054304, 2002.

\bibitem{Kimura2003}
G.~Kimura.
\newblock The bloch vector for n-level systems.
\newblock {\em Phys. Lett. A}, 314(5):339 -- 349, 2003.

\bibitem{Byrd2003}
M.~S. Byrd and N.~Khaneja.
\newblock Characterization of the positivity of the density matrix in terms of
  the coherence vector representation.
\newblock {\em Phys. Rev. A}, 68:062322, Dec 2003.

\bibitem{miszczak2009sub}
J.~A. Miszczak, Z.~Pucha{\l}a, P.~Horodecki, A.~Uhlmann, and K.~Zyczkowski.
\newblock Sub-and super-fidelity as bounds for quantum fidelity.
\newblock {\em Quantum Inf. Comput.}, 9(1):103--130, 2009.

\bibitem{ma2008geometric}
Z.~Ma, F.-L. Zhang, and J.-L. Chen.
\newblock Geometric interpretation for the {A} fidelity and its relation with
  the {B}ures fidelity.
\newblock {\em Phys. Rev. A}, 78(6):064305, 2008.

\bibitem{Raggio1984}
G.~A. Raggio.
\newblock Generalized transition probabilities and applications.
\newblock In {\em Quantum Probability and Applications to the Quantum Theory of
  Irreversible Processes}, pages 327--335. Springer, 1984.

\bibitem{wang2008alternative}
X.~Wang, C.-S. Yu, and X.X. Yi.
\newblock An alternative quantum fidelity for mixed states of qudits.
\newblock {\em Phys. Lett. A}, 373(1):58--60, 2008.

\bibitem{audenaert2007discriminating}
K.~M.~R. Audenaert, J.~Calsamiglia, R.~Munoz-Tapia, E.~Bagan, L.l. Masanes,
  A.~Ac\'{\i}n, and F.~Verstraete.
\newblock Discriminating states: the quantum {C}hernoff bound.
\newblock {\em Phys. Rev. Lett.}, 98(16):160501, 2007.

\bibitem{bhatia97}
R.~Bhatia.
\newblock {\em Matrix Analysis}, volume 169.
\newblock Springer, 1997.

\bibitem{renyi1961measures}
A.~R{\'e}nyi.
\newblock On measures of entropy and information.
\newblock Technical report, Hungarian Academy of Science Budapest Hungary,
  1961.

\bibitem{hastings2010measuring}
M.~B. Hastings, I.~Gonz{\'a}lez, A.~B. Kallin, and R.~G. Melko.
\newblock Measuring {Renyi} entanglement entropy in quantum monte carlo
  simulations.
\newblock {\em Phys. Rev. Lett.}, 104(15):157201, 2010.

\bibitem{preskill2015lecture}
J.~Preskill.
\newblock {\em Lecture Notes for Physics 229:Quantum Information and
  Computation}.
\newblock CreateSpace Independent Publishing Platform, 2015.

\bibitem{Carlen2008}
E.~A. Carlen and E.~H. Lieb.
\newblock A {M}inkowski type trace inequality and strong subadditivity of
  quantum entropy ii: Convexity and concavity.
\newblock {\em Lett. Math. Phys.}, 83(2):107--126, Feb 2008.

\bibitem{Hubner1992}
M.~{H{\"u}bner}.
\newblock {Explicit computation of the Bures distance for density matrices}.
\newblock {\em Phys. Lett. A}, 163:239--242, March 1992.

\bibitem{Rastegin06}
A.~E. Rastegin.
\newblock Sine distance for quantum states.
\newblock arXiv: quant-ph/0602112.

\bibitem{ma2009pla}
Z.~Ma, F.-L. Zhang, and J.-L. Chen.
\newblock Fidelity induced distance measures for quantum states.
\newblock {\em Phys. Lett. A}, 373(38):3407 -- 3409, 2009.

\bibitem{raggio1982comparison}
G.~A. Raggio.
\newblock Comparison of {Uhlmann's} transition probability with the one induced
  by the natural positive cone of von {N}eumann algebras in standard form.
\newblock {\em Lett. Math. Phys.}, 6(3):233--236, 1982.

\bibitem{ginibre1965statistical}
J.~Ginibre.
\newblock Statistical ensembles of complex, quaternion, and real matrices.
\newblock {\em J. Math. Phys.}, 6(3):440--449, 1965.

\bibitem{zyczkowski2011generating}
K.~{\.Z}yczkowski, K.~A Penson, I.~Nechita, and B.~Collins.
\newblock Generating random density matrices.
\newblock {\em J. Math. Phys.}, 52(6):062201, 2011.

\bibitem{zyczkowski2005average}
K.~{\.Z}yczkowski and H.-J. Sommers.
\newblock Average fidelity between random quantum states.
\newblock {\em Phys. Rev. A}, 71(3):032313, 2005.

\bibitem{caves1981quantum}
C.~M. Caves.
\newblock Quantum-mechanical noise in an interferometer.
\newblock {\em Phys. Rev. D}, 23(8):1693, 1981.

\bibitem{walls1983squeezed}
D.~F. Walls.
\newblock Squeezed states of light.
\newblock {\em Nature}, 306(5939):141--146, 1983.

\bibitem{EPR}
A.~Einstein, B.~Podolsky, and N.~Rosen.
\newblock Can quantum-mechanical description of physical reality be considered
  complete?
\newblock {\em Phys. Rev.}, 47:777--780, May 1935.

\bibitem{wiseman2007}
H.~M. Wiseman, S.~J. Jones, and A.~C. Doherty.
\newblock Steering, entanglement, nonlocality, and the
  {Einstein-Podolsky-Rosen} paradox.
\newblock {\em Phys. Rev. Lett.}, 98:140402, Apr 2007.

\bibitem{reid2009colloquium}
M.D. Reid, P.D. Drummond, W.P. Bowen, E.~G. Cavalcanti, P.~K. Lam, H.A. Bachor,
  U.~L. Andersen, and G.~Leuchs.
\newblock Colloquium: the {Einstein-Podolsky-Rosen} paradox: from concepts to
  applications.
\newblock {\em Rev. Mod. Phys.}, 81(4):1727, 2009.

\bibitem{ma2017proposal}
Y.~Ma, H.~Miao, B.~H. Pang, M.~Evans, C.~Zhao, J.~Harms, R.~Schnabel, and
  Y.~Chen.
\newblock Proposal for gravitational-wave detection beyond the standard quantum
  limit through {EPR} entanglement.
\newblock {\em Nat. Phys.}, 13(8):776, 2017.

\bibitem{abbott2016observation}
B.~P. Abbott, R.~Abbott, T.D. Abbott, M.R. Abernathy, F.~Acernese, K.~Ackley,
  C.~Adams, T.~Adams, P.~Addesso, R.X. Adhikari, et~al.
\newblock Observation of gravitational waves from a binary black hole merger.
\newblock {\em Phys. Rev. Lett.}, 116(6):061102, 2016.

\bibitem{einstein1918gravitationswellen}
A.~Einstein.
\newblock {\"U}ber gravitationswellen.
\newblock {\em Sitzber. K. Preuss. Aka.}, 1918.

\bibitem{gisin2002quantum}
N.~Gisin, G.~Ribordy, W.~Tittel, and H.~Zbinden.
\newblock Quantum cryptography.
\newblock {\em Rev. Mod. Phys.}, 74(1):145, 2002.

\bibitem{steane1998quantum}
A.~Steane.
\newblock Quantum computing.
\newblock {\em Rep. Prog. Phys.}, 61(2):117, 1998.

\bibitem{fialko2015fate}
O.~Fialko, B.~Opanchuk, A.~Sidorov, P.~D. Drummond, and J.~Brand.
\newblock Fate of the false vacuum: Towards realization with ultra-cold atoms.
\newblock {\em Europhys. Lett.}, 110(5):56001, 2015.

\bibitem{bernien2017probing}
H.~Bernien, S.~Schwartz, A.~Keesling, H.~Levine, A.~Omran, H.~Pichler, S.~Choi,
  A.~S. Zibrov, M.~Endres, M.~Greiner, et~al.
\newblock Probing many-body dynamics on a 51-atom quantum simulator.
\newblock {\em Nature}, 551(7682):579, 2017.

\bibitem{schumacher1996quantum}
B.~Schumacher and M.~A. Nielsen.
\newblock Quantum data processing and error correction.
\newblock {\em Phys. Rev. A}, 54(4):2629, 1996.

\bibitem{millen2016perspective}
J.~Millen and A.~Xuereb.
\newblock Perspective on quantum thermodynamics.
\newblock {\em New. J. Phys.}, 18(1):011002, 2016.

\bibitem{Yang2014}
T.~H. Yang, T.~V\'ertesi, J.-D. Bancal, V.~Scarani, and M.~Navascu\'es.
\newblock Robust and versatile black-box certification of quantum devices.
\newblock {\em Phys. Rev. Lett.}, 113:040401, Jul 2014.

\bibitem{Kaniewski2016}
J.~Kaniewski.
\newblock Analytic and nearly optimal self-testing bounds for the
  {Clauser-Horne-Shimony-Holt} and {Mermin} inequalities.
\newblock {\em Phys. Rev. Lett.}, 117:070402, Aug 2016.

\bibitem{Sekatski1802}
P.~Sekatski, J.-D. Bancal, S.~Wagner, and Sangouard N.
\newblock Certifying the building blocks of quantum computers from bell's
  theorem.
\newblock eprint arXiv:1802.02170, February 2018.

\bibitem{Chen:PRL:2016}
S.-L. Chen, C.~Budroni, Y.-C. Liang, and Y.-N. Chen.
\newblock Natural framework for device-independent quantification of quantum
  steerability, measurement incompatibility, and self-testing.
\newblock {\em Phys. Rev. Lett.}, 116:240401, Jun 2016.

\bibitem{Cavalcanti:PRA:2016}
D.~Cavalcanti and P.~Skrzypczyk.
\newblock Quantitative relations between measurement incompatibility, quantum
  steering, and nonlocality.
\newblock {\em Phys. Rev. A}, 93:052112, May 2016.

\bibitem{Rosset:PRX:2018}
D.~Rosset, F.~Buscemi, and Y.-C. Liang.
\newblock Resource theory of quantum memories and their faithful verification
  with minimal assumptions.
\newblock {\em Phys. Rev. X}, 8:021033, May 2018.

\bibitem{knill2008randomized}
E.~Knill, D.~Leibfried, R.~Reichle, J.~Britton, R.B. Blakestad, J.~D. Jost,
  C.~Langer, R.~Ozeri, S.~Seidelin, and D.~J. Wineland.
\newblock Randomized benchmarking of quantum gates.
\newblock {\em Phys. Rev. A}, 77(1):012307, 2008.

\bibitem{harty2014high}
T.P. Harty, D.T.C. Allcock, C.J. Ballance, L.~Guidoni, H.A. Janacek, N.M.
  Linke, D.N. Stacey, and D.M. Lucas.
\newblock High-fidelity preparation, gates, memory, and readout of a
  trapped-ion quantum bit.
\newblock {\em Phys. Rev. Lett.}, 113(22):220501, 2014.

\bibitem{horodecki1999}
M.~Horodecki, P.~Horodecki, and R.~Horodecki.
\newblock General teleportation channel, singlet fraction, and
  quasidistillation.
\newblock {\em Phys. Rev. A}, 60:1888--1898, Sep 1999.

\bibitem{nocloningtheorem}
W.~K. Wootters and W.~H. Zurek.
\newblock A single quantum cannot be cloned.
\newblock {\em Nature}, 299(5886):802--803, 1982.

\bibitem{quantum-internet}
H.~J. Kimble.
\newblock The quantum internet.
\newblock {\em Nature}, 453(7198):1023, 2008.

\bibitem{popescu1994}
S.~Popescu.
\newblock Bell's inequalities versus teleportation: What is nonlocality?
\newblock {\em Phys. Rev. Lett.}, 72:797--799, Feb 1994.

\bibitem{hillery-buzek-clone}
V.~Bu{\v{z}}ek and M.~Hillery.
\newblock Quantum copying: Beyond the no-cloning theorem.
\newblock {\em Phys. Rev. A}, 54(3):1844, 1996.

\bibitem{classical-fidelity}
N.~Gisin and S.~Massar.
\newblock Optimal quantum cloning machines.
\newblock {\em Phys. Rev. Lett.}, 79(11):2153, 1997.

\bibitem{Bruss1998}
D.~Bruss, A.~Ekert, and C.~Macchiavello.
\newblock Optimal universal quantum cloning and state estimation.
\newblock {\em Phys. Rev. Lett.}, 81:2598--2601, Sep 1998.

\bibitem{ham-1}
K.~Hammerer, M.~M. Wolf, E.~S. Polzik, and J.~I. Cirac.
\newblock Quantum benchmark for storage and transmission of coherent states.
\newblock {\em Phys. Rev. Lett.}, 94(15):150503, 2005.

\bibitem{fidelity-bounds}
S.~Massar and S.~Popescu.
\newblock Optimal extraction of information from finite quantum ensembles.
\newblock In {\em Asymptotic Theory Of Quantum Statistical Inference: Selected
  Papers}, pages 356--364. World Scientific, 2005.

\bibitem{hillerybuzek-2}
V.~Bu{\v{z}}ek and M.~Hillery.
\newblock Universal optimal cloning of arbitrary quantum states: from qubits to
  quantum registers.
\newblock {\em Phys. Rev. Lett.}, 81(22):5003, 1998.

\bibitem{dag}
D.~Bru{\ss}, D.~P. DiVincenzo, A.~Ekert, C.~A. Fuchs, C.~Macchiavello, and
  J.~A. Smolin.
\newblock Optimal universal and state-dependent quantum cloning.
\newblock {\em Phys. Rev. A}, 57(4):2368, 1998.

\bibitem{teleexp}
D.~Bouwmeester, J.~W. Pan, K.~Mattle, M.~Eibl, H.~Weinfurter, and A.~Zeilinger.
\newblock Experimental quantum teleportation.
\newblock {\em Nature}, 390:575, 1997.

\bibitem{teleexpdemartini}
D.~Boschi, S.~Branca, F.~De~Martini, L.~Hardy, and S.~Popescu.
\newblock Experimental realization of teleporting an unknown pure quantum state
  via dual classical and {Einstein-Podolsky-Rosen} channels.
\newblock {\em Phys. Rev. Lett.}, 80(6):1121, 1998.

\bibitem{Wang:2015aa}
X.-L. Wang, X.-D. Cai, Z.-E. Su, M.-C. Chen, D.~Wu, L.~Li, N.-L. Liu, C.-Y. Lu,
  and J.-W. Pan.
\newblock Quantum teleportation of multiple degrees of freedom of a single
  photon.
\newblock {\em Nature}, 518:516 EP --, 02 2015.

\bibitem{vaidcvtele}
L.~Vaidman.
\newblock Teleportation of quantum states.
\newblock {\em Phys. Rev. A}, 49(2):1473, 1994.

\bibitem{bkcvtele}
S.~L. Braunstein and H.~J. Kimble.
\newblock Teleportation of continuous quantum variables.
\newblock In {\em Quantum Information with Continuous Variables}, pages 67--75.
  Springer, 1998.

\bibitem{reid1989demonstration}
M.~D. Reid.
\newblock Demonstration of the {E}instein-{P}odolsky-{R}osen paradox using
  nondegenerate parametric amplification.
\newblock {\em Phys. Rev. A}, 40(2):913, 1989.

\bibitem{cvteleexperiments}
A.~Furusawa, J.~L. S{\o}rensen, S.~L. Braunstein, C.~A. Fuchs, H.~J. Kimble,
  and E.~S. Polzik.
\newblock Unconditional quantum teleportation.
\newblock {\em Science}, 282(5389):706--709, 1998.

\bibitem{cvteleexperiments-1}
W.~P. Bowen, N.~Treps, B.~C. Buchler, R.~Schnabel, T.~C. Ralph, H.-A. Bachor,
  T.~Symul, and P.~K. Lam.
\newblock Experimental investigation of continuous-variable quantum
  teleportation.
\newblock {\em Phys. Rev. A}, 67(3):032302, 2003.

\bibitem{cvteleexperiments-2}
T.~C. Zhang, K.W. Goh, C.W. Chou, P.~Lodahl, and H.~J. Kimble.
\newblock Quantum teleportation of light beams.
\newblock {\em Phys. Rev. A}, 67(3):033802, 2003.

\bibitem{cvtelenocloning}
N.~Takei, H.~Yonezawa, T.~Aoki, and A.~Furusawa.
\newblock High-fidelity teleportation beyond the no-cloning limit and
  entanglement swapping for continuous variables.
\newblock {\em Phys. Rev. Lett.}, 94(22):220502, 2005.

\bibitem{clonein}
F.~Grosshans and P.~Grangier.
\newblock Quantum cloning and teleportation criteria for continuous quantum
  variables.
\newblock {\em Phys. Rev. A}, 64(1):010301, 2001.

\bibitem{cerf}
N.~J. Cerf, A.~Ipe, and X.~Rottenberg.
\newblock Cloning of continuous quantum variables.
\newblock {\em Phys. Rev. Lett.}, 85(8):1754, 2000.

\bibitem{mixed-state-tele-1}
F.~Verstraete and H.~Verschelde.
\newblock Optimal teleportation with a mixed state of two qubits.
\newblock {\em Phys. Rev. Lett.}, 90(9):097901, 2003.

\bibitem{epr-steer-tele}
Q.~He, L.~Rosales-Z{\'a}rate, G.~Adesso, and M.~D. Reid.
\newblock Secure continuous variable teleportation and
  {Einstein-Podolsky-Rosen} steering.
\newblock {\em Phys. Rev. Lett.}, 115(18):180502, 2015.

\bibitem{brunner-rmp}
N.~Brunner, D.~Cavalcanti, S.~Pironio, V.~Scarani, and S.~Wehner.
\newblock Bell nonlocality.
\newblock {\em Rev. Mod. Phys.}, 86:419--478, Apr 2014.

\bibitem{hsieh2016}
C.-Y. Hsieh, Y.-C. Liang, and R.-K. Lee.
\newblock Quantum steerability: Characterization, quantification,
  superactivation, and unbounded amplification.
\newblock {\em Phys. Rev. A}, 94:062120, Dec 2016.

\bibitem{Hillery_Review_1984_DistributionFunctions}
M.~Hillery, R.~F. O'Connell, M.~O. Scully, and E.~P. Wigner.
\newblock {Distribution functions in physics: Fundamentals}.
\newblock {\em Phys. Rep.}, 106:121--167, 1984.

\bibitem{Wigner_1932}
E.~Wigner.
\newblock {On the Quantum Correction For Thermodynamic Equilibrium}.
\newblock {\em Phys. Rev.}, 40:749--759, 1932.

\bibitem{Moyal_1949}
J.~E. Moyal.
\newblock Quantum mechanics as a statistical theory.
\newblock {\em Mathematical Proceedings of the Cambridge Philosophical
  Society}, 45(01):99--124, 1949.

\bibitem{lvovsky2009continuous}
A.~I. Lvovsky and M.~G. Raymer.
\newblock Continuous-variable optical quantum-state tomography.
\newblock {\em Rev. Mod. Phys.}, 81(1):299, 2009.

\bibitem{Husimi1940}
K.~Husimi.
\newblock Some formal properties of the density matrix.
\newblock {\em Proc. Phys. Math. Soc. Jpn.}, 22:264--314, 1940.

\bibitem{Glauber_1963_P-Rep}
R.~J. Glauber.
\newblock {Coherent and Incoherent States of the Radiation Field}.
\newblock {\em Phys. Rev.}, 131:2766--2788, 1963.

\bibitem{Drummond1980posp}
P.~D. Drummond and C.~W. Gardiner.
\newblock Generalised {P}-representations in quantum optics.
\newblock {\em J. Phys. A-Math. Theor.}, 13(7):2353, 1980.

\bibitem{Gilchrist1997posp}
A.~Gilchrist, C.~W. Gardiner, and P.~D. Drummond.
\newblock Positive {P} representation: Application and validity.
\newblock {\em Phys. Rev. A}, 55:3014--3032, Apr 1997.

\bibitem{Deuar:2002}
P.~Deuar and P.~D. Drummond.
\newblock {Gauge $P$ representations for quantum-dynamical problems: Removal of
  boundary terms}.
\newblock {\em Phys. Rev. A}, 66:033812, 2002.

\bibitem{drummond1981_II_nonequilibriumparamp}
P.~D. Drummond, K.~J. McNeil, and D.~F. Walls.
\newblock Non-equilibrium transitions in sub/second harmonic generation.
\newblock {\em J. Mod. Opt.}, 28(2):211--225, 1981.

\bibitem{DrummondGardinerWalls1981}
P.~D. Drummond, C.~W. Gardiner, and D.~F. Walls.
\newblock Quasiprobability methods for nonlinear chemical and optical systems.
\newblock {\em Phys. Rev. A}, 24:914--926, 1981.

\bibitem{Drummond_EPL_1993}
P.~D. Drummond and A.~D. Hardman.
\newblock {Simulation of Quantum Effects in Raman-Active Waveguides}.
\newblock {\em Europhys. Lett.}, 21:279--284, 1993.

\bibitem{Kheruntsyan2005BEC}
K.~V. Kheruntsyan, M.~K. Olsen, and P.~D. Drummond.
\newblock {Einstein-Podolsky-Rosen} correlations via dissociation of a
  molecular {B}ose-{E}instein condensate.
\newblock {\em Phys. Rev. Lett.}, 95:150405, Oct 2005.

\bibitem{Opanchuk2012BEC}
B.~Opanchuk, Q.~Y. He, M.~D. Reid, and P.~D. Drummond.
\newblock Dynamical preparation of {Einstein-Podolsky-Rosen} entanglement in
  two-well {B}ose-{E}instein condensates.
\newblock {\em Phys. Rev. A}, 86:023625, Aug 2012.

\bibitem{Opanchuk2013WignerBEC}
B.~Opanchuk and P.~D. Drummond.
\newblock Functional {Wigner} representation of quantum dynamics of
  {Bose-Einstein} condensate.
\newblock {\em J. Math. Phys.}, 54(4):042107, 2013.

\bibitem{Kiesewetter2014opto}
S.~Kiesewetter, Q.~Y. He, P.~D. Drummond, and M.~D. Reid.
\newblock Scalable quantum simulation of pulsed entanglement and
  {Einstein-Podolsky-Rosen} steering in optomechanics.
\newblock {\em Phys. Rev. A}, 90:043805, Oct 2014.

\bibitem{Kiesewetter2017opto}
S.~Kiesewetter, R.~Y. Teh, P.~D. Drummond, and M.~D. Reid.
\newblock Pulsed entanglement of two optomechanical oscillators and furry's
  hypothesis.
\newblock {\em Phys. Rev. Lett.}, 119:023601, Jul 2017.

\bibitem{Teh2017opto}
R.~Y. Teh, S.~Kiesewetter, M.~D. Reid, and P.~D. Drummond.
\newblock Simulation of an optomechanical quantum memory in the nonlinear
  regime.
\newblock {\em Phys. Rev. A}, 96:013854, Jul 2017.

\bibitem{Arecchi_SUN}
F.~T. Arecchi, Eric Courtens, Robert Gilmore, and Harry Thomas.
\newblock {Atomic Coherent States in Quantum Optics}.
\newblock {\em Phys. Rev. A}, 6:2211--2237, 1972.

\bibitem{Agarwal:1981}
G.~S. Agarwal.
\newblock Relation between atomic coherent-state representation, state
  multipoles, and generalized phase-space distributions.
\newblock {\em Phys. Rev. A}, 24:2889--2896, 1981.

\bibitem{Barry_PD_qubit_SU}
D.~W. Barry and P.~D. Drummond.
\newblock {Qubit phase space: $\mathrm{SU}(n)$ coherent-state {$P$}
  representations}.
\newblock {\em Phys. Rev. A}, 78:052108, 2008.

\bibitem{corney2006gaussian}
J.~F. Corney and P.~D. Drummond.
\newblock Gaussian phase-space representations for fermions.
\newblock {\em Phys. Rev. B}, 73(12):125112, 2006.

\bibitem{Cahill1969}
K.~E. Cahill and R.~J. Glauber.
\newblock Density operators and quasiprobability distributions.
\newblock {\em Phys. Rev. A}, 177:1882--1902, Jan 1969.

\bibitem{lobino2008complete}
M.~Lobino, D.~Korystov, C.~Kupchak, E.~Figueroa, B.~C. Sanders, and A.~I.
  Lvovsky.
\newblock Complete characterization of quantum-optical processes.
\newblock {\em Science}, 322(5901):563--566, 2008.

\bibitem{Rosales-Zarate2011entropy}
L.~E.~C. Rosales-Z\'arate and P.~D. Drummond.
\newblock Linear entropy in quantum phase space.
\newblock {\em Phys. Rev. A}, 84:042114, Oct 2011.

\bibitem{corney2003gaussian}
J.~F. Corney and P.~D. Drummond.
\newblock Gaussian quantum operator representation for bosons.
\newblock {\em Phys. Rev. A}, 68(6):063822, 2003.

\bibitem{corney2005gaussian}
J.~F. Corney and P.~D. Drummond.
\newblock Gaussian operator bases for correlated fermions.
\newblock {\em J. Phys. A-Math. Gen.}, 39(2):269, 2005.

\bibitem{joseph2018phase}
R.~R. Joseph, L.~E.~C. Rosales-Z{\'a}rate, and P.~D. Drummond.
\newblock Phase space methods for majorana fermions.
\newblock {\em J. Phys. A-Math. Theor.}, 51(24):245302, 2018.

\bibitem{divincenzo2000physical}
D.~P. DiVincenzo.
\newblock The physical implementation of quantum computation.
\newblock {\em Fortschr. Phys.}, 48(9-11):771--783, 2000.

\bibitem{longdell2004experimental}
J.~J. Longdell and M.~J. Sellars.
\newblock Experimental demonstration of quantum-state tomography and
  qubit-qubit interactions for rare-earth-metal-ion-based solid-state qubits.
\newblock {\em Phys. Rev. A}, 69(3):032307, 2004.

\bibitem{ospelkaus2011microwave}
C.~Ospelkaus, U.~Warring, Y.~Colombe, K.R. Brown, J.M. Amini, D.~Leibfried, and
  D.J. Wineland.
\newblock Microwave quantum logic gates for trapped ions.
\newblock {\em Nature}, 476(7359):181, 2011.

\bibitem{fuchs2011quantum}
G.~D. Fuchs, G.~Burkard, P.~V. Klimov, and D.~D. Awschalom.
\newblock A quantum memory intrinsic to single nitrogen-vacancy centres in
  diamond.
\newblock {\em Nat. Phys.}, 7(10):789--793, 2011.

\bibitem{CHANELIERE201877}
T.~Chaneli{\`e}re, G.~H{\'e}tet, and N.~Sangouard.
\newblock Chapter two - quantum optical memory protocols in atomic ensembles.
\newblock volume~67 of {\em Advances In Atomic, Molecular, and Optical
  Physics}, pages 77 -- 150. Academic Press, 2018.

\bibitem{HeReidPhysRevA.79.022310}
Q.~Y. He, M.~D. Reid, E.~Giacobino, J.~Cviklinski, and P.~D. Drummond.
\newblock Dynamical oscillator-cavity model for quantum memories.
\newblock {\em Phys. Rev. A}, 79:022310, Feb 2009.

\bibitem{Bartkiewicz:PRA:2013}
K.~Bartkiewicz, K.~Lemr, and A.~Miranowicz.
\newblock Direct method for measuring of purity, superfidelity, and subfidelity
  of photonic two-qubit mixed states.
\newblock {\em Phys. Rev. A}, 88:052104, Nov 2013.

\bibitem{nicolas2014quantum}
A.~Nicolas, L.~Veissier, L.~Giner, E.~Giacobino, D.~Maxein, and J.~Laurat.
\newblock A quantum memory for orbital angular momentum photonic qubits.
\newblock {\em Nat. Photonics}, 8(3):234, 2014.

\bibitem{lvovsky2001quantum}
A.~I. Lvovsky, H.~Hansen, T.~Aichele, O.~Benson, J.~Mlynek, and S.~Schiller.
\newblock Quantum state reconstruction of the single-photon {Fock} state.
\newblock {\em Phys. Rev. Lett.}, 87(5):050402, 2001.

\bibitem{MassarPopPhysRevLett.74.1259}
S.~Massar and S.~Popescu.
\newblock Optimal extraction of information from finite quantum ensembles.
\newblock {\em Phys. Rev. Lett.}, 74:1259--1263, Feb 1995.

\bibitem{Choi:1975}
M.~D. Choi.
\newblock Completely positive linear maps on complex matrices.
\newblock {\em Linear Algebr. Appl.}, 10:285, 1975.

\bibitem{jamiolkowski1974effective}
A.~Jamio{\l}kowski.
\newblock An effective method of investigation of positive maps on the set of
  positive definite operators.
\newblock {\em Rep. Math. Phys.}, 5(3):415--424, 1974.

\bibitem{o2004quantum}
J.~L. O'Brien, G.~J. Pryde, A.~Gilchrist, D.~F.~V. James, N.~K. Langford, T.~C.
  Ralph, and A.G. White.
\newblock Quantum process tomography of a controlled-not gate.
\newblock {\em Phys. Rev. Lett.}, 93(8):080502, 2004.

\bibitem{benhelm2008towards}
J.~Benhelm, G.~Kirchmair, C.~F. Roos, and R.~Blatt.
\newblock Towards fault-tolerant quantum computing with trapped ions.
\newblock {\em Nat. Phys.}, 4(6):463, 2008.

\bibitem{ryan2009randomized}
C.~A. Ryan, M.~Laforest, and R.~Laflamme.
\newblock Randomized benchmarking of single-and multi-qubit control in
  liquid-state {NMR} quantum information processing.
\newblock {\em New. J. Phys.}, 11(1):013034, 2009.

\bibitem{veldhorst2015two}
M.~Veldhorst, C.H. Yang, J.C.C. Hwang, W.~Huang, J.P. Dehollain, J.T. Muhonen,
  S.~Simmons, A.~Laucht, F.E. Hudson, K.~M. Itoh, et~al.
\newblock A two-qubit logic gate in silicon.
\newblock {\em Nature}, 526(7573):410, 2015.

\bibitem{lucero2008high}
E.~Lucero, M.~Hofheinz, M.~Ansmann, R.~C. Bialczak, N.~Katz, M.~Neeley, A.~D.
  O'Connell, H.~Wang, A.~N. Cleland, and J.~M. Martinis.
\newblock High-fidelity gates in a single {J}osephson qubit.
\newblock {\em Phys. Rev. Lett.}, 100:247001, Jun 2008.

\bibitem{Schoenberg1938}
I.~J. Schoenberg.
\newblock Metric spaces and positive definite functions.
\newblock {\em Trans. Amer. Math. Soc.}, 44:522--536, 1938.

\bibitem{ma2011}
Z.-H. Ma and J.-L. Chen.
\newblock Metrics of quantum states.
\newblock {\em J. Phys. A-Math. Theor.}, 44(19):195303, 2011.

\end{thebibliography}

} 
\end{document}